\documentclass[11pt,a4paper]{amsart}

\usepackage[foot]{amsaddr}
\usepackage{ifxetex}
\ifxetex
  \usepackage[no-math]{fontspec}
\else
\fi

\newcommand{\ifarxiv}[2]{#2}
\renewcommand{\ifarxiv}[2]{#1}

\usepackage{amsmath}
\usepackage{amsfonts}
\usepackage{amssymb}
\usepackage{amsthm}
\usepackage{verbatim}
\usepackage{dsfont}
\usepackage{fullpage}
\usepackage{microtype}
\usepackage[libertine]{newtxmath}
\usepackage[tt=false]{libertine} 
\usepackage{caption}
\usepackage{bbm}
\usepackage{hyperref, color}
\hypersetup{colorlinks=true,citecolor=blue, linkcolor=blue, urlcolor=blue}
\usepackage[linesnumbered,boxed,ruled,vlined]{algorithm2e}
\usepackage{algorithmicx}
\usepackage{bm}
\usepackage{bbm}
\usepackage[numbers]{natbib}
\usepackage{xcolor}
\usepackage{enumerate} 
\usepackage{enumitem}
\usepackage{tabularx}
\usepackage{array}
\usepackage{cleveref}
\usepackage{pifont}
\usepackage{xifthen}
\usepackage{xstring}

\newcolumntype{L}[1]{>{\raggedright\arraybackslash}p{#1}}
\newcolumntype{C}[1]{>{\centering\arraybackslash}m{#1}}
\newcolumntype{R}[1]{>{\raggedleft\arraybackslash}p{#1}}

\usepackage{makecell}
\usepackage{footnote}
\makesavenoteenv{tabular}

\newcommand{\TODO}[1]{\typeout{TODO: \the\inputlineno: #1}\textbf{{\color{red}[[[ #1 ]]]}}}
\newcommand{\HKTODO}[1]{\typeout{TODO: \the\inputlineno: #1}\textbf{{\color{blue}[[[ #1 ]]]}}}

\newcommand{\MSC}[1]{\mathscr{{#1}}}

\newcommand{\pth}{\textnormal{\textsf{Path}}}

\newcommand{\recalc}{\textnormal{\textsf{RecursiveApproximator}}}

\newcommand{\calc}{\textnormal{\textsf{MarginalApproximator}}}

\newcommand{\texe}{T_\textnormal{\textsf{Enu}}}

\newcommand{\tma}{T_\textnormal{\textsf{MA}}}
\newcommand{\tra}{T_\textnormal{\textsf{RA}}}

\newcommand{\tcmain}{T_\textnormal{\textsf{Count}}}

\def\Pr{\mathop{\mathbf{Pr}}\nolimits}

\newcommand{\Mod}[3]{{#1}_{{#2}\gets{#3}}}
\newcommand{\vbl}{\mathsf{vbl}}
\newcommand{\True}{\mathtt{True}}
\newcommand{\False}{\mathtt{False}}

\newcommand{\abs}[1]{\left\vert#1\right\vert}

\newcommand{\vstar}[1]{V^{#1}_{\star}}

\SetKwRepeat{Do}{do}{while}

\newtheorem{theorem}{Theorem}[section]

\newtheorem*{claim*}{Claim}
\newtheorem{condition}[theorem]{Condition}

\newtheorem{lemma}[theorem]{Lemma}
\newtheorem{proposition}[theorem]{Proposition}
\newtheorem{corollary}[theorem]{Corollary}
\theoremstyle{definition}

\newtheorem{definition}[theorem]{Definition}
\newtheorem{remark}[theorem]{Remark}
\newtheorem*{remark*}{Remark}
\newtheorem{question}[theorem]{Question}

\renewcommand{\Pr}[2][]{ \ifthenelse{\isempty{#1}}
  {\mathop{\mathbf{Pr}}\left[#2\right]} {\mathop{\mathbf{Pr}}_{#1}\left[#2\right]} }
\newcommand{\E}[1]{\mathbb{E}\left[{#1}\right]}
\newcommand{\Var}[2][]{ \ifthenelse{\isempty{#1}}
  {\mathbf{\mathbf{Var}}\left[#2\right]}
  {\mathbf{\mathbf{Var}}_{#1}\left[#2\right]} }
\newcommand{\var}[1]{{{\vbl}}\left({#1}\right)}  
\newcommand{\poly}{{\rm poly}}  
\newcommand{\nextvar}[1]{{{\mathsf{NextVar}}}\left({#1}\right)} 
\newcommand{\Lin}[1]{{{\mathsf{Lin}}}\left({#1}\right)} 

\newcommand{\cfrozen}[1]{\+{C}^{#1}_{\mathsf{frozen}}}

\newcommand{\vfix}[1]{V^{#1}_{\mathsf{fix}}}

\newcommand{\hfix}[1]{H^{#1}_{\mathsf{fix}}}
\newcommand{\vinf}[1]{V^{#1}_{\star{\mathsf{\text{-}inf}}}}

\newcommand{\ccon}[1]{\+{C}^{#1}_{\star{\mathsf{\text{-}con}}}}
\newcommand{\vcon}[1]{V^{#1}_{\star{\mathsf{\text{-}con}}}}

\newcommand{\vst}[1]{V^{#1}_{\star}}

\newcommand{\csfrozen}[1]{\+{C}^{#1}_{\star{\mathsf{\text{-}frozen}}}}

\newcommand{\qs}{\+{Q}^*}

\newcommand{\tv}{d_{\rm TV}}
\newcommand{\gvc}{G_{\mathsf{VC}}}

\newcommand{\one}[1]{\mathbbm{1}\left[#1\right]}
\def\^#1{\mathbb{#1}} 
\def\*#1{\mathbf{#1}} 
\def\+#1{\mathcal{#1}} 
\def\-#1{\mathrm{#1}} 
\def\=#1{\boldsymbol{#1}} 
\newcommand{\set}[1]{\left\{#1\right\}}

\newcommand{\defeq}{\triangleq}

\DeclareMathAlphabet{\mathcal}{OMS}{cmsy}{m}{n}

\newcommand{\induceddist}[2]{\gamma^{{#1}}_{#2}}

\newcommand{\qus}[1]{\+{Q}^{\star}_{#1}}
\newcommand{\pprime}{\alpha}

\usepackage{amssymb}
\usepackage{stackengine}
\usepackage{scalerel}
\usepackage{xcolor}
\usepackage{graphicx}
\newcommand\openbigstar[1][0.7]{%
  \scalerel*{%
    \stackinset{c}{-.125pt}{c}{}{\scalebox{#1}{\color{white}{$\bigstar$}}}{%
      $\bigstar$}%
  }{\bigstar}
}
\newcommand{\hollowstar}{\text{$\scriptstyle\openbigstar[.7]$}}

\title{Deterministic counting Lov\'{a}sz local lemma \\ beyond linear programming}

\ifarxiv{\author{
Kun He
}
\author{
Chunyang Wang
}
\author{
Yitong Yin
}
\address[Kun He]{Institute of Computing Technology, Chinese Academy of Sciences, No.6 Kexueyuan South Road Zhongguancun, Haidian District, Beijing, China. \textnormal{E-mail: \url{hekun@ict.ac.cn}}.
The research of K.\ He is supported by the Strategic Priority Research
Program of Chinese Academy of Sciences under Grant
No. XDA27000000, the National Natural Science Foundation
of China Grants No. 62002231, 61832003.}
\address[Chunyang Wang, Yitong Yin]{ State Key Laboratory for Novel Software Technology, Nanjing University, 163 Xianlin Avenue, Nanjing, Jiangsu Province, China. \textnormal{E-mails: \url{wcysai@smail.nju.edu.cn}, \url{yinyt@nju.edu.cn} }}}
{

\author{Anonymous Authors}

}
\date{}
\begin{document}


\allowdisplaybreaks
\maketitle

\begin{abstract}
We give a simple combinatorial algorithm to deterministically approximately count the number of satisfying assignments of general constraint satisfaction problems (CSPs). 
Suppose that the CSP has domain size $q=O(1)$, 
each constraint contains at most $k=O(1)$ variables, shares variables with at most $\Delta=O(1)$ constraints, and is violated with probability at most $p$ by a uniform random assignment.
The algorithm returns in polynomial time in an improved local lemma regime:
\[
q^2\cdot k\cdot p\cdot\Delta^5\le C_0\quad\text{for a suitably small absolute constant }C_0.
\]
%
Here the key term $\Delta^5$ improves the previously best known $\Delta^7$ for general CSPs~\cite{Vishesh21towards} and $\Delta^{5.714}$ for the special case of $k$-CNF~\cite{Vishesh21sampling, HSW21}.

Our deterministic counting algorithm is a derandomization of the very recent fast sampling algorithm in~\cite{he2022sampling}.
It departs substantially from all previous deterministic counting Lov\'{a}sz local lemma algorithms which relied on linear programming,
and gives a deterministic approximate counting algorithm that straightforwardly derandomizes a fast sampling algorithm,
hence unifying the fast sampling and deterministic approximate counting in the same algorithmic framework.


To obtain the improved regime, in our analysis we develop a refinement of the $\{2,3\}$-trees that were used in the previous analyses of counting/sampling LLL.
Similar techniques can be applied to the previous LP-based algorithms to obtain the same improved regime 
and may be of independent interests.
\end{abstract}


  
\section{Introduction}\label{sec:intro}
Approximate counting and almost uniform sampling are two intimately related classes of computational problems that have been extensively studied in theoretical computer science. 
It was well-known that randomized approximate counting can be achieved by almost uniform sampling through the generic approaches of self-reduction~\cite{jerrum1986random} or annealing~\cite{Dyer1991random,stefankovic2009adpative}. 
%

On the other hand,
\emph{deterministic} approximate counting algorithms use different approaches such as decay of correlation~\cite{weitz06counting}, zero-freeness~\cite{barvinok2016combinatorics,patel2017deterministic}, and cluster-expansion~\cite{helmuth2020algorithmic,jenssen2020algorithms}, or in the case of counting constraint satisfaction solutions, the linear programming~\cite{Moi19,guo2019counting,Vishesh21towards}.
All these deterministic approximate counting methods have running times where the exponent over the input size depends on additional parameters such as degree of the underlying graph. 
And more fundamentally, all these deterministic counting algorithm work in quite different algorithmic frameworks that deviate far from those of the fast sampling algorithms where the exponents of the running times are universal constants.
There is one exception very recently~\cite{vishesh2022approximate}, where for matchings/independent sets with a given size,
a \emph{unified} algorithm based on a new technique called local central limit theorems was found to simultaneously resolve deterministic counting and fast randomized sampling within the same algorithmic framework.

We are focused on the problem of counting general constraint satisfaction solutions. 
Our goal is to give a unified approach for deterministic counting Lov\'{a}sz Local Lemma (LLL)~\cite{Moi19,guo2019counting,Vishesh21towards} and fast sampling LLL~\cite{GJL19,FGYZ20,feng2021sampling,Vishesh21sampling,HSW21,he2022sampling}.

\vspace{8pt}
\textbf{CSPs and Lov\'{a}sz Local Lemma.} 
An instance of constraint satisfaction problem (CSP), 
called a \emph{CSP formula}, denoted by $\Phi=(V,\+{Q},\+{C})$, is defined as follows: $V$ is a set of $n=|V|$ variables;
$\+{Q}\triangleq\bigotimes_{v\in V}Q_v$ is a product space of all assignments of variables, 
where each $Q_v$ is a finite domain of size $q_v\triangleq\abs{Q_v}\ge 2$ over where the variable $v$ ranges; 
and $\+{C}$ is a collection of local constraints
where each $c\in \+{C}$ is a constraint function $c:\bigotimes_{v\in \vbl(c)}Q_v\to\{\True,\False\}$ defined on a subset of variables, denoted by $\vbl(c)\subseteq V$.
An assignment $\={x}\in \+Q$ is called \emph{satisfying} for $\Phi$ if 
\[
\Phi(\={x})\triangleq\bigwedge\limits_{c\in\+{C}} c\left(\={x}_{\vbl(c)}\right)=\True.
\]

Some key parameters of a CSP formula $\Phi=(V,\+{Q},\+{C})$ are listed in the following:
\begin{itemize}
    \item \emph{domain size} $q=q_\Phi\triangleq\max\limits_{v\in V}\abs{Q_{v}}$  
    and \emph{width} $k=k_\Phi\triangleq\max\limits_{e\in \+{C}}\abs{ {\vbl}(c)}$;
    \item \emph{constraint degree} $\Delta=\Delta_\Phi\triangleq\max\limits_{c\in \+{C}}\abs{\{c'\in \+{C}\mid \vbl(c)\cap \vbl(c')\neq\emptyset\}}$;
    \item \emph{violation probability} $p=p_{\Phi}\triangleq\max\limits_{c\in \+{C}}\mathbb{P}[\neg c]$, where $\mathbb{P}$ denotes the law for the uniform assignment, 
    in which each $v\in V$ draws its evaluation from~$Q_v$ uniformly and independently at random.
\end{itemize}

A characterization for the existence of a satisfying solution to CSP is given by the celebrated \emph{Lov\'{a}sz Local Lemma (LLL)}~\cite{LocalLemma}. By interpreting the space of all possible assignments as a
probability space and the violation of each constraint as a bad event, the local lemma provides a
sufficient condition
\begin{align}\label{eq:classic-LLL-condition}
    \mathrm{e}p\Delta\le 1.
\end{align}
for the existence of an assignment to avoid all the bad events, i.e., the existence of a solution to the CSP.

\vspace{8pt}
\textbf{Counting/Sampling LLL.} A counting/sampling variant of the Lov\'{a}sz Local Lemma, which seeks algorithms to efficiently (approximate) count and sample (almost-uniform) solutions to CSPs in the local lemma regime, has drawn lots of recent attention~\cite{GJL19,Moi19,guo2019counting,galanis2019counting,FGYZ20,feng2021sampling,Vishesh21sampling,Vishesh21towards,HSW21,galanis2021inapproximability,feng2022improved,he2022sampling,qiu2022perfect}. 
There are two separate lines of work on {deterministic} counting LLL and {fast} sampling, using very different approaches.

To this date, all existing deterministic counting algorithms for LLL are based on linear programming.
The algorithm was first found in a major breakthrough~\cite{Moi19}. 
The algorithm properly marked the variables using algorithmic LLL and then constructed a polynomial-time deterministic oracle for approximately computing the marginal probabilities of marked variables via linear programs of sizes $n^{\poly(\Delta,k)}$,
which can be used to deterministically approximately count the number of satisfying solutions to $k$-CNF formulas in $n^{\poly(\Delta,k)}$  time when $p\Delta^{60}\lesssim 1$. 
This LP-based approach was later extended to work for hypergraph colorings~\cite{guo2019counting} and random CNF formulas~\cite{galanis2019counting} and finally, for general CSP instances with a substantially improved LLL regime of $p\Delta^7\lesssim 1$~\cite{Vishesh21towards}. 

Another line of work for the counting/sampling local lemma focuses on \emph{fast} sampling an almost-uniform satisfying solution. In~\cite{FGYZ20}, an algorithm was given for approximate sampling uniform solutions to $k$-CNF formulas with a near-linear running time $\widetilde{O}(n^{1.001})$ when $p\Delta^{20}\lesssim 1$.  Their approach was based on a Markov chain on a projected space constructed using the mark/unmark strategy invented in~\cite{Moi19}. This projected Markov chain approach was later refined in~\cite{feng2021sampling, Vishesh21sampling, HSW21} for fast sampling nearly-atomic CSP solutions, where by atomic we mean each constraint is violated by one forbidden configuration, which achieved the state-of-the-arts regime $p\Delta^{5.714}\lesssim 1$. Very recently in~\cite{he2022sampling},  a new approach based on the recursive marginal sampler in~\cite{anand2021perfect} was given for sampling general CSP solutions in the local lemma regime within near-linear time. This new sampling algorithm was very different from all Markov chain based sampling algorithms.

\subsection{Our results}
We give a new deterministic algorithm for approximately counting the number of satisfying solutions for general CSPs in an improved local lemma regime.
This new deterministic approximate counting algorithm is a combinatorial one, which does not rely on linear programming, and hence is considerably simpler and more intuitive than all previous deterministic algorithms for counting LLL~\cite{Moi19,guo2019counting,galanis2019counting,Vishesh21towards} that were LP-based.

This new algorithm is in fact a derandomization of the very recent fast sampling algorithm in~\cite{he2022sampling}.
Furthermore, 
we obtain an improved regime with a much refined analysis, as stated in the following theorem.


\begin{theorem}[informal]\label{thm:main-counting}
There is an algorithm such that given as input any $\varepsilon\in(0,1)$ and any CSP formula $\Phi=(V,\+{Q},\+{C})$ with $n$ variables  satisfying 
\begin{align}\label{eq:main-thm-LLL-condition}
q^2\cdot k\cdot p\cdot\Delta^5\leq \frac{1}{256\mathrm{e}^3},
\end{align}
The algorithm terminates within $\left(\frac{n}{\varepsilon}\right)^{\poly(\log{q},\Delta,k)}$ time and deterministically outputs an $\varepsilon$-approximation of $Z$, the number of satisfying solutions to $\Phi$.
\end{theorem}

This improves the current state-of-the-arts $p\Delta^{7}\lesssim 1$ for general CSP~\cite{Vishesh21towards}, and $p\Delta^{5.714}\lesssim 1$ for nearly-atomic CSP~\cite{HSW21} including $k$-SAT.
The  $O\left(n^{\poly(\log{q},\Delta,k)}\right)$ running time aligns with previous LP-based algorithms.
The formal statement of the theorem is in Theorem \ref{thm:main-counting-refined}.

We also show that our analysis can be used to improve the bound of the algorithm in~\cite{Vishesh21towards} to the same regime stated as in Theorem \ref{thm:main-counting}. This is described in \Cref{sec:JPV}.

\subsection{Technique Overview}
Our method follows the adaptive mark/unmark framework of counting LLL~\cite{guo2019counting,Vishesh21towards}.
We briefly describe the previous approach before introducing our modifications. 

Given a CSP instance $\Phi$,
it has been observed by~\cite{haeupler2011new} that the marginal distribution of every variable is close to uniform within a local lemma regime. This is referred to as the ``local uniformity" property.

In previous works of counting LLL~\cite{Moi19,guo2019counting,galanis2019counting,Vishesh21towards}, a key ingredient is a marginal approximator, which approximates the marginal distribution of some variable conditioning on the current partial assignment. 
This approximator was built on a novel coupling procedure, first proposed by Moitra~\cite{Moi19}.

In the procedure, two copies of the Gibbs distribution (which in our context is the uniform distribution over all satisfying assignments) conditioning on partial assignments where only one variable is assigned different values are maximally coupled in a sequential and variable-wise fashion. 
In addition, the variables are picked in a manner such that all the variables assigned in the coupling procedure have the local uniformity property. Initially presented as mark/unmark framework by Moitra~\cite{Moi19}, the rule for picking variables was later refined to become adaptive~\cite{guo2019counting,Vishesh21towards}. An observation is that after sufficiently many steps of the idealized coupling procedure,
there is a good chance that the component containing $v$ in 
the resulting formula is of logarithmic size, 
from where one can efficiently calculate the ratio of the number of satisfying assignments extending two partial assignments using exhaustive enumeration. This observation allows one to truncate the procedure up to some certain threshold so that there remains a large probability that the two distributions are successfully coupled.

Then a linear program is set up to mimic the transition probabilities in the (truncated) coupling procedure, so one can use a binary search to approximate the marginal distribution of $v$. 
The coupling procedure and the linear program are employed for marginal approximating by all the algorithms of counting LLL~\cite{Moi19,guo2019counting,galanis2019counting,Vishesh21towards}. It is worth noting that this linear program is of size $n^{\poly(k,\Delta,\log{q})}$ and requires a polynomial-time algorithm for solving linear programs to achieve the desired running time.

In this paper, we propose a new combinatorial approximator for approximating the marginal distribution. 
Rather than dealing with the coupling of two Gibbs distributions, we decompose a \emph{single} Gibbs distribution directly. 
Given a variable $v$ with domain $Q_v$,
if $v$ satisfies the local uniformity property, 
there exists $\theta_v< \frac{1}{q_v}$ close enough to $\frac{1}{q_v}$, such that for each $i\in Q_v$ the probability that $v$ is set as $i$ is no less than $\theta_v$.
Thus, there are $q_v+1$ branches for the possibilities of $v$:
for each $i\in Q_v$, there is a branch of assignment $i$ with probability $\theta_v$,
and the last branch is with probability $1- q_v\theta_v$ and its assignment follows a ``overflow" distribution $\+{D}_v$.
For the last branch, we repeatedly find a variable $u$ whose successful pinning might help factorize the formula with respect to $v$, and calculate the marginal by recursively applying the marginal approximator using the chain rule. During the process, a similar rule in the adaptive mark/unmark framework by \cite{Vishesh21towards} is taken to guarantee that the local uniformity property persists throughout the algorithm for each chosen variable. A similar observation as in the LP approach that, with some appropriately chosen truncation condition, under a large fraction of the partial assignments generated from the recursive procedure, the component containing $v$ in 
the resulting formula is of logarithmic size, from where one can efficiently calculate $\+{D}_v$ using exhaustive enumeration. 


Our marginal approximator is a derandomization of the marginal sampler of the recent sampling algorithm in~\cite{he2022sampling}.
Given a variable $v$, the marginal sampler
samples an assignment of $v$ from its marginal distribution, while our marginal approximator calculates the marginal probability that $v$ is assigned as $i$ for each $i\in Q_v$. 
Moreover, equipped with the marginal approximator,
we use the same method as in~\cite{Vishesh21towards} to find a ``guiding assignment", which can be viewed as a method of conditional expectation for derandomization, to achieve a complete algorithm for estimating the number of satisfying solutions.

To bound the error and running time of our marginal approximator,
we design a new combinatorial structure named generalized $\{2,3\}$-tree, 
which leads to the improved bound $p\Delta^5\lesssim 1$.
In most works on counting/sampling LLL, 
two types of bad events are considered: 
one is that the assignment of a marked variable does not fall into the zone of local uniformity;
the other is that a constraint is still not satisfied after that a large proportion of its variables are assigned~\cite{Moi19,feng2021sampling,guo2019counting,Vishesh21towards,he2022sampling}.
In previous work,
these two bad events are treated similarly and bounded using a combinatorial structure named $\{2,3\}$-tree~\cite{alon1991parallel}.
A crucial observation is that the densities of these two types of bad events are different, which inspires our design of this new combinatorial structure to take advantage of this property and push the bounds beyond state-of-the-arts. We remark that the generalized $\{2,3\}$-tree can also be applied to improve the bounds in~\cite{Vishesh21towards}, and may be of independent interests.

\section{Notation and preliminaries}\label{sec:prelim}


\subsection{Notations for CSP}
Recall the definition of CSP formula $\Phi=(V,\+{Q},\+{C})$ in \Cref{sec:intro}.
We further define the following notations.
Let $\Omega=\Omega_{\Phi}$ be the set of all satisfying assignments of $\Phi$, $Z= Z_{\Phi}$ be the size of $\Omega$, and $\mu=\mu_{\Phi}$ be the uniform distribution over $\Omega$. 
Recall that $\mathbb{P}$ denotes the law for the uniform product distribution over $\+{Q}$.
For $C\subseteq\+{C}$, denote $\vbl(C)\triangleq\bigcup_{c\in C}\vbl(c)$;
and for $\Lambda\subseteq V$, denote ${\+Q}_\Lambda\triangleq\bigotimes_{v\in \Lambda}Q_v$.

For each $v\in V$, we use an extra symbol $\hollowstar\not \in Q_v$ to denote that $v$ is \emph{unassigned} with any value.
Given a CSP formula $\Phi=(V,\+{Q},\+{C})$ and a partial assignment $\sigma \in \bigotimes_{v\in V}\left(Q_v\cup \{\hollowstar\}\right)$,
let $\Lambda(\sigma)$ denote $\{v\in V:\sigma(v) \in Q_v\}$.
The simplification of $\Phi$ under $\sigma$ is a new CSP formula $\Phi^\sigma=(V^\sigma,\+{Q}^\sigma,\+{C}^\sigma)$,
where $V^\sigma=V\setminus \Lambda(\sigma)$, $\+{Q}^\sigma=\+{Q}_{V\setminus \Lambda(\sigma)}$, 
and the $\+{C}^\sigma$ is obtained from $\+{C}$ by: 
\begin{enumerate}
\item
removing all the constraints that have already been satisfied
 by $\sigma$;
\item
for  the remaining constraints, 
replacing the variables $v\in \Lambda(\sigma)$ with their values $\sigma(v)$.
\end{enumerate}
It is easy to see that the $\mu_{\Phi^\sigma}$ is the same as the marginal distribution induced by $\mu$ on $V\setminus \Lambda(\sigma)$, conditional on the assignment over $\Lambda(\sigma)$ is $\sigma$.


A CSP formula $\Phi = (V, \+{Q}, \+{C})$ can be naturally represented as a (multi-)hypergraph $H_{\Phi}$,
where each variable $v\in V$ corresponds to a vertex in $H_{\Phi}$ and each constraint $c\in\+{C}$ corresponds to a hyperedge in $H_{\Phi}$ which joins the vertices corresponding to $\vbl(c)$.
We slightly abuse the notation and write $H_{\Phi}=(V, \+{C})$.

%
Let $H_{i}=(V_i, \+{C}_i)$ for $1\le i\le K$ denote all $K\ge 1$ connected components in $H_{\Phi}$, and $\Phi_i=(V_i, \+{Q}_{V_i}, \+{C}_i)$ their formulas.
Obviously $\Phi=\Phi_1\land\Phi_2\land \cdots \land \Phi_K$ with disjoint $\Phi_i$, and 
$Z_{\Phi}$ is the product of $Z_{\Phi_i}$.

\subsection{Lov\'{a}sz Local Lemma}
In the context of CSP, the celebrated asymmetric Lov\'{a}sz Local Lemma is as follows.

\begin{theorem}[Erd\"{o}s and Lov\'{a}sz~\cite{LocalLemma}]\label{locallemma}
    Given a CSP formula $\Phi=(V,\+{Q},\+{C})$, if the following holds
    \begin{align}\label{llleq}
    \exists x\in (0, 1)^\+{C}\quad \text{ s.t.}\quad \forall c \in \+{C}:\quad
        {\mathbb{P}[\neg c]\leq x(c)\prod_{\substack{c'\in \+{C}\\ {\vbl}(c)\cap {\vbl}(c')\neq \emptyset}}(1-x(c'))},
    \end{align}
    then  
    $$
        {\mathbb{P}\left[ \bigwedge\limits_{c\in \+{C}} c\right]\geq \prod\limits_{c\in \+{C}}(1-x(c))>0},
    $$
\end{theorem}

The following result shows that 
when the condition \eqref{llleq} is satisfied, 
the probability of any event in the uniform distribution $\mu$ over all satisfying assignments can be well approximated by the probability of the event in the product distribution.
This was observed in~\cite{haeupler2011new}:

\begin{theorem}[Haeupler, Saha, and Srinivasan~\cite{haeupler2011new}]\label{HSS}
Given a CSP formula $\Phi=(V,\+{Q},\+{C})$, if $\eqref{llleq}$ holds, 
then for any event $A$ that is determined by the assignment on a subset of variables $\var{A}\subseteq V$, 
\[
    \Pr[\mu]{A}=\mathbb{P}\left[A\mid \bigwedge\limits_{c\in \+{C}}  c\right]\leq \mathbb{P}[A]\prod_{\substack{c\in \+{C}\\ \var{c}\cap \var{A}\neq \emptyset}}(1-x(c))^{-1},
\]
where $\mu$ denotes the uniform distribution over all satisfying assignments of $\Phi$ and $\mathbb{P}$ denotes the law of the uniform product distribution over $\+{Q}$.
\end{theorem}

By setting $x(c)=\mathrm{e}p$ for every $c\in \+{C}$ in Theorem \ref{HSS},
it is straightforward to prove the following
``local uniformity" property,
where the lower bound is calculated by $\mu_v(x)=1-\sum_{y\in Q_v\setminus\{x\}}\mu_v(y)$ for each $x\in Q_v$.


\begin{corollary}[local uniformity]\label{generaluniformity}
Given a CSP formula $\Phi=(V,\+{Q},\+{C})$, if $\mathrm{e}p\Delta<1$, then for any variable $v\in V$ and any value $x\in Q_v$, it holds that
\[
\frac{1}{q_v}-\left((1-\mathrm{e}p)^{-\Delta}-1\right)\leq \mu_v(x) \leq \frac{1}{q_v}+\left((1-\mathrm{e}p)^{-\Delta}-1\right).
\]
\end{corollary}

\section{The counting algorithm}\label{sec:counting}
We now present our algorithm for deterministically approximately counting CSP solutions.
\subsection{The main counting algorithm}\label{sec:main-counting-alg}
The main counting algorithm takes as input a CSP formula $\Phi=(V,\+{Q},\+{C})$
with domain size $q=q_\Phi$, 
width $k=k_\Phi$, 
constraint degree  $\Delta=\Delta_\Phi$, 
and violation probability $p=p_{\Phi}$, where the meaning of these parameters are as defined in \Cref{sec:intro}. 

We assume that the $n=|V|$ variables are enumerated as $V=\{v_1,v_2,\ldots,v_n\}$ in an arbitrary order.
For the CSP formula  $\Phi=(V,\+{Q},\+{C})$ presented to the algorithm, we assume that given any constraint $c\in \+{C}$ (or any variable $v\in V$), 
the $\vbl(c)$ (or $\{c\in\+{C}\mid v\in\vbl(c)\}$) can be retrieved in $\poly(k,\Delta)$ time,
and furthermore, given any assignment $\sigma\in\+Q_{\var{c}}$, it can be determined in $\poly(q,k)$ time 
if $c$ is already satisfied by $\sigma$. 
It is also safe to assume $\Delta\geq 2$ as otherwise the problem would be trivial. 
%


The main counting algorithm incorporates the idea of ``guiding assignment" proposed in~\cite{Vishesh21towards}.  We will construct a sequence of partial assignments $P_0,P_1,\dots,P_s\in \bigotimes_{v\in V}\left(Q_v\cup \{\hollowstar\}\right)$, where $P_0$ is the empty assignment, and for each $i\in [s]$, $P_i$ extends $P_{i-1}$ by assigning value to some unassigned variable $v_i^*$.  For any partial assignment $X$ we use $\+S_X$ to denote the set of satisfying assignments that agree with $X$ on the assigned variables. Then we will use the following telescopic product to estimate $Z_{\Phi}$:
\begin{align}\label{eq-telescope}
Z_{\Phi}=\frac{\abs{\+{S}_{P_{0}}}}{\abs{\+{S}_{P_{1}}}}\cdot\frac{\abs{\+{S}_{P_{1}}}}{\abs{\+{S}_{P_{2}}}}\cdots\frac{\abs{\+{S}_{P_{s-1}}}}{\abs{\+{S}_{P_{s}}}}\cdot \abs{\+{S}_{P_{s}}}=\abs{\+{S}_{P_{s}}}\cdot\prod\limits_{i\in [s]}\frac{\abs{\+{S}_{P_{i-1}}}}{\abs{\+{S}_{P_{i}}}}=\abs{\+{S}_{P_{s}}}\cdot\prod\limits_{i\in [s]}\left(\mu^{P_{i-1}}_{v^*_i}\right)^{-1}.
\end{align}

We will then calculate the number $\abs{\+{S}_{P_{s}}}$ and approximate the marginal probability $\mu^{P_{i-1}}_{v^*_i}$ conditional on $P_{i-1}$ for each $i$ respectively. We will calculate the former using a subroutine that exhaustively enumerates all possible satisfying assignment and approximate the latter using a ``marginal approximator" subroutine.

Intuitively, we need to carefully construct such ``guiding assignment" to  meet the following two requirements:
\begin{itemize}
    \item For each $i\in [s]$, the marginal probability $\mu^{P_{i-1}}_{v^*_i}$ is efficiently approximable with enough accuracy.
    \item The number of satisfying assignments $\abs{\+{S}_{P_{s}}}$ is efficiently enumerable.
\end{itemize}

The precise construction of such guiding assignment is a bit technical and involved. We then present the main framework of the algorithm, leaving some details to be specified later. One of the key steps is to ``freeze" the constraints with high violation probability to ensure no constraint becomes too easy to violate, so that a ``local uniformity" property is maintained throughout. The same idea has been used in~\cite{Vishesh21towards,he2022sampling} and dated back to~\cite{beck1991algorithmic}.

A key threshold $\pprime$ for the violation probability is chosen, for now to satisfy:
\begin{align}\label{eq:parameter-p-prime}
p<\pprime<(\mathrm{e}q\Delta)^{-1}
\end{align}
We will fix the specific choice of $\pprime$ later. 

%
\noindent

Given $\sigma \in \bigotimes_{v\in V}\left(Q_v\cup \{\hollowstar\}\right)$
and $c\in \+C$, we say $c$ is $\sigma$-frozen if 
$\mathbb{P}[\neg c \mid \sigma] > \pprime$.
Let $\digamma$ be some potential function defined over partial assignments which will be specified later.
The main counting algorithm then follows the procedure below, where the guiding assignment $X$ is constructed on the fly:
\par\addvspace{.5\baselineskip}
\framebox{
  \noindent
  \begin{tabularx}{15cm}{@{\hspace{\parindent}} l X c}
    \multicolumn{2}{@{\hspace{\parindent}}l}{\underline{Main counting algorithm}} \\
  1. & Initialize $X$ as the empty assignment and $Z$ as $1$.\\
  2. & For $i=1,\dots,n$, if $v_i$ is not involved in any $X$-frozen constraint, do the \hypertarget{line2}{followings}:\\
      & \quad (a) estimate the marginal distribution $\mu^X_{v_i}$ using \Cref{Alg:calc} and let $\hat{\mu}_{v_i}$ be the \hypertarget{line2a}{estimator};\\
      & \quad (b) $X(v_i)\leftarrow  \mathop{\arg\min}\limits_{a\in Q_v}\digamma(X_{v_i\gets a})$ and \hypertarget{line2b}{$Z\leftarrow Z/ \hat{\mu}_{v_i}(X(v_i))$}.\\
  3. & Use exhaustive enumeration for each connected component in $H_{\Phi^X}$ to compute $\abs{\+{S}_X}$ , the number  of ways to extend $X$ to a full satisfying assignment, and return 
  \hypertarget{line3}{$Z\cdot \abs{\+{S}_X}$.}
  \end{tabularx}
 }
\par\addvspace{.5\baselineskip}

\begin{remark}[Upper bound function $\digamma(\cdot)$]\label{remark:potential-function}
The upper bound function $\digamma(\cdot)$ plays a key role in the main counting algorithm.
It is chosen to be in conformity with our analysis as in \Cref{def:potential-refined}, and 
thus not explicitly defined here. 
In \Cref{def:potential-refined},
$\digamma(\cdot)$ is well designed such that 
both the upper bounds on the error and the time cost of the main counting algorithm can be derived from it.
Concretely, we have that
\sloppy
\begin{itemize}
    \item for each partial assignment $\sigma$, $\digamma(\sigma)$ is always (for any $v\in V$) an upper bound for the total variation distance between the output of \Cref{Alg:calc} and the marginal distribution $\mu^{\sigma}_v$;
     \item if $\digamma(X)$ is small, then we can obtain a good upper bound on the running time of the exhaustive enumeration part for calculating $\abs{\+{S}_X}$.
\end{itemize}
\end{remark}




\subsection{A marginal approximator}
The main tool of the main counting algorithm is a subroutine which returns a probability vector approximating the (conditional) marginal distribution $\mu^{\sigma}_v$ of a variable $v$.
Before presenting our subroutine, we need to formally define the notion of partial assignments, which is the same as the one defined in~\cite{he2022sampling}.

\begin{definition}[partial assignment]\label{def:partial-assignment}
Given a CSP formula $\Phi=(V,\+{Q},\+{C})$, 
let $\star$ and $\hollowstar$ be two special symbols not in any $Q_v$.
Define: 
\[
{\+Q}^\ast\triangleq\bigotimes_{v\in V}\left(Q_v\cup\{\star,\hollowstar\}\right).
\]
Each $\sigma\in {\+Q}^\ast$ is called a \emph{partial assignment}.
\end{definition}

Given a partial assignment $\sigma\in {\+Q}^\ast$, each variable $v\in V$ has three possibilities:
\begin{itemize}
\item $\sigma(v)\in Q_v$. That is, $v$ is \emph{accessed} by the algorithm and \emph{assigned} with the value $\sigma(v)\in Q_v$;
\item $\sigma(v)=\star$. That is, $v$ is just \emph{accessed} by the algorithm but  \emph{unassigned} yet with a value in $Q_v$;
\item $\sigma(v)=\hollowstar$. That is, $v$ is \emph{unaccessed} by the algorithm and hence \emph{unassigned} with any value.
\end{itemize}

%
%

Recall the notation $\Lambda(\sigma)\triangleq \{v\in V\mid  \sigma(v) \in Q_v\}$.
Given a partial assignment $\sigma\in {\+Q}^\ast$,
we further define $\Lambda^{+}(\sigma)
\triangleq \{v\in V\mid  \sigma(v)\neq \hollowstar\}$
to be the sets of accessed variables $\sigma$.
For any variable $v\in V$, let $\Mod{\sigma}{v}{x}$ be the partial assignment modified from $\sigma$ by replacing $\sigma(v)$ with $x\in Q_v\cup\{\star,\hollowstar\}$.

Given any two partial assignments $\sigma,\tau\in {\+Q}^\ast$,
if 
$\Lambda(\sigma)\subseteq\Lambda(\tau)$, $\Lambda^+(\sigma)\subseteq\Lambda^+(\tau)$, 
and $\sigma,\tau$ agree with each other over all  variables in $\Lambda(\sigma)$,
$\tau$ is said to \emph{extend} $\sigma$.
A partial assignment $\sigma$ is said to satisfy a constraint $c\in\+{C}$,
if $c$ is satisfied by all full assignments extending $\sigma$.
And $\sigma$ is said to be \emph{feasible}, if there is a satisfying assignment extending $\sigma$.

For each variable $v\in V$, we always assume an arbitrary order over all values in $Q_v$ in the paper.
Let $q_v\triangleq\abs{Q_v}$.
The following parameters are used in our subroutine:
\begin{align}
\theta_v \triangleq
\frac{1}{q_v}-\eta
\quad\text{and}\quad
\theta\triangleq\frac{1}{q}-\eta
\quad\text{ where}\quad 
\eta=\left(1-\mathrm{e}\pprime q\right)^{-\Delta}-1
\label{eq:parameter-theta}
\end{align}
Assuming the LLL condition in \eqref{eq:main-thm-LLL-condition}, we always have $\eta<\frac{1}{q_v}$, and hence $\theta_v,\theta>0$.

Next, we define some distributions used in the algorithm.
For any feasible $\sigma\in {\+Q}^\ast$ and any $S\subseteq V$, we denote by $\mu_S^{\sigma}$ the marginal distribution induced by $\mu$ on $S$ conditional on $\sigma$. 
Formally, for each $ \tau\in \+{Q}_{S}$, 
$\mu_S^{\sigma}(\tau)=\Pr[X\sim\mu]{X(S)=\tau\mid \forall v\in\Lambda(\sigma), X(v)=\sigma(v)}$.
We write
 $\mu_v^{\sigma}=\mu_{\{v\}}^{\sigma}$ for $v\in V$.
Similarly, for any $\sigma\in {\+Q}^\ast$ and any event $A\subseteq \+{Q}$,
denote that  
$\mathbb{P}[A\mid \sigma]= \mathbb{P}_{X\in\+{Q}}[X\in A\mid  \forall v\in\Lambda(\sigma), X(v)=\sigma(v)]$,
recalling that $\mathbb{P}$ is the law for the uniform product distribution over $\+{Q}$. 

For any $\sigma\in \qs$ and $v\in V$, define:
\begin{align}\label{eq:definition-recalc}
\forall x\in Q_v,\qquad
\+{D}^{\sigma}_{v}(x)\triangleq\frac{\mu_v^{\sigma}(x)-\theta_v}{1-q_v\cdot \theta_v}.
\end{align}
In our subroutine for approximately calculating $\mu^{\sigma}_v$, it is guaranteed that $\theta_v$ always lower bounds the marginal probability (\Cref{prop:theta-recalc-lower-bound}). 
Therefore, $\+{D}^{\sigma}_{v}$ is a well-defined probability distribution over $Q_v$.

\sloppy
By \eqref{eq:definition-recalc} we have that $\mu^{\sigma}_v=\theta_v+(1-q_v\theta_v)\+{D}^{\sigma}_{v} $. The $\calc{}$ then simply uses this equation to approximate the distribution $\mu^{\sigma}_v$, assuming another subroutine $\recalc{}(\Phi,\Mod{\sigma}{v}{\star},v)$ for approximately calculating $\+{D}^{\sigma}_{v}$. 
This is formally described in \Cref{Alg:calc}.

\begin{algorithm}  
  \caption{$\calc{}(\Phi,\sigma,v)$} \label{Alg:calc}
  \KwIn{a CSP formula $\Phi=(V,\+{C})$, a partial assignment $\sigma\in \qs$ and a variable $v$}
  \KwOut{a distribution approximating $\mu^\sigma_v(\cdot)$}
   $\hat{\+D}\gets \recalc{}(\Phi,\sigma_{v\gets \star},v)$\; \label{Line-calc-recalc}
        $\hat{\mu}(i) \gets \theta_v + (1-q_v\theta_v)\cdot \hat{\+D}(i)$ for each $1\leq i\leq q_v$\; \label{Line-calc-zone}
    \Return $\hat{\mu}$\;\label{Line-calc-R}
\end{algorithm} 

\subsection{A recursive approximator}
The goal of the $\recalc{}$ subroutine is to approximate the distribution $\+{D}^{\sigma}_{v}$ which is computed from the marginal distribution $\mu_v^{\sigma}$ as defined in \eqref{eq:definition-recalc}. This subroutine is a derandomization of the recursive marginal sampler in~\cite{he2022sampling}.

Note that we can compute the exact distribution $\+{D}^{\sigma}_{v}$ by exhaustively enumerating all assignments and checking if the assignment satisfies the formula. Still, such exhaustive enumeration can be inefficient as we must enumerate the assignment over too many variables.

Nevertheless, such exhaustive enumeration subroutine for computing $\+{D}^{\sigma}_{v}$ may serve as the basis of the recursion.
If sufficiently many variables are assigned during the recursion, 
the remaining CSP formula will be ``factorized'' with respect to $v$.
In most cases, the connected component containing $v$ in $H_{\Phi^\sigma}$ is small, 
in which case the exhaustive enumeration subroutine for computing $\+{D}_v^{\sigma}$ becomes efficient. 
Therefore, our approximator will try to assign all possible values to some variable that can help ``factorize" the formula and approximate the distribution $\+{D}^{\sigma}_{v}$ recursively. However, this may still be inefficient as the number of recursive calls may grow at an exponential rate in the recursion depth. 
To resolve this issue, we will {truncate} the recursion when some suitable condition is reached. Later we will show that with some properly formulated condition for truncation, our approximator is both efficient and accurate enough.

%


Before presenting the subroutine, we formally define notions of frozen constraints and fixed variables. 
Note that this definition also apply in Line \hyperlink{line2}{2} of the main counting algorithm.
\begin{definition}[frozen and fixed]\label{definition:frozen-fixed}
Let $\sigma\in \+{Q}^*$ be a partial assignment.
\begin{itemize}
\item  
A constraint $c\in \+{C}$ is called \emph{$\sigma$-frozen} if  $\mathbb{P}[\neg c \mid \sigma] > \pprime$.
Let $\cfrozen{\sigma}\triangleq \left\{c\in\+{C}\mid \mathbb{P}[\neg c \mid \sigma] > \pprime\right\}$ be  the set of all $\sigma$-frozen constraints.
 \item 
A variable $v\in V$ is called \emph{$\sigma$-fixed} if $v$ is accessed in $\sigma$ or  is involved in some $\sigma$-frozen constraint. 
Let $\vfix{\sigma}\triangleq  \Lambda^+(\sigma) \cup \bigcup_{c\in \cfrozen{\sigma}}\vbl(c)$ 
be the set of all $\sigma$-fixed variables.
\end{itemize}
\end{definition}

Given a partial assignment $\sigma$, the following definition specifies the next variable to assign according to $\sigma$,
which has already appeared in~\cite{he2022sampling}.

\begin{definition}[$\star$-influenced variables]\label{definition:boundary-variables}
Given a partial assignment $\sigma\in \+{Q}^*$, let $H^\sigma=H_{\Phi^{\sigma}}=(V^{\sigma},\+{C}^{\sigma})$ be the hypergraph of the simplified formula $\Phi^{\sigma}$ and $\hfix{\sigma}$ be the sub-hypergraph of $H^\sigma$ induced by $V^{\sigma}\cap\vfix{\sigma}$.
%
\begin{itemize}

\item
Let $\vcon{\sigma}\subseteq V^{\sigma}\cap\vfix{\sigma}$ be the set of vertices belong to the connected components in $\hfix{\sigma}$ that contain any $v$ with $\sigma(v)=\star$.
\item
Let $\vinf{\sigma}\triangleq\left\{u\in V^{\sigma}\setminus \vcon{\sigma}\mid \exists c\in \+{C}^{\sigma}, v\in \vcon{\sigma}:u,v\in\vbl(c)\right\}$ be the vertex boundary of $\vcon{\sigma}$ in~$H^\sigma$.

\item 
Let $\nextvar{\sigma}$ be the next variable to assign under $\sigma$ where
\begin{align}
\nextvar{\sigma}
\triangleq
\begin{cases}
v_i\in \vinf{\sigma}\text{ with smallest $i$} & \text{if }\vinf{\sigma}\neq\emptyset,\\
\perp & \text{otherwise}.
\end{cases}
\label{eq:definition-var}
\end{align}
\end{itemize}
\end{definition}

Intuitively, given a partial assignment $\sigma \in \+{Q}^*$, a variable $u$ is a good candidate for assignment if it has enough ``freedom'' under $\sigma$ ( $u$ is not $\sigma$-fixed) and can ``influence'' the variables that we are trying to approximate the marginal in the recursion (which are marked by $\star$)
through a chain of constraints in the simplified formula $\Phi^{\sigma}$.
The first such variable is returned by $\nextvar{\sigma}$.

The \recalc{} subroutine is given in \Cref{Alg:recalc}.

\begin{remark}[Truncation condition $f(\cdot)$]\label{remark:truncation-function}
Note that we haven't explicitly define the function $f(\cdot)$ in \Cref{Line-recalc-truncate} of \Cref{Alg:recalc}. This is for the same reason we didn't explicitly define $\digamma(\cdot)$ in the main counting algorithm. The function $f: \qs\rightarrow \{\True,\False\}$ is some kind of condition for ``truncation" that decides when we should stop the recursion.  An implementation of $f(\cdot)$ will be specified later in \Cref{def:truncate-refined}, to be in conformity with the analysis. 
\end{remark}

\begin{algorithm}  
  \caption{$\recalc{}(\Phi,\sigma,v)$} \label{Alg:recalc}
  \KwIn{a CSP formula $\Phi=(V,\+{C})$, a feasible partial assignment $\sigma\in \qs$ and a variable $v$} 
  \KwOut{a distribution over $Q_v$ that approximates the distribution  $\+{D}=\+{D}^{\sigma}_{v}$ defined in \eqref{eq:definition-recalc}}
  \If(\hspace{48pt}// \texttt{\small the condition for truncation is satisfied }\label{Line-recalc-truncate}) {$f(\sigma)=\True$ }
  { 
    \Return $\left(\frac{1}{q_v},\frac{1}{q_v},\cdots,\frac{1}{q_v}\right)$\;\label{Line-recalc-uniform}
  }
  \Else(\hspace{152pt}\label{Line-recalc-else} \texttt{\small })
  {
    $u\gets \nextvar{\sigma}$\;\label{Line-recalc-nextvar}
    \If{$u\neq \perp$\label{Line-recalc-loop}}
    {
        $\hat{\+{D}} \gets \left(0,0,\cdots,0\right)$\;\label{Line-recalc-mu}
        $\hat{\+D}^{\sigma}_{u}\gets \recalc{}(\Phi,\sigma_{u\gets \star},u)$\;\label{Line-recalc-recur1}
        $\hat{\mu}^{\sigma}_{u}(i) \gets \theta_u + (1-q_u\theta_u)\hat{\+D}^{\sigma}_{u}(i) \text{ for each } 1\leq i\leq q_u$\;   \label{Line-recalc-zone}
        \For(\hspace{60pt} // \texttt{\small approximate $\+D^{\sigma}_v(\cdot)$ by reduction}){$1\leq i\leq q_u$ }
        {
            $\hat{\+D}^{\sigma_{u\gets i}}_{v} \gets \recalc{}(\Phi,\sigma_{u\gets i},v)$\;\label{Line-recalc-recur2}
            \For(){$1\leq j\leq q_v$ }
            {
                $\hat{\+{D}}(j) \gets \hat{\+{D}}(j) + \hat{\mu}^{\sigma}_{u}(i)\cdot \hat{\+D}^{\sigma_{u\gets i}}_{v}(j)$\;\label{Line-recalc-calcmu}
            }
        }
        \Return $\hat{\+{D}}$\;\label{Line-recalc-return}
    }
    \Else(\hspace{136pt}// \texttt{\small the Factorization succeeds.})
    {
        Calculate $\mu_v^{\sigma}$ by counting the number of satisfying assignments exhaustively for the connected component in $H_{\Phi^{\sigma}}$ containing $v$\;
        Calculate $\+{D}^{\sigma}_{v}$ with $\mu_v^{\sigma}$ according to \eqref{eq:definition-recalc}\;\label{Line-recalc-enu}
        \Return $\+{D}^{\sigma}_{v}$\;
    }
  }
 
\end{algorithm} 

\subsection{The choice of the truncation condition and the upper bound function}
It remains to explicitly specify the upper bound function $\digamma(\cdot)$  and the truncation condition $f(\cdot)$, stated respectively in \Cref{remark:potential-function} and \Cref{remark:truncation-function}, 
to complete the definition of our algorithm.
For this purpose, we bring forward some definitions used in the analysis. In particular, we will introduce the notion of generalized $\{2,3\}$-tree, which is crucial to our choice of the upper bound function and is also a main technical contribution. 

\subsubsection{The choice of the truncation condition}

To specify our choice of the truncation condition, we need to classify those ``bad constraints" with respect to a partial assignment $\sigma$. 

\begin{definition}
\label{def:cbad}
Let $\sigma\in \+{Q}^*$ be a partial assignment.  
\begin{itemize}
\item 
Define $\ccon{\sigma}$ to be the set of constraints $c\in\+{C}$ such that $\vbl(c)$ intersects $\vcon{\sigma}$, where $\vcon{\sigma}$ is as defined in \Cref{definition:boundary-variables}.
\item
Define $\csfrozen{\sigma} \triangleq \cfrozen{\sigma}\cap \ccon{\sigma}$.
\item
Define  $\vstar{\sigma}\defeq\{v\in V\mid \sigma(v)=\star\}$ to be the set of variables set to $\star$ in $\sigma$.
\end{itemize}
\end{definition}

We are now ready to specify our choice of the truncation condition $f(\cdot)$.

\begin{definition}[Choice of the truncation condition $f(\cdot)$ ]\label{def:truncate-refined}
The truncation condition $f:\qs\rightarrow\{\True,\False\}$ is chosen as 
$$f(\sigma)\triangleq\one{\abs{\vst{\sigma}}+\Delta\cdot \abs{\csfrozen{\sigma}}\geq L\Delta}$$
for some integer parameter $L> 1$ to be specified later.
\end{definition}

\subsubsection{The choice of the upper bound function}

To specify our choice of the upper bound function, we will introduce the notion of ``generalized $\{2,3\}$-tree", which is a  combinatorial structure refined from the $\{2,3\}$-trees used in the analysis of~\cite{Vishesh21towards} and~\cite{he2022sampling}. 

Given a hypergraph $H=(V,\+E)$, let $\Lin{H}$ be the line graph of $H$ whose vertex set is the hyperedges in $\+E$ and two hyperedges in $\+E$ are adjacent in $\Lin{H}$ if and only if they share some vertex in $H$. 
Let $\text{dist}_{\Lin{H}}(\cdot,\cdot)$ be the shortest path distance in $\Lin{H}$. 

\begin{definition}{(generalized $\{2,3\}$-tree)}\label{WTdef}
Given a hypergraph $H=(V,\+E)$, A \emph{generalized $\{2,3\}$-tree} $T = U\cup E$, where $U\subseteq V$ and $E\subseteq  \+E$, is a subset of vertices and edges of $H$ such that the followings hold:
\begin{enumerate}
    \item For all distinct $u, v\in E$, $\text{dist}_{\Lin{H}}(u,v)\geq 2$.  \label{WT-1}
    \item 
    It holds for the directed graph $G(T,\mathcal{A})$ that
    there is a vertex $r\in T$ (called a root) which can reach all other vertices through directed paths,
    where the $G(T,\mathcal{A})$ is constructed on the vertex set $T$ as that, 
    for any $u, v\in T$ there is an arc $(u,v)\in\mathcal{A}$ if and only if at least one of the following conditions is satisfied: \label{WT-2}  
    \begin{itemize}
        \item $u,v\in E$ and $\text{dist}_{\Lin{H}}(u,v)=2\text{ or }3$;
        \item $u\in U, v\in E$ and  there exists $e\in \+E$ such that $u\in e\land \text{dist}_{\Lin{H}}(v,e)=1$;
        \item $u\in E,v\in U$ and there exists $e\in \+{E}$ such that $v\in e\land \text{dist}_{\Lin{H}}(u,e)=1\text{ or }2$;
        \item $u,v\in U$ and  there exists $e\in \+{E}$ such that $u,v\in e$.
    \end{itemize}
\end{enumerate}
Furthermore, any rooted directed spanning tree of the directed graph $G(T,\mathcal{A})$ constructed as above is called an \emph{auxiliary tree} of the generalized $\{2,3\}$-tree $T$.
\end{definition}

The generalized $\{2,3\}$-tree in \Cref{WTdef} is inspired by the the notion of $\{2,3\}$-tree defined for the line graph $\Lin{H}$~\cite{alon1991parallel}.
We generalize this notion to the original hypergraph $H$ to simultaneously depict the distances between vertices and hyperedges in $H$.
One can verify that every $\{2,3\}$-tree in the line graph $\Lin{H}$ is some generalized $\{2,3\}$-tree in the hypergraph $H$. 
Moreover, a generalized $\{2,3\}$-tree $T = U\cup E$ further restricts that each vertex in $U$ is close to its nearest neighbour in $T$.

Specifically, when the underlying hypergraph in \Cref{WTdef} is the hypergraph representation $H_{\Phi}=(V,\+C)$ of some CSP $\Phi$, a generalized $\{2,3\}$-tree $T\subseteq V\cup \+C$ in $H_{\Phi}$ becomes a subset of variables and constraints.

Given a subset $T\subset V\cup \+C$ of variables and constraints, we use $T=U \circ E$ to denote $T=U\cup E$  where $U\subseteq V$ and $E\subseteq \+C$.
We are now ready to specify our choice of the upper bound function $\digamma(\cdot)$.

\begin{definition}[Choice of the upper bound function $\digamma(\cdot)$ ]\label{def:potential-refined} The upper bound function $\digamma:\qs\rightarrow \mathbb{R}$ is fixed as follows.

For any subset of vertices and constraints $T = U\circ E$ and any partial assignment $\sigma\in \qs$,  define
\begin{align}\label{eq-def-fsimgat}
F(\sigma,T)\defeq \left(1-q\theta\right)^{\abs{U}}
\prod\limits_{c\in E}\left(\pprime^{-1}\mathbb{P}[\neg c\mid \sigma](1+\eta)^k\right).
\end{align}

For any integer $t>0$,  define \begin{align}\label{eq-definition-FT-1}
\+{T}^{t} \defeq \left\{T\mid T= U\circ E \text{ is a generalized $\{2,3\}$-tree in $H_{\Phi}$ satisfying}  \abs{U}+\Delta\cdot \abs{E}=t \right\} 
\end{align}
Moreover, for any integer $t>0$ and $v\in V$,  define \begin{align}\label{eq-definition-FT-2}
\+{T}^{t}_{v} \defeq \left\{T\in \+{T}^{t}\mid \text{ there exists an auxiliary tree of $T$ rooted at $v$}  \right\} 
\end{align}
Finally, for any partial assignment $\sigma\in \qs$, we define

\begin{align}\label{eq-definition-F-2}
\digamma(\sigma)\defeq \sum\limits_{i=L}^{L\Delta}\sum\limits_{v\in V} \sum\limits_{T\in \+{T}^{i}_v}F(\sigma,T\setminus \set{v} ),
\end{align}
where $L$ is the same unspecified parameter as in the definition of the truncation condition.
\end{definition}

\section{Analysis of the counting algorithm}\label{sec:analysis}

In this section, we present the analysis of our deterministic approximate counting algorithm. We will prove the following theorem.

\begin{theorem}\label{thm:main-counting-refined}
With the $f(\cdot)$ and $\digamma(\cdot)$ as specified respectively in \Cref{def:truncate-refined} and  \Cref{def:potential-refined}, 
for any CSP formula $\Phi=(V,\+{C})$ satisfying \eqref{eq:main-thm-LLL-condition} and $0<\varepsilon<1$, the main counting algorithm (given in \Cref{sec:main-counting-alg}) returns a $\widehat{Z}$ satisfying
$(1-\varepsilon)Z_{\Phi} \leq \widehat{Z} \leq (1+\varepsilon)Z_{\Phi}$, within time $O\left(\left(\frac{n}{\varepsilon}\right)^{\poly(\log{q},\Delta,k)}\right)$.
\end{theorem}

\subsection{Invariants and local uniformity }

In this subsection, we present some basic facts that guarantee our algorithm is well-defined.  The following two invariants are respectively satisfied by the \calc{} and \recalc{} subroutine called within the counting algorithm (formally proved in \Cref{lemma:invariant-counting}).

\begin{condition}[invariant for \calc{}]\label{inputcondition-calc}
The followings hold for the input tuple $(\Phi, \sigma, v)$:
\begin{itemize}
\item $\Phi=(V,\+{Q},\+{C})$ is a CSP formula, $\sigma\in\+{Q}^*$ is a feasible partial assignment, and  $v\in V$ is a variable;
\item $v$ is not $\sigma$-fixed and $\sigma(v)=\hollowstar$, and for all $u\in V$, $\sigma(u)\in Q_u\cup\{\hollowstar\}$;
\item $\mathbb{P}[\neg c\mid \sigma]\leq \pprime q$ for all $c\in \mathcal{C}$.
\end{itemize}
\end{condition}

\begin{condition}[invariant for \recalc{}]\label{inputcondition-recalc}
The followings hold for the input tuple $(\Phi, \sigma, v)$:
\begin{itemize}
\item $\Phi=(V,\+{Q},\+{C})$ is a CSP formula, $\sigma\in\+{Q}^*$ is a feasible partial assignment, and  $v\in V$ is a variable;
\item $\sigma(v)=\star$;
\item $\mathbb{P}[\neg c\mid \sigma]\leq \pprime q$ for all $c\in \mathcal{C}$.
\end{itemize}
\end{condition}

\begin{lemma}\label{lemma:invariant-counting}
When the input CSP formula $\Phi=(V,\+{Q},\+{C})$ satisfies \eqref{eq:main-thm-LLL-condition},
the invariants are satisfied during the execution of the algorithm: 
\begin{enumerate}
\item\label{item:invariant-calc} 
whenever $\calc{}(\Phi, \sigma, v)$ is called, \Cref{inputcondition-calc} is satisfied by its input $(\Phi, \sigma, v)$;
\item \label{item:invariant-recalc}
whenever \recalc{}$(\Phi, \sigma, v)$ is called, \Cref{inputcondition-recalc}  is satisfied by its input $(\Phi, \sigma, v)$.
\end{enumerate}
\end{lemma}

We then prove \Cref{lemma:invariant-counting}. Before that, we formally define the sequence of partial assignments that evolve in the main counting algorithm.

\begin{definition}[partial assignments in main counting algorithm]\label{def-pas-cmain}
Let $X^0,X^1,\dots,X^n\in\+{Q}^*$ denote the sequence of partial assignments, 
where $X^0=\hollowstar^V$ 
and for every $1\le i\le n$, 
$X^i$ is the partial assignment $X$ after the $i$-th iteration of the for-loop in Line \hyperlink{line2}{2} of the main counting algorithm.
\end{definition}

The following two lemmas are immediate by~\cite[Lemma 5.8]{he2022sampling} and~\cite[Lemma 5.9]{he2022sampling}, respectively. We then omit the proof.

\begin{lemma}\label{lemma:invariant-p-prime-q-bound}
For the $X^0,X^1,\dots,X^n$ in \Cref{def-pas-cmain}, it holds for all $0\le i\le n$ that $X^i$ is feasible and
\begin{align}
\forall c\in \+{C},\qquad \mathbb{P}[\neg c\mid X^i]\leq \pprime q. \label{eq:invariant-p-prime-q-bound}
\end{align}
\end{lemma}

\begin{lemma}\label{lemma:general-invariant}
Assume \Cref{inputcondition-recalc} for $(\Phi,\sigma,v)$.
For any $u\in V$, if $u$ is not $\sigma$-fixed, then $(\Phi,\Mod{\sigma}{u}{a},v)$ satisfies \Cref{inputcondition-recalc} for any $a\in Q_u\cup \set{\star}$.
\end{lemma}

The invariant of \Cref{inputcondition-calc} for \calc{} stated in
\Cref{lemma:invariant-counting}-(\ref{item:invariant-calc})
follows directly from \Cref{lemma:invariant-p-prime-q-bound}.  The invariant of \Cref{inputcondition-recalc} for \recalc{} stated in
\Cref{lemma:invariant-counting}-(\ref{item:invariant-recalc})
follows from \Cref{lemma:general-invariant}, because during the execution,  
the algorithm will only change an input partial assignment $\sigma$ to $\Mod{\sigma}{u}{a}$ for $u=\nextvar{\sigma}$ and for $a\in Q_{u}\cup\{\star\}$, and that by the definition of $\nextvar{\cdot}$ we have $u$ is not $\sigma$-fixed.
Therefore, \Cref{lemma:invariant-counting} is proved.

The next proposition, shows that $\theta_v$ always lower bounds the marginal probability $\mu^{\sigma}_v(\cdot)$ for $(\Phi,\sigma,v)$ satisfying \Cref{inputcondition-recalc}.  Combining with \Cref{lemma:invariant-counting}-(\ref{item:invariant-recalc}), we have shown the well-definedness of the distribution $D^{\sigma}_v$ defined in \eqref{eq:definition-recalc} and \Cref{Alg:recalc}.

\begin{proposition}\label{prop:theta-recalc-lower-bound}
Assuming \Cref{inputcondition-recalc} for the input $(\Phi, \sigma ,v)$, 
it holds that $\min\limits_{x\in Q_v}\mu^{\sigma}_v(x)\ge  \theta_v$ and for $u=\nextvar{\sigma}$, if $u\neq\perp$ then it also holds that $\min\limits_{x\in Q_u}\mu^{\sigma}_u(x)\ge  \theta_u$.
\end{proposition}

Recall $\pprime$ defined in \eqref{eq:parameter-p-prime} and $\theta_v,  \eta$ defined in \eqref{eq:parameter-theta}. The following corollary implied by the ``local uniformity" property is immediate by~\cite[Proposition 3.9]{he2022sampling} and directly proves \Cref{prop:theta-recalc-lower-bound}.

\begin{corollary}\label{localuniformitycor}
For any CSP formula $\Phi=(V,\+{Q},\+{C})$ and any partial assignment $\sigma\in\+{Q}^*$,
if
\[
\forall c\in \mathcal{C},\quad \mathbb{P}[\neg c\mid \sigma]\leq \pprime q,
\]  
then $\sigma$ is feasible, and for any $v\in V\setminus\Lambda(\sigma)$, and any $x\in Q_v$
\[
  \mu^{\sigma}_v(x) \geq \theta_v.
\]
\end{corollary}
%
%


\subsection{The recursive cost tree} A key combinatorial structure used in our proof for Theorem \ref{thm:main-counting-refined} is the Recursive Cost Tree (RCT). For each $v\in V$, we further define $\qus{v}\defeq Q_v\cup \set{\star}$ as the extended domain for accessment. Note that a difference between the RCT here and that in~\cite{he2022sampling} is that the RCT here stops growing once the truncation condition is satisfied.

\begin{definition}[Recursive Cost Tree]\label{RCTdef}
For any partial assignment {$\sigma\in \qs$}, 
let $\+{T}_{\sigma}=(T_{\sigma},\rho_{\sigma})$, where $T_{\sigma}$ is a rooted tree with nodes $V(T_\sigma)\subseteq \qs$ and $\rho_\sigma :V(T_\sigma)\rightarrow [0,1]$ is a labeling of nodes in $T_\sigma$, be constructed as:
\begin{enumerate}
    \item The root of $T_{\sigma}$ is $\sigma$, with $\rho_{\sigma}(\sigma)=1$ and depth of $\sigma$ being 0; \label{RCT-1}
    \item for $i=0,1,\ldots$: for all nodes $X\in V(T_\sigma)$ of depth $i$ in the current $T_{\sigma}$,\label{RCT-2}
    \begin{enumerate}
        \item if $\nextvar{X}=\perp$ or $f(X)=\True$, then leave $X$  as a leaf node in $T_\sigma$; \label{RCT-2-a}
        \item otherwise, supposed $u=\nextvar{X}$, append $\{X_{u\gets x}\mid x\in \qus{u}\}$ as the $q_u+1$ children to the node $X$  in $T_\sigma$, and label them as: \label{RCT-2-b}
        \begin{align*}
        \forall x\in \qus{u},\quad
        \rho_\sigma(X_{u\gets x})=
        \begin{cases}
        (1-{q_u\theta_u}) \rho_\sigma(X) &  \text{if }x=\star,\\
        \mu^\sigma_u(x)\cdot \rho_\sigma(X) & \text{if }x\in Q_u.
        \end{cases}
        \end{align*}
    \end{enumerate}
\end{enumerate}
The resulting $\+{T}_{\sigma}=(T_{\sigma},\rho_{\sigma})$ is called the \emph{ recursive cost tree (RCT) rooted at $\sigma$}.
\end{definition}


For any RCT $\+{T}_{\sigma}$, let $\+{L}(\+{T}_{\sigma})$ be the set of leaf nodes in $\+{T}_{\sigma}$. Let $\+{L}_g(\+{T}_{\sigma})\defeq\{X\in \+{L}(\+{T}_{\sigma}): f(X)=\False \}$ and $\+{L}_b(\+{T}_{\sigma})\defeq\{X\in \+{L}(\+{T}_{\sigma}): f(X)=\True\}$ be the sets of leaf nodes in $\+{T}_{\sigma}$ that don't and do satisfy the truncation condition, respectively. We also define the following function $\lambda(\cdot)$ on $\+{T}_{\sigma}$:
\begin{align}\label{eq-lambdadef}
\lambda(\+{T}_{\sigma}) \defeq \sum_{X\in  \+{L}_b(\+{T}_{\sigma})}\rho_{\sigma}(X) .
\end{align}



Recall the definition of total variation distance. Let $\mu$ and $\nu$ be two probability distributions
over the same sample space $\Omega_{S}$. The total variation distance between $u$ and $v$ is defined by
$$\tv(\mu,\nu)\triangleq \frac{1}{2}\sum\limits_{x\in \Omega_{S}}\abs{ \mu(x)-\nu(x)}.$$

The total variation distance between the distribution returned by the subroutine \calc{} and the true marginal distribution is upper bounded through $\lambda(\cdot)$.

 \begin{lemma}\label{boundedRCTtvdcor}

For any $(\Phi,\sigma,v)$ satisfying \Cref{inputcondition-calc},  it holds that
$$\tv(\hat{\mu},\mu^{\sigma}_v)\leq \lambda(\+{T}_{\sigma_{v\gets\star}}),$$
where $\hat{\mu}$ is the distribution returned by $\calc{}(\Phi,\sigma,v)$.
\end{lemma}

We then prove \Cref{boundedRCTtvdcor}. 
The following recursive relation for RCT is immediate by definition.

\begin{proposition}\label{fact-RCT}
Let $\sigma\in \qs$, $u=\nextvar{\sigma}$.
If $u\neq\perp$ and $f(\sigma)=\False$, then
$$\lambda(\+{T}_{\sigma}) = (1-q_u\theta_u)\lambda\left(\+{T}_{\Mod{\sigma}{u}{\star}}\right)
             +\sum_{x\in Q_u}\left(\mu^\sigma_u(x)\cdot \lambda\left(\+{T}_{\Mod{\sigma}{u}{x}}\right)\right)$$
\end{proposition}

We have the following lemma which bounds the total variation distance between the distribution returned by the subroutine \recalc{} and the ``overflow" marginal distribution $\+{D}$.

\begin{lemma}\label{boundedRCTtvd}
For any $(\Phi,\sigma,v)$ satisfying \Cref{inputcondition-recalc}, it holds that
$$\tv(\hat{\+D},\+D)\leq \lambda(\+{T}_{\sigma}),$$
where $\hat{\+D}$ is the distribution returned by $\recalc{}(\Phi,\sigma,v)$ and $\+{D}\triangleq\+{D}^{\sigma}_v=\frac{\mu_v^{\sigma}-\theta_v}{1-q_v \theta_v}$.
\end{lemma}

\begin{proof}
We prove the lemma by an induction on the structure of the RCT. The base case is when $T_\sigma$ is a single root.
Thus we have $\sigma\in \+{L}(\+T_\sigma)$.
By \Cref{RCT-2-a} of \Cref{RCTdef},
we also have $f(\sigma)=\True$ or $\nextvar{\sigma}= \perp$.
For these two possibilities,
we always have $\tv(\hat{\+{D}},\+D)\leq \lambda(\+{T}_{\sigma})$.
\begin{enumerate}
    \item If $f(\sigma)=\True$, we have $\sigma\in \+{L}_b(\+T_\sigma)$ by $\sigma\in \+{L}(\+T_\sigma)$.
    Thus, $\lambda(\+{T}_{\sigma}) = \rho_{\sigma}(\sigma) = 1\geq \tv(\hat{\+{D}},\+D)$.
    \item Otherwise, $f(\sigma)=\False$ and $\nextvar{\sigma}= \perp$.
    Thus, the condition in \Cref{Line-recalc-loop} of \recalc{}($\Phi,\sigma,v$) is not satisfied and Lines \ref{Line-recalc-nextvar}-\ref{Line-recalc-calcmu} are skipped.
    According to \Cref{Line-recalc-enu}, we have
    $\hat{\+{D}}$ is exactly $\+D$.
    Therefore $\tv(\hat{\+{D}},\+D)=0\leq \lambda(\+{T}_{\sigma})$, because $\lambda(\+{T}_{\sigma})$ is nonnegative.
\end{enumerate}

For the induction step, we assume that $T_{\sigma}$ is a tree of depth $>0$, which implies $f(\sigma)=\False$ and $\nextvar{\sigma}=u\neq \perp$ for some $u\in V$ by \Cref{RCT-2-a} of \Cref{RCTdef}. 
Let $\hat{\+{D}}_{u}^{\sigma_{u\leftarrow \star}}$ be the probability vector returned by $\recalc(\Phi,\sigma_{u\gets \star},u)$ and $\hat{\+{D}}_{v}^{\sigma_{u\leftarrow x}}$ be the probability vector returned by $\recalc(\Phi,\sigma_{u\gets x},v)$ for each $x\in Q_u$.
Given a probability vector $\mathbf{p}$,
we use $\mathbf{p}(j)$ to denote the $j$-th item of $\mathbf{p}$.
By \Cref{Line-recalc-recur1}-\Cref{Line-recalc-calcmu} of \Cref{Alg:recalc}, we have
for each $j\in Q_v$,
\begin{equation}
\begin{aligned}\label{boundedRCTtvd-eq1}
    \hat{\+{D}}(j)&=\sum\limits_{i\in Q_u}\left(\theta_u+(1-q_u\theta_u)\hat{\+{D}}_{u}^{\sigma_{u\leftarrow \star}}(i)\right)\cdot \hat{\+{D}}_{v}^{\sigma_{u\leftarrow i}}(j)\\
    &=\sum\limits_{i\in Q_u}\left(\theta_u+(1-q_u\theta_u)\+D_{u}^{\sigma_{u\leftarrow \star}}(i)+(1-q_u\theta_u)\left(\hat{\+{D}}_{u}^{\sigma_{u\leftarrow \star}}(i)-\+D_{u}^{\sigma_{u\leftarrow \star}}(i)\right)\right)\cdot \hat{\+{D}}_{v}^{\sigma_{u\leftarrow i}}(j)\\
    &=\sum\limits_{i\in Q_u}\left(\mu_{u}^{\sigma}(i)+(1-q_u\theta_u)\left(\hat{\+{D}}_{u}^{\sigma_{u\leftarrow \star}}(i)-\+D_{u}^{\sigma_{u\leftarrow \star}}(i)\right)\right)\cdot \hat{\+{D}}_{v}^{\sigma_{u\leftarrow i}}(j),
\end{aligned}
\end{equation}
where the last equality is by the definition of $\+D(\cdot)$.
Moreover, by the chain rule, we also have for each $j\in Q_v$,
\begin{equation}
\begin{aligned}\label{boundedRCTtvd-eq2}
\+D(j)&=\sum\limits_{i\in Q_u}\+{\mu}^{\sigma}_u(i)\cdot \+{D}^{\sigma_{u\leftarrow i}}_v(j).
\end{aligned}
\end{equation}
Combining \eqref{boundedRCTtvd-eq1} with \eqref{boundedRCTtvd-eq2}, we have
\begin{equation*}
\begin{aligned}
    \hat{\+{D}}(j)-\+D(j)=\sum\limits_{i\in Q_u}\left(\mu^{\sigma}_u(i)\cdot (\hat{\+{D}}_{v}^{\sigma_{u\leftarrow i}}(j)-\+D_{v}^{\sigma_{u\leftarrow i}}(j))+(1-q_u\theta_u)\left(\hat{\+{D}}_{u}^{\sigma_{u\leftarrow \star}}(i)-\+D_{u}^{\sigma_{u\leftarrow \star}}(i)\right)\cdot \hat{\+{D}}_{v}^{\sigma_{u\leftarrow i}}(j)\right).
\end{aligned}
\end{equation*}
Combining with the triangle inequality for absolute values, we have 
\begin{equation*}
\begin{aligned}
    \abs{\hat{\+{D}}(j)-\+D(j)}&\leq\sum\limits_{i\in Q_u}\left(\mu^{\sigma}_u(i)\cdot \abs{\hat{\+{D}}_{v}^{\sigma_{u\leftarrow i}}(j)-\+D_{v}^{\sigma_{u\leftarrow i}}(j)}+(1-q_u\theta_u)\cdot  \abs{\hat{\+{D}}_{u}^{\sigma_{u\leftarrow \star}}(i)-\+D_{u}^{\sigma_{u\leftarrow \star}}(i)}\cdot\hat{\+{D}}_{v}^{\sigma_{u\leftarrow i}}(j)\right).
\end{aligned}
\end{equation*}
Therefore, we have
\begin{equation} \label{boundedRCTtvd-eq3}
\begin{aligned}
\tv(\hat{\+{D}},\+D)
=
&\frac{1}{2}\sum\limits_{j\in Q_v}\abs{ \hat{\+{D}}(j)-\+D(j)}\\
\leq
&\frac{1}{2}
\sum\limits_{i\in Q_u}\left(\mu^{\sigma}_u(i)\cdot \sum\limits_{j\in Q_v}\abs{\hat{\+{D}}_{v}^{\sigma_{u\leftarrow i}}(j)-\+D_{v}^{\sigma_{u\leftarrow i}}(j)}\right.\\
&\qquad\qquad\left.+(1-q_u\theta_u)\cdot \abs{\hat{\+{D}}_{u}^{\sigma_{u\leftarrow \star}}(i)-\+D_{u}^{\sigma_{u\leftarrow \star}}(i)}\cdot\sum_{j\in Q_v} \hat{\+{D}}_{v}^{\sigma_{u\leftarrow i}}(j)\right)\\
\leq
&\sum\limits_{i\in Q_u}\left(u^{\sigma}_u(i)\cdot \tv(\hat{\+{D}}_{v}^{\sigma_{u\leftarrow i}},\+D_{v}^{\sigma_{u\leftarrow i}})+\frac{1}{2}\cdot(1-q_u\theta_u)\cdot \abs{\hat{\+{D}}_{u}^{\sigma_{u\leftarrow \star}}(i)-\+D_{u}^{\sigma_{u\leftarrow \star}}(i)}\right)\\
\leq
&(1-q_u\theta_u)\cdot \tv(\hat{\+{D}}_{u}^{\sigma_{u\leftarrow \star}},\+D_{u}^{\sigma_{u\leftarrow \star}})+ \sum\limits_{i\in Q_u}\left(\mu^{\sigma}_u(i)\cdot \tv(\hat{\+{D}}_{v}^{\sigma_{u\leftarrow i}},\+D_{v}^{\sigma_{u\leftarrow i}})\right).
\end{aligned}
\end{equation}

Note that by \Cref{RCT-2-b} in \Cref{RCTdef}, for each $x\in \qus{u}$, the subtree in $T_{\sigma}$ rooted by $\sigma_{u\gets x}$ is precisely the $T_{X}$ in the RCT  $\+{T}_X=(T_X,\rho_X)$ rooted at $X=\sigma_{u\gets x}$. Moreover, by \Cref{lemma:general-invariant}, \Cref{inputcondition-recalc} is still satisfied by $(\Phi,\sigma_{u\gets x},v)$.
Thus, by induction hypothesis, we have for each $x\in \qus{u}$,
\begin{align*}
\tv(\hat{\+{D}}_{v}^{\sigma_{u\leftarrow x}},\+D_{v}^{\sigma_{u\leftarrow x}})\leq \lambda(\+{T}_{\sigma_{u\gets x}}).
\end{align*}

Combining with \eqref{boundedRCTtvd-eq3}, it follows that

\begin{align*}
\tv(\hat{\+{D}},\+D)\leq &(1-q_u\theta_u)\cdot \tv(\hat{\+{D}}_{u}^{\sigma_{u\leftarrow \star}},\+D_{u}^{\sigma_{u\leftarrow \star}})+ \sum\limits_{i\in Q_u}\left(\mu^{\sigma}_u(i)\cdot \tv(\hat{\+{D}}_{v}^{\sigma_{u\leftarrow i}},\+D_{v}^{\sigma_{u\leftarrow i}})\right)\\
\leq &(1-q_u\theta_u)\lambda(\+{T}_{\sigma_{u\gets \star}})+\sum\limits_{x\in Q_u}\mu_{u}^\sigma(x)\lambda(\+{T}_{\sigma_{u\gets x}})\\
=&\lambda(\+{T}_{\sigma}).
\end{align*}
where the equality follows by \Cref{fact-RCT}.
\end{proof}

For any $(\Phi,\sigma,v)$ satisfying \Cref{inputcondition-calc}, 
one can verify that $(\Phi,\sigma_{v\leftarrow \star},v)$ satisfies \Cref{inputcondition-recalc}.
Let $\hat{\+D}$ be the distribution returned by $\recalc{}(\Phi,\sigma,v)$.
By \Cref{boundedRCTtvd} we have $\tv(\hat{\+D},\+D^{\sigma}_v)\leq \lambda(\+{T}_{\sigma_{v\gets\star}})$.
Thus, by Lines \ref{Line-calc-recalc}-\ref{Line-calc-R} of $\calc{}(\Phi,\sigma,v)$, we have
$$\tv(\hat{\mu},\mu^{\sigma}_{v})= (1-q_v\theta_v)\tv(\hat{\+D},\+D^{\sigma}_v)\leq (1-q_v\theta_v)\lambda(\+{T}_{\sigma_{v\gets\star}})
\leq (1-q\theta)\lambda(\+{T}_{\sigma_{v\gets\star}})
\leq \lambda(\+{T}_{\sigma_{v\gets\star}}),$$
where the equality is by 
the definition of $\+D^{\sigma}_v$. This proves \Cref{boundedRCTtvdcor}.

\subsection{A random path simulating RCT}

The recursive cost tree in \Cref{RCTdef} inspires the following random process of partial assignments.
Given a partial assignment $\sigma$ and a variable $v\in V\setminus \Lambda(\sigma)$, define 
\begin{align*}
\induceddist{\sigma}{v}(\star)
&=\frac{1-q_v\cdot \theta_v}{2-q_v\cdot \theta_v}, \\
\forall x\in Q_v,\qquad
\induceddist{\sigma}{v}(x)
&=\frac{\mu_{v}^{\sigma}(x)}{2-q_v\cdot \theta_v}.
\end{align*}
It is obvious to see that $\induceddist{\sigma}{v}(\cdot)$ is a well-defined probability distribution over $\qus{v}$.

\begin{definition}[the $\pth(\sigma)$ process]\label{pathdef}
For any {$\sigma\in \qs$}, 
$\pth(\sigma)=(\sigma_0,\sigma_1,\ldots,\sigma_\ell)$ is a random sequence of partial assignments generated from the initial $\sigma_0 = \sigma$ as that for $i=0,1,\ldots$:
\begin{enumerate}
    \item 
    if $\nextvar{\sigma_i}=\perp$ or $f(\sigma_i)=\True$, the sequence stops at $\sigma_i$;\label{rp-1}
    \item
    otherwise $u=\nextvar{\sigma_i}\in V$, 
    the partial assignment $\sigma_{i+1}\in\qs$ is generated from $\sigma_{i}$ by randomly giving $\sigma(u)$ a value $x\in \qus{u}$, \label{rp-2}
    such that
    \begin{align*}
    \forall x\in \qus{u},\qquad
    \Pr{\sigma_{i+1}=\Mod{(\sigma_i)}{u}{x}}=\induceddist{\sigma}{u}(x).
    \end{align*}
\end{enumerate}
\end{definition}

The length $\ell(\sigma)$ of $\pth(\sigma) = \left(\sigma_0,\sigma_1,\cdots,\sigma_{\ell(\sigma)}\right)$ is a random variable whose distribution is determined by $\sigma$.
We simply write $\ell=\ell(\sigma)$ and $\pth(\sigma) = \left(\sigma_0,\sigma_1,\cdots,\sigma_{\ell}\right)$ if $\sigma$ is clear from the context. It is straightforward to verify that $\pth(\sigma)$ satisfies the Markov property. 

The significance of the random process $\pth(\sigma)$ is that it is related to the total variation distance between the distribution returned by $\calc{}(\Phi,\sigma,v)$ and $\mu^{\sigma}_{v}$ through the following function $H(\cdot)$.
For any two partial assignments $\tau_1,\tau_2$, define
$$\chi(\tau_1,\tau_2) \triangleq \prod\limits_{v\in \Lambda^{+}(\tau_1)\setminus \Lambda^{+}(\tau_2)}\left(2-q_v\theta_v\right).$$
Given any sequence  $P=(\sigma_0,\sigma_1,\dots,\sigma_{\ell})\in(\qs)^{\ell+1}$ with $\ell\ge 0$, define
\begin{align}\label{eq-pathweightdef}
H(P)\defeq \one{f(\sigma_{\ell})=\True}\cdot \chi(\sigma_{\ell},\sigma_0).
\end{align}

Recall the $\lambda(\+{T}_{\sigma})$ defined in \eqref{eq-lambdadef}. We have the following lemma.

\begin{lemma}\label{pathprop}
For any partial assignment $\sigma\in \qs$, the following holds
for $P=\pth(\sigma)$:
\begin{align*}
   \E{H(P)} = \lambda(\+T_{\sigma})
   \end{align*}
\end{lemma}
\begin{proof}

We show the lemma by an induction on the structure of the RCT. The base case is when $T_\sigma$ is a single root. 
Thus we have $\sigma\in \+{L}(\+T_\sigma)$.
By \Cref{RCT-2-a} of \Cref{RCTdef}, we have $\nextvar{\sigma}= \perp$ or $f(\sigma)=\True$.
Also, by \Cref{pathdef} we have $\pth(\sigma)=(\sigma)$ and $\ell(\sigma)=0$. 
Then we always have $\lambda(\+{T}_{\sigma}) = H(\pth(\sigma))$ no matter whether $f(\sigma)=\True$:
\begin{enumerate}
    \item If $f(\sigma)=\True$, then we have $\sigma\in \+{L}_b(\+T_\sigma)$ by $\sigma\in \+{L}(\+T_\sigma)$.
    Thus, $\lambda(\+{T}_{\sigma}) = \rho_{\sigma}(\sigma) = 1$.
    Meanwhile, by $\pth(\sigma)=(\sigma)$,
    we have $\sigma_\ell = \sigma_0 = \sigma$.
    Thus,
    $f(\sigma_\ell)=f(\sigma)=\True$ and $\Lambda^{+}(\sigma_{\ell})\setminus \Lambda^{+}(\sigma_0) =\emptyset$.
    We have $H(\pth(\sigma))=1$.
    In summary, we have
    $\lambda(\+{T}_{\sigma}) = H(\pth(\sigma))$.
    \item Otherwise, $f(\sigma)=\False$.
    We have 
    $\sigma\not\in \+{L}_b(\+T_\sigma)$
    and $\+{L}_b(\+T_\sigma) = \emptyset$.
    Thus, we have $\lambda(\+T_{\sigma})=0$.
    Also, by $\pth(\sigma)=(\sigma)$,
    we have $\sigma_\ell =  \sigma$.
    Combining with $f(\sigma)=\False$,
    we have 
    $f(\sigma_{\ell})=f(\sigma) =\False$.
    Thus $H(\pth(\sigma))=0$.
    In summary,  $\lambda(\+T_{\sigma})= H(\pth(\sigma))$.
\end{enumerate}

For the induction step, we assume that $T_{\sigma}$ is a tree of depth $>0$.
Thus by \Cref{RCT-2-a} of \Cref{RCTdef},
we have $f(\sigma)=\False$ and  $\nextvar{\sigma}=u\neq \perp$ for some $u\in V$.
According to \Cref{rp-2} of \Cref{pathdef}, we have $\ell(\sigma)\geq 1$ and
\begin{align}
    \forall x\in \qus{u},\quad \Pr{\sigma_{1} = \sigma_{u\gets x}}&=\induceddist{\sigma}{u}(x).\label{eq-x1-yi}
\end{align}
Moreover, by the Markov property of $\pth(\sigma)$, given $\sigma_1 = \sigma_{u\gets x}$ for each $x\in \qus{u}$, the subsequence  $(\sigma_1,\sigma_2,\cdots,\sigma_{\ell})$  is identically  distributed as $\pth(\sigma_{u\gets x})$. In addition, it can be verified that for any  sequence of partial assignments $P=(\tau_0,\tau_1,\dots,\tau_{\ell})$ with $\ell\geq 1$ satisfying $\Pr{\pth(\sigma)= P}> 0$, 
\begin{align}
H(P)=(2-q_u\theta_u)H((\tau_1,\dots,\tau_{\ell})).\label{eq:weightfact}
\end{align}
There are two possibilities:
\begin{enumerate}
    \item If $f(\tau_{\ell})=\False$, then we have $H(P)=H((\tau_1,\dots,\tau_{\ell}))=0 = (2-q_u\theta_u)H((\tau_1,\dots,\tau_{\ell}))$.
    \item Otherwise,
    $f(\tau_{\ell})=\True$.
    By $\Pr{\pth(\sigma)= P}> 0$, we have $\tau_0= \sigma$, $\Lambda^+(\tau_1) = \Lambda^+(\tau_0)\cup\{\nextvar{\tau_0}\}$, and $\Lambda^+(\tau_0)\subsetneq\Lambda^+(\tau_1)\subseteq \Lambda^+(\tau_{\ell})$.
    Combining with 
    $\nextvar{\sigma}=u$,
    we have $\nextvar{\tau_0}=\nextvar{\sigma} =u$
    , $\Lambda^+(\tau_1) = \Lambda^+(\tau_0)\cup\{u\}$, and $u\not\in \Lambda^+(\tau_0)$.
    Combining with $\Lambda^+(\tau_1)\subseteq \Lambda^+(\tau_{\ell})$,
    we have 
    $\Lambda^{+}(\tau_{\ell})\setminus \Lambda^{+}(\tau_0) = \{u\} \biguplus \left(\Lambda^{+}(\tau_{\ell})\setminus \Lambda^{+}(\tau_1)\right)$.
    Therefore, we have 
    \begin{align*}
    H(P)&= \one{f(\tau_{\ell})=\True}\cdot 
    \chi(\tau_{\ell},\tau_{0})
    =\chi(\tau_{\ell},\tau_{0})\\
    &=(2-q_u\theta_u)\chi(\tau_{\ell},\tau_{1})=(2-q_u\theta_u)H((\tau_1,\dots,\tau_{\ell})).
    \end{align*}
\end{enumerate}

By the law of total expectation, we have
\begin{equation}
\begin{aligned}\label{eq-e-sum-j}
\E{H(\pth(\sigma))}
=&\sum\limits_{x\in \qus{u}}\left(\Pr{\sigma_1=\sigma_{u\gets x}}\cdot \E{H(\pth(\sigma))\mid \sigma_1=\sigma_{u\gets x}}\right)\\
=&\sum\limits_{x\in \qus{u}}\left(\induceddist{\sigma}{u}(x)\cdot (2-q_u\theta_u)\E{H(\pth(\sigma_{u\gets x}))}\right)\\
=&(1-q_u\theta_u)\E{H(\pth(\sigma_{u\gets \star}))}+\sum\limits_{x\in Q_u}\left(\mu_{u}^{\sigma}(x)\E{H(\pth(\sigma_{u\gets x}))}\right),
\end{aligned}
\end{equation}
where the second equality is by \eqref{eq-x1-yi} and \eqref{eq:weightfact}.
Note that by \Cref{RCT-2-b} in \Cref{RCTdef}, for each $x\in \qus{u}$, the subtree in $T_{\sigma}$ rooted by $\sigma_{u\gets x}$ is precisely the $T_{X}$ in the RCT  $\+{T}_X=(T_X,\rho_X)$ rooted at $X=\sigma_{u\gets x}$.  Thus, by induction hypothesis, we have
\begin{equation*}
    \forall x\in \qus{u},\E{H(\pth(\sigma_{u\gets x}))} = \lambda(\+T_{\sigma_{u\gets x}})
\end{equation*}
Combining with \eqref{eq-e-sum-j}, we have
\begin{equation*}
\begin{aligned}
\E{H(\pth(\sigma))}=(1-q_u\theta_u)\lambda(\+T_{\sigma_{u\gets \star}})+\sum\limits_{x\in Q_u}\mu_{u}^{\sigma}(x)\lambda(\+T_{\sigma_{u\gets x}})
=\lambda(\+T_{\sigma}),
\end{aligned}
\end{equation*}
where the last equality is by \Cref{fact-RCT}.
\end{proof}

\subsection{Correctness of the counting algorithm}

In this subsection we bound the total variation distance between the distribution returned by \Cref{Alg:calc} and the true marginal distribution by the upper bound function $\digamma(\cdot)$. The whole subsection will devote to proving the following proposition.

\begin{proposition}\label{tvdbound2}

For any input $(\Phi,\sigma,v)$ satisfying \Cref{inputcondition-calc},
it holds that
\[
\tv(\xi,\mu^{\sigma}_{v})\leq \digamma(\sigma), 
\]
where $\xi$ is the distribution returned by $\calc{}(\Phi,\sigma,v)$.
\end{proposition}

\subsubsection{Generalized \texorpdfstring{$\{2,3\}$}{2,3}-tree witness for truncation}

Given $\sigma\in \qs$ and generalized $\{2,3\}$-tree $T=U\circ E$ in $H_{\Phi}$, let $\pth(\sigma) =
(\sigma_0,\sigma_1,\ldots,\sigma_{\ell})$. 
Define the event $\+{E}^{\sigma}_T$ as
\begin{equation}\label{eq-WTevent}
     \+{E}^{\sigma}_T: U=\vst{\sigma_{\ell}}\land E\subseteq \csfrozen{\sigma_\ell}.
\end{equation}
Let $\sigma\in\+{Q}^*$ be a partial assignment such that only one variable $v\in V$ has $\sigma(v) = \star$. The following lemma shows that if the path $\pth(\sigma) =
(\sigma_0,\sigma_1,\ldots,\sigma_{\ell})$ generated from such $\sigma$ gets truncated at \Cref{rp-1} of \Cref{pathdef} for satisfying $f(\sigma_{\ell})=\True$, then
there must be a large generalized $\{2,3\}$-tree in $H_{\Phi}$.  Its proof is deferred to \Cref{sec:proof of WTsize}.

\begin{lemma}\label{WTsize}
Let $\sigma\in\qs$ 
be a partial assignment with exactly one variable $v\in V$ having $\sigma(v) = \star$ and $\pth(\sigma) =
(\sigma_0,\sigma_1,\ldots,\sigma_{\ell})$. Suppose $f(\sigma_{\ell})=\True$, then there exists a generalized $\{2,3\}$-tree $T=U\circ E$ in $H_{\Phi}$ with some auxiliary tree rooted at $v$ satisfying $$ L\leq \abs{U} + \Delta\cdot \abs{E}\leq L\Delta$$ such that $\+{E}^{\sigma}_T$ happens.
\end{lemma}

\subsubsection{Probability bounds for generalized \texorpdfstring{$\{2,3\}$}{2,3}-tree witness }

Recall the function $H(\cdot)$ defined in \eqref{eq-pathweightdef} and the event $\+{E}^{\sigma}_T$ defined in \eqref{eq-WTevent}.  A crucial lemma we will show is given as follows, which gives a probability bound for certain generalized $\{2,3\}$-tree witness in $H_{\Phi}$.

\begin{lemma}\label{WTprobcor}
Let $\sigma\in\qs$ 
be a partial assignment with exactly one variable $v\in V$ having $\sigma(v) = \star$ satisfying $\mathbb{P}[\neg c\mid \sigma]\leq \pprime q$ for all $c\in \mathcal{C}$. Let $T$ be any generalized $\{2,3\}$-tree in $H_{\Phi}$, then we have
 $$\Pr{\+{E}^{\sigma}_T}\cdot \E{H(\pth(\sigma))\mid \+{E}^{\sigma}_T}\leq F(\sigma,T\setminus \set{v}).$$
\end{lemma}

For any constraint $c\in \+{C}$ and partial assignment  $\sigma\in \qs$, we define
$$Z(\sigma,c)\defeq\abs{\var{c}\setminus \Lambda(\sigma) }. $$
as the number of unassigned variables in $\var{c}$.

For any generalized $\{2,3\}$-tree $T$ in $H_{\Phi}$ and partial assignment $\sigma\in \qs$,
we further define
\begin{equation}\label{eq-definition-g}
    g(\sigma,T)\triangleq \prod\limits_{v\in U\setminus \vst{\sigma}}\left(1-q_v\theta_v\right)\prod\limits_{c\in E}\left(\pprime^{-1}\mathbb{P}[\neg c\mid \sigma](1+\eta)^{Z(\sigma,c)}\right).
\end{equation}

To prove \Cref{WTprobcor}, it is sufficient to show the following.

\begin{lemma}\label{WTprob}
Let $\sigma\in\qs$ 
be a partial assignment satisfying $\mathbb{P}[\neg c\mid \sigma]\leq \pprime q$ for all $c\in \mathcal{C}$,
and let $\pth(\sigma) =
(\sigma_0,\sigma_1,\ldots,\sigma_{\ell})$. 
Then for any generalized $\{2,3\}$-tree $T=U\circ E$ in $H_{\Phi}$, 
\begin{align}\label{eq-lemma-WTprob}
\Pr{ \+{E}^{\sigma}_T}\cdot \E{H(\pth(\sigma))\mid \+{E}^{\sigma}_T}\leq g(\sigma,T).\end{align}
\end{lemma}
\begin{proof}

We show the lemma by a structural induction on $\pth(\sigma)$. The base case is when $\pth(\sigma)=(\sigma)$. Then we have $\ell(\sigma)=0$ and $\sigma=\sigma_{\ell}$. 
In this case,
$\+{E}^{\sigma}_T$ is the deterministic event $U=\vst{\sigma}\land E\subseteq \csfrozen{\sigma}$.
If $U\neq\vst{\sigma}$ or  $E\not\subseteq \csfrozen{\sigma}$, we have 
$\Pr{\+{E}^{\sigma}_T} =0$ and the lemma is immediate.
Otherwise, we have $U=\vst{\sigma}\land E\subseteq \csfrozen{\sigma}$.
By $\sigma=\sigma_{\ell}$ and \eqref{eq-pathweightdef},
we have $H(\pth(\sigma))\leq 1$.
Thus, we have 
$$\Pr{ \+{E}^{\sigma}_T}\cdot \E{H(\pth(\sigma))\mid \+{E}^{\sigma}_T}\leq H(\pth(\sigma)) \leq 1.$$
Meanwhile, by $E\subseteq \csfrozen{\sigma}$ and \Cref{def:cbad}, we have $E\subseteq \csfrozen{\sigma}\subseteq \cfrozen{\sigma}$. 
Thus, for each $c\in E$,
we have $c$ is $\sigma$-frozen.
Combining with \Cref{definition:frozen-fixed},
we have $\mathbb{P}[\neg c \mid \sigma] > \pprime$.
Combining with $U=\vst{\sigma}$,
we have 
\begin{align*}
    g(\sigma,T)=\prod\limits_{c\in E}\left(\pprime^{-1}\mathbb{P}[\neg c\mid \sigma](1+\eta)^{Z(\sigma,c)}\right)\geq \prod\limits_{c\in E} \left(\pprime^{-1}\pprime(1+\eta)^{Z(\sigma,c)}\right)\geq 1\geq \Pr{ \+{E}^{\sigma}_T}\cdot \E{H(\pth(\sigma))}.
\end{align*}
The base case is proved.

For the induction steps, we assume that $\ell(\sigma)\geq 1$, which by \Cref{rp-1} of \Cref{pathdef}, says that $\nextvar{\sigma}=u\neq \perp$ for some $u\in V$ and $f(\sigma)=\False$. According to \Cref{rp-2} of \Cref{pathdef},
we have 
\begin{align*}
    \forall x\in \qus{u},\quad \Pr{\sigma_{1} = \sigma_{u\gets x}}&=\induceddist{\sigma}{u}(x).
\end{align*}
Thus, by the law of total probability, we have
\begin{equation}
\begin{aligned}\label{WTprob-eq4}
&\Pr{ \+{E}^{\sigma}_T}\cdot \E{H(\pth(\sigma))\mid \+{E}^{\sigma}_T}\\
= &\sum\limits_{x\in \qus{u}}\left(\Pr{\sigma_1=\sigma_{u\gets x}}\Pr{\+{E}^{\sigma}_T\mid \sigma_1=\sigma_{u\gets x}} \E{H(\pth(\sigma))\mid \sigma_1=\sigma_{u\gets x}\land  \+{E}^{\sigma}_T}\right)\\
=&\sum\limits_{x\in \qus{u}}\left(\induceddist{\sigma}{u}(x)\cdot\Pr{\+{E}^{\sigma}_T\mid \sigma_1=\sigma_{u\gets x}}\E{H(\pth(\sigma))\mid \sigma_1=\sigma_{u\gets x}\land  \+{E}^{\sigma}_T}\right)\\
\end{aligned}
\end{equation}
Moreover, by \eqref{eq-pathweightdef} we have
\begin{equation}\label{eq-ehpath-sigma}
\begin{aligned}
   &\E{H(\pth(\sigma))\mid \sigma_1=\sigma_{u\gets x}\land  \+{E}^{\sigma}_T}\\
   =&\E{\one{f(\sigma_{\ell})=\True}\cdot \chi(\sigma_{\ell},\sigma_{0}) \mid \sigma_1=\sigma_{u\gets x}\land  \+{E}^{\sigma}_T}\\
   =&\left(2-q_u\theta_u\right)\E{\one{f(\sigma_{\ell})=\True} \cdot \chi(\sigma_{\ell},\sigma_{1})\mid \sigma_1=\sigma_{u\gets x}\land  \+{E}^{\sigma}_T}
\end{aligned}
\end{equation}
In addition, by the Markov property, given $\sigma_1 = \tau \triangleq \sigma_{u\gets x}$ for each $x\in \qus{u}$, 
the subsequence  $(\sigma_1,\sigma_2,\cdots,\sigma_{\ell})$  is identically  distributed as $\pth(\tau)$.
Thus, we have $\sigma_{\ell}$ is identically  distributed as $\tau_{\ell(\tau)}$.
combining with \eqref{eq-pathweightdef} and \eqref{eq-WTevent}
, we have
\begin{equation*}
    \begin{aligned}
    &\E{\one{f(\sigma_{\ell})=\True} \cdot \chi(\sigma_{\ell},\sigma_{1})\mid \sigma_1=\tau\land  \+{E}^{\sigma}_T}\\
   =&\E{\one{f(\tau_{\ell(\tau)})=\True}\cdot \chi(\tau_{\ell(\tau)},\sigma_{1}) \mid \sigma_1=\tau\land  \+{E}^{\tau}_T}\\
   =&\E{\one{f(\tau_{\ell(\tau)})=\True} \cdot \chi(\tau_{\ell(\tau)},\tau)\mid  \+{E}^{\tau}_T}\\
    =&\E{H(\pth(\tau))\mid  \+{E}^{\tau}_T}.
    \end{aligned}
\end{equation*}
Combining with \eqref{eq-ehpath-sigma}, we have 
\begin{equation}\label{eq-ehpath-sigma-sigmaux}
\begin{aligned}
   \E{H(\pth(\sigma))\mid \sigma_1=\sigma_{u\gets x}\land  \+{E}^{\sigma}_T}
   =&\left(2-q_u\theta_u\right)\E{H(\pth(\sigma_1))\mid \sigma_1=\sigma_{u\gets x}\land  \+{E}^{\sigma}_T}\\
    =&\left(2-q_u\theta_u\right)\E{H(\pth(\sigma_{u\gets x}))\mid  \+{E}^{\sigma_{u\leftarrow x}}_T}.
\end{aligned}
\end{equation}
Recall that given $\sigma_1 = \sigma_{u\gets x}$ for each $x\in \qus{u}$, 
$\sigma_{\ell}$ is identically  distributed as $\tau_{\ell(\tau)}$ where $\tau=\sigma_{u\leftarrow x}$.
Combining with \eqref{eq-WTevent},
we have 
$\Pr{ \+{E}^{\sigma}_T\mid \sigma_1=\sigma_{u\gets x}} = \Pr{ \+{E}^{\sigma_{u\leftarrow x}}_T}$.
Combining with \eqref{WTprob-eq4}
and \eqref{eq-ehpath-sigma-sigmaux}, we have
\begin{equation}
\begin{aligned}\label{eq-esigmatehpath-twopart}
\Pr{ \+{E}^{\sigma}_T}\cdot \E{H(\pth(\sigma))\mid \+{E}^{\sigma}_T}
=\sum\limits_{x\in \qus{u}}\left(\left(2-q_u\theta_u\right)\induceddist{\sigma}{u}(x)\cdot\Pr{\+{E}^{\sigma_{u\leftarrow x}}_T}\cdot \E{H(\pth(\sigma_{u\gets x})\mid \+{E}^{\sigma_{u\gets x}}_T}\right).
\end{aligned}
\end{equation}

We then show the induction step for two cases respectively, namely the case when $u\in U$ and the case when $u\notin U$.
At first we assume $u\in U$.
Given $x\in Q_u$ and $\tau = \sigma_{u\leftarrow x}$, by $\tau(u)=x$,
we also have $\tau_{\ell(\tau)}(u)=x\neq \star$.
Thus $u\not\in \vst{\tau_{\ell(\tau)}}$.
Combining with $u\in U$,
we have $U \neq \vst{\tau_{\ell(\tau)}}$.
Combining with \eqref{eq-WTevent},
we have $\+{E}^{\tau}_T$ does not happen.
In summary, for each $x\in Q_u$, $\+{E}^{\sigma_{u\leftarrow x}}_T$ does not happen 
and $\Pr{\+{E}^{\sigma_{u\leftarrow x}}_T} = 0$.
Combining with \eqref{eq-esigmatehpath-twopart}, we have 
\begin{equation}\label{eq-pesigmattimesehpath-uinvvart}
\begin{aligned}
    \Pr{ \+{E}^{\sigma}_T}\cdot \E{H(\pth(\sigma))\mid \+{E}^{\sigma}_T}=(1-q_u\theta_u)\cdot \Pr{ \+{E}^{\sigma_{u\leftarrow \star}}_T}\cdot \E{H(\pth(\sigma_{u\gets \star}))\mid  \+{E}^{\sigma_{u\gets \star}}_T}.
\end{aligned}
\end{equation}
In addition, by $\sigma\in\+{Q}^*$ is a partial assignment satisfying $\mathbb{P}[\neg c\mid \sigma]\leq \pprime q$ for all $c\in \mathcal{C}$,
one can also verify
$\mathbb{P}[\neg c\mid \sigma_{u\leftarrow x}]\leq \pprime q$ for all $c\in \mathcal{C}$ and $x\in \qus{u}$
by a similar argument as  \Cref{lemma:general-invariant}.
Thus by the induction hypothesis, for each $x\in \qus{u}$ we have
\begin{equation}\label{eq-pesigmauxttimesehpath-induct}
\begin{aligned}
\Pr{ \+{E}^{\sigma_{u\gets x}}_T}\cdot \E{H(\pth(\sigma_{u\gets x}))\mid \+{E}^{\sigma_{u\gets x}}_T}
\leq g(\sigma_{u\gets x},T).
\end{aligned}
\end{equation}
Combining with \eqref{eq-pesigmattimesehpath-uinvvart},
we have
\begin{equation}\label{eq-presimgatehpaht-uinvvart}
\begin{aligned}
&\Pr{ \+{E}^{\sigma}_T}\cdot \E{H(\pth(\sigma))\mid \+{E}^{\sigma}_T}
\\ \leq&(1-q_u\theta_u)\cdot g(\sigma_{u\gets \star},T)\\
  = &(1-q_u\theta_u) \prod\limits_{v\in U\setminus \vst{\sigma_{u\gets \star}}}\left(1-q_v\theta_v\right)\prod\limits_{c\in E}\left(\pprime^{-1}\mathbb{P}[\neg c\mid \sigma_{u\gets \star}](1+\eta)^{Z(\sigma_{u\gets \star},c)}\right)\\
    =& (1-q_u\theta_u) \prod\limits_{v\in U\setminus \vst{\sigma_{u\gets \star}}}\left(1-q_v\theta_v\right)\prod\limits_{c\in E}\left(\pprime^{-1}\mathbb{P}[\neg c\mid \sigma](1+\eta)^{Z(\sigma,c)}\right),
\end{aligned}
\end{equation}
where the last equality is by $\mathbb{P}[\neg c\mid \sigma_{u\gets \star}] = \mathbb{P}[\neg c\mid \sigma]$ and
$Z(\sigma_{u\gets \star},c) = Z(\sigma,c)$ for each $\sigma$ and $c$.
In addition, by $u=\nextvar{\sigma}$, we have $\sigma(u) = \hollowstar\neq \star$.
Thus, $u\not\in \vst{\sigma}$.
Meanwhile, by $\sigma_{u\leftarrow \star}(u) = \star$, we have $u\in \vst{u\leftarrow \star}$.
Thus, $\vst{\sigma_{u\leftarrow \star}} = \vst{\sigma}\biguplus \{u\}$.
Combining with $u\in U$,
we have $U\setminus \vst{\sigma} = \left(U\setminus \vst{\sigma_{u\gets \star}}\right)\biguplus \{u\}$.
Therefore, 
$$ (1-q_u\theta_u) \prod\limits_{v\in U\setminus \vst{\sigma_{u\gets \star}}}\left(1-q_v\theta_v\right)= \prod\limits_{v\in U\setminus \vst{\sigma}}\left(1-q_v\theta_v\right).$$
Combining with \eqref{eq-presimgatehpaht-uinvvart},
\eqref{eq-lemma-WTprob} is immediate.
This finishes the induction step for the case when $u\in U$.

In the following, we assume $u\not\in U$. Given $\tau = \sigma_{u\leftarrow \star}$,
by $\tau(u)=\star$, we also have $\tau_{\ell(\tau)}(u)=\star$.
Thus $u \in \vst{\tau_{\ell(\tau)}}$.
Combining with $u\not\in U$,
we have $U \neq \vst{\tau_{\ell(\tau)}}$.
Combining with \eqref{eq-WTevent},
we have $\+{E}^{\tau}_T = \+{E}^{\sigma_{u\leftarrow \star}}_T$ does not happen
and $\Pr{\+{E}^{\sigma_{u\leftarrow \star}}_T} = 0$.
Combining with \eqref{eq-esigmatehpath-twopart}, we have 
\begin{equation*}
\begin{aligned}
    \Pr{ \+{E}^{\sigma}_T}\cdot \E{H(\pth(\sigma))\mid \+{E}^{\sigma}_T}=\sum\limits_{x\in Q_u}\left(\mu_{u}^{\sigma}(x)\cdot\Pr{\+{E}^{\sigma_{u\leftarrow x}}_T}\cdot \E{H(\pth(\sigma_{u\gets x})\mid \+{E}^{\sigma_{u\gets x}}_T}\right).
\end{aligned}
\end{equation*}

Combining with \eqref{eq-pesigmauxttimesehpath-induct}, we have
\begin{equation}\label{eq-presimgatehpaht-unotinvvart}
\begin{aligned}
&\Pr{ \+{E}^{\sigma}_T}\cdot \E{H(\pth(\sigma))\mid \+{E}^{\sigma}_T}
 \leq \sum\limits_{x\in Q_u}\left(\mu_{u}^{\sigma}(x)\cdot g(\sigma_{u\gets x},T)\right)\\
  = &\sum\limits_{x\in Q_u}\left(\mu^{\sigma}_{u}(x)\prod\limits_{v\in U\setminus \vst{\sigma_{u\gets x}}}\left(1-q_v\theta_v\right)\prod\limits_{c\in E}\left(\pprime^{-1}\mathbb{P}[\neg c\mid \sigma_{u\gets x}](1+\eta)^{Z(\sigma_{u\gets x},c)}\right)\right).
\end{aligned}
\end{equation}
In addition, by $T$ is a generalized $\{2,3\}$-tree and \Cref{WT-1} of \Cref{WTdef}, we have $\var{c}\cap \var{c'}\neq \emptyset$ for any different $c,c'\in E$.
Thus, there exists at most one unique constraint $c_0\in E$ such that $u\in \var{c_0}$.
Let $S = E\setminus \{c_0\}$ if $u\in \var{E}$ and $S = E$ otherwise.
Thus for each $c\in S$, we have $u\not\in \var{c}$.
Then
$\mathbb{P}[\neg c\mid \sigma_{u\gets x}]=\mathbb{P}[\neg c\mid \sigma]$ and $Z(\sigma_{u\gets x},c) = Z(\sigma,c)$ for each $x\in Q_u$.
Therefore,
\begin{equation}\label{eq-sum-musigmauxtimespivcons}
\begin{aligned}
&\sum\limits_{x\in Q_u}\left(\mu^{\sigma}_{u}(x)\prod\limits_{c\in E}\left(\mathbb{P}[\neg c\mid \sigma_{u\gets x}](1+\eta)^{Z(\sigma_{u\gets x},c)}\right)\right)\\
= & \prod\limits_{c\in S}\left(\mathbb{P}[\neg c\mid \sigma_{u\gets x}](1+\eta)^{Z(\sigma_{u\gets x},c)}\right)\sum\limits_{x\in Q_u}\left(\mu^{\sigma}_{u}(x)\prod\limits_{c\in E\setminus S}\left(\mathbb{P}[\neg c\mid \sigma_{u\gets x}](1+\eta)^{Z(\sigma_{u\gets x},c)}\right)\right)\\
= &\prod\limits_{c\in S}\left(\mathbb{P}[\neg c\mid \sigma](1+\eta)^{Z(\sigma,c)}\right)\sum\limits_{x\in Q_u}\left(\mu^{\sigma}_{u}(x)\prod\limits_{c\in E\setminus S}\left(\mathbb{P}[\neg c\mid \sigma_{u\gets x}](1+\eta)^{Z(\sigma_{u\gets x},c)}\right)\right).
\end{aligned}    
\end{equation}
In addition, by \Cref{generaluniformity} and the assumption that $\mathbb{P}[\neg c\mid \sigma]\leq \pprime q$ for all $c\in \mathcal{C}$, we have for each $x\in Q_u$,
$\mu^{\sigma}_u(x)\leq q_u^{-1}(1+\eta).$
Therefore,
\begin{align*}
\sum\limits_{x\in Q_u}\left(\mu^{\sigma}_{u}(x)\cdot \mathbb{P}[\neg c_0\mid \sigma_{u\gets x}]\right)\leq (1+\eta)\cdot q_u^{-1}\sum\limits_{x\in Q_u} \mathbb{P}[\neg c_0\mid \sigma_{u\gets x}]
= (1+\eta)\cdot \mathbb{P}[\neg c_0\mid \sigma].
\end{align*}
Thus, we have 
\begin{equation}\label{eq-sum-muponeplusetaz}
\begin{aligned}
    \sum_{x\in Q_u}\left(\mu^{\sigma}_u(x) \mathbb{P}[\neg c_0\mid \sigma_{u\gets x}](1+\eta)^{Z(\sigma_{u\gets x},c_0)}\right)=&\sum_{x\in Q_u}\left(\mu^{\sigma}_u(x) \mathbb{P}[\neg c_0\mid \sigma_{u\gets x}](1+\eta)^{Z(\sigma,c_0)-1}\right)\\
    \leq & \mathbb{P}[\neg c_0\mid \sigma](1+\eta)^{Z(\sigma,c_0)}.
\end{aligned} 
\end{equation}
Moreover, we always have 
\begin{align}\label{eq-sum-musigmauxtimespivconstsetminuss}
    \sum\limits_{x\in Q_u}\left(\mu^{\sigma}_{u}(x)\prod\limits_{c\in E\setminus S}\left(\mathbb{P}[\neg c\mid \sigma_{u\gets x}](1+\eta)^{Z(\sigma_{u\gets x},c)}\right)\right) \leq \prod\limits_{c\in E\setminus S}\left(\mathbb{P}[\neg c\mid \sigma](1+\eta)^{Z(\sigma,c)}\right),
\end{align}
where we assume that a product over an empty set is 1.
Because $E\setminus S$ is either $\{c_0\}$ or an empty set. 
If $E\setminus S = \{c_0\}$,
\eqref{eq-sum-musigmauxtimespivconstsetminuss} is immediate by \eqref{eq-sum-muponeplusetaz}.
Otherwise, $E\setminus S = \emptyset$.
Then both sides of \eqref{eq-sum-musigmauxtimespivconstsetminuss} are equal to 1.
Combining \eqref{eq-sum-musigmauxtimespivcons} with \eqref{eq-sum-musigmauxtimespivconstsetminuss},
we have 
\begin{equation}\label{eq-sum-musigmauxtimespivconst-simplify}
\begin{aligned}
\sum\limits_{x\in Q_u}\left(\mu^{\sigma}_{u}(x)\prod\limits_{c\in E}\left(\mathbb{P}[\neg c\mid \sigma_{u\gets x}](1+\eta)^{Z(\sigma_{u\gets x},c)}\right)\right)\leq \prod\limits_{c\in E}\left(\mathbb{P}[\neg c\mid \sigma](1+\eta)^{Z(\sigma,c)}\right).
\end{aligned}    
\end{equation}
Moreover, by $u=\nextvar{\sigma}$, we have $\sigma(u) = \hollowstar\neq \star$.
Thus, $u\not\in \vst{\sigma}$.
Meanwhile, by $\sigma_{u\leftarrow x}(u) = x\neq \star$, we also have $u\not\in \vst{u\leftarrow x}$ for each $x\in Q_u$.
Thus, $U\setminus \vst{\sigma} = U\setminus \vst{\sigma_{u\gets x}}$.
Combining with \eqref{eq-presimgatehpaht-unotinvvart} and \eqref{eq-sum-musigmauxtimespivconst-simplify},
\eqref{eq-lemma-WTprob} is immediate.
This finishes the induction step for the  case when $u\notin U$.
The lemma is proved.
\end{proof}

Combining \Cref{WTprob} with the condition that $\sigma\in\qs$ is a partial assignment with exactly one variable $v\in V$ having $\sigma(v) = \star$ and comparing \eqref{eq-def-fsimgat} with \eqref{eq-definition-g}, \Cref{WTprobcor} is proved.

We are now ready to prove \Cref{tvdbound2}.

\begin{proof}[Proof of \Cref{tvdbound2}]
Let $\tau=\sigma_{v\gets \star}$.
If $H(\pth(\tau))>0$, we have $\one{f(\tau_{\ell})=\True}\cdot \chi(\tau_{\ell},\tau_0)>0$ by \eqref{eq-pathweightdef}.
Combining with $\chi(\tau_{\ell},\tau_0)\geq 0$,
we have $f(\tau_{\ell})$ is true if $H(\pth(\tau))>0$.
Therefore, by \Cref{def:truncate-refined} we have $\abs{\vst{\tau_{\ell}}}+\Delta\cdot \abs{\csfrozen{\tau_{\ell}}}\geq L\Delta$.
Combining with \Cref{WTsize} we have there always exists a generalized $\{2,3\}$-tree $T=U\circ E$ in $H_{\Phi}$ with some auxiliary tree rooted at $v$ such that $L\leq \Delta\cdot \abs{E}+\abs{U}\leq L\Delta$ and $\+{E}^{\tau}_{T}$ happens. Let $\+{U}$ denote the set $\{Y\in \+{T}^{t}_{v}: L\leq t\leq L\Delta\land v\in V\}$, then we have $T\in \+{U}$.
In summary, if $H(\pth(\tau))>0$, there exists some $T\in \+{U}$ such that $\+{E}^{\tau}_{T}$ happens.
Therefore by the law of total expectation and the nonnegativity of $H(\pth(\tau))$,
we have 
\begin{align}\label{eq-totallaw-ehpath}
    \E{H(\pth(\tau))} \leq \sum\limits_{T\in \+{U}}\Pr{\+{E}^{\tau}_T}\cdot\E{H(\pth(\tau))\mid \+{E}^{\tau}_T}.
\end{align}
Thus, we have 
\begin{align*}
    &\quad\tv(\xi,\mu^{\sigma}_{v})\\
(\text{by Lemmas \ref{boundedRCTtvdcor} and \ref{pathprop}}) \quad  &\leq \E{H(\pth(\tau))}\\
(\text{by \eqref{eq-totallaw-ehpath}}) \quad &\leq \sum\limits_{T\in \+{U}}\Pr{\+{E}^{\tau}_T}\cdot\E{H(\pth(\tau))\mid \+{E}^{\tau}_T}\\
  \text{(by  \Cref{WTprobcor})}\quad   &\leq \sum\limits_{T\in \+{U}}F(\tau,T\setminus \set{v})\\
  &\leq \sum\limits_{i=L}^{L\Delta}\sum\limits_{v\in V}\sum\limits_{T\in \+{T}^{i}_{v}}F(\tau,T\setminus \{v\} )\\
   \text{(by \eqref{eq-definition-F-2})}\quad &= \digamma(\tau)\\
   &=\digamma(\sigma).
\end{align*}
\end{proof}

\subsection{Efficiency of the counting algorithm}
We then show the efficiency of the algorithm, given that the upper bound function $\digamma(\cdot)$ is small. Recall the definition of $X^n$ in \Cref{def-pas-cmain}. We will show two crucial propositions, namely \Cref{lem-cvxl-in-cbad2-cor} and \Cref{lem-exeef2}. \Cref{lem-cvxl-in-cbad2-cor} bounds the running time of the marginal approximator subroutine called within the main counting algorithm, and \Cref{lem-exeef2}
 the running time on the exhaustive enumeration part in the main counting algorithm if $\digamma(X^n)$ is small.
\begin{proposition}\label{lem-cvxl-in-cbad2-cor}
For any $(\Phi,\sigma,v)$ satisfying \Cref{inputcondition-calc}, let $\tma{(\Phi,\sigma,v)}$ denote the running time of  $\calc{}(\Phi,\sigma,v)$. Then 
$\tma(\Phi,\sigma,v)\leq \poly(n,q^{k L\Delta^2 })$.
\end{proposition}

\begin{proposition}\label{lem-exeef2}
 If $\digamma(X^n)<1$, then
$ \texe{}(\Phi,X^n)=\poly(n,q^{kL\Delta^2 })$,
where $\texe{}(\Phi,X^n)$ denotes the running time of the exhaustive enumeration in the main counting algorithm. 
\end{proposition}

For each $\sigma\in\qs$ and $v\in V^{\sigma}$, 
$H_v^{\sigma}=(V_v^\sigma,\+{C}_v^{\sigma})$ denotes the connected component in $H^{\sigma}$ that contains the vertex/variable $v$, 
where $H^{\sigma}$ is the hypergraph representation for the CSP formula $\Phi^\sigma$ obtained from the simplification of $\Phi$ under $\sigma$. 

We further stipulate that 
$H_v^{\sigma}=(V_v^\sigma,\+{C}_v^{\sigma})=(\emptyset,\emptyset)$ is the empty hypergraph when $v\in \Lambda(\sigma)$.

Recall that $\+{T}_{\sigma}=(T_{\sigma},\rho_{\sigma})$ is the RCT rooted at $\sigma$. To show \Cref{lem-cvxl-in-cbad2-cor}, we need the following lemma, which bounds the efficiency of the subroutine \recalc{}.

\begin{lemma}\label{RCTtime}
Let  $(\Phi,\sigma,v)$ be the input to $\recalc{}$ (\Cref{Alg:recalc}) satisfying \Cref{inputcondition-recalc},
and let $\tra{}(\Phi,\sigma,v)$ denote the running time of $\recalc{}{(\Phi,\sigma,v)}$. 
It holds that 
\[
\tra{}(\Phi,\sigma,v)\le \abs{\+T_{\sigma}} \cdot \poly(n,\Delta,q^k) +  O\left(\sum_{\tau\in \+{L}_g(\+{T}_{\sigma})}\left( (k\abs{ \+{C}_v^\tau}+q\abs{ V_v^\tau})\cdot q^{\abs{ V_v^\tau}}\right)\right).
\]
\end{lemma}

The following lemma will be used in the proof of \Cref{RCTtime}. Its proof is similar to that of \Cref{lemma:invariant-counting}-(\ref{item:invariant-recalc}) and omitted here.

\begin{lemma}\label{lem-rct-satisfy-condition-recalc}
Let  $(\Phi,\sigma,v)$ satisfy \Cref{inputcondition-recalc}.
Then for each node $\tau$ in $T_{\sigma}$, $(\Phi,\tau,v)$ also satisfies \Cref{inputcondition-recalc}.
\end{lemma}

By \Cref{prop:theta-recalc-lower-bound} and \Cref{lem-rct-satisfy-condition-recalc} , we have the following lemma.

\begin{lemma}\label{lem-pathinrct-largerprob-inpath}
Let  $(\Phi,\sigma,v)$ satisfy \Cref{inputcondition-recalc}.
Recall that $\+{T}_{\sigma}=(T_{\sigma},\rho_{\sigma})$ is the RCT rooted at $\sigma$.
For any leaf node $\tau$ in $T_{\sigma}$,
let $P$ be the path from $\sigma$ to $\tau$ in $T_{\sigma}$.
Then $\Pr{\pth(\sigma)=P}>0$.

\end{lemma}
\begin{proof}
Let $(\tau_0=\sigma,\tau_1,\tau_2,\tau_3,\cdots,\tau_{r}=\tau)$ be the path from $\sigma$ to $\tau$ in $T_{\sigma}$.
Given $0\leq i< r$,
let $u_i = \nextvar{\tau_i}$.
Then we have $f(\tau_i)\neq \True$ and $u_i\neq \perp$.
Otherwise, by \Cref{RCT-2-a} of \Cref{RCTdef},
we have $\tau_i$ is a leaf of $T_\sigma$, which is contradictory with $i<r$.
By \Cref{RCT-2-b} of \Cref{RCTdef},
we also have $\tau_{i+1}\in \{(\tau_i)_{u_i\gets x}\mid x\in \qus{u_i}\}$.
In addition, by \Cref{lem-rct-satisfy-condition-recalc}, we have 
$(\Phi,\tau_i,v)$
    satisfies \Cref{inputcondition-recalc}.
Combining with \Cref{prop:theta-recalc-lower-bound}, we have  $\mu_{u_i}^{\tau_i}(x)\geq \theta_{u_i}$ for each $x\in \qus{u_i}$.
Thus, combining $f(\tau_i)\neq \True$, $u_i=\nextvar{\tau_i}\neq \perp$, $\tau_{i+1}\in \{(\tau_i)_{u_i\gets x}\mid x\in \qus{u_i}$ with \Cref{pathdef}, we have
\begin{align*}
\Pr{\sigma_{i+1}=\tau_{i+1} \mid \sigma_0=\tau_0,\cdots,\sigma_i=\tau_i} &\geq  (2-q_{u_i} \theta_{u_i})^{-1}\min\left\{1-q_{u_i} \theta_{u_i},
    \min_{x\in Q_{u_i}}\left\{\mu_{u_i}^{\tau_i}(x)\right\}\right\}\\
    &\geq  (2-q_{u_i} \theta_{u_i})^{-1}\min\left\{1-q_{u_i} \theta_{u_i},
    \theta_{u_i}\right\}>0,
\end{align*}
where the second inequality is by that $\mu_{u_i}^{\tau_i}(x)\geq \theta_{u_i}$ for each $x\in \qus{u_i}$.
Thus, the lemma is immediate by the chain rule.
\end{proof}

By~\cite[Proposition 6.28]{he2022sampling}, the following lemma is immediate.  

\begin{lemma}\label{lemma:efficiency-var}
For any $\sigma \in \+{Q}^\ast$, 
$\nextvar{\sigma}$ and $f(\sigma)$ can be computed in  $\poly(n,\Delta,q^k)$ cost.
\end{lemma}

Now we can prove \Cref{RCTtime}.

\begin{proof}[Proof of \Cref{RCTtime}]
We prove this lemma by an induction on the structure of RCT. The base case is when $T_\sigma$ is just a single root,
in which case $\nextvar{\sigma}= \perp$ or $f(\sigma)=\True$.
If $f(\sigma)=\True$, 
the condition in \Cref{Line-recalc-truncate} of
$\recalc{}(\Phi,\sigma,v)$ is not satisfied and Lines \ref{Line-recalc-else}-\ref{Line-recalc-return} are omitted.
Thus, we have
$\tra{}(\Phi,\sigma,v)\le  \poly(n,\Delta,q^k)$.
Otherwise, $f(\sigma)=\False$ and
$\nextvar{\sigma}= \perp$.
By $T_\sigma$ is just a single root and $f(\sigma)=\False$, we have $\sigma\in \+{L}_g(\+{T}_{\sigma})$.
By $\nextvar{\sigma}= \perp$, we have
the condition in \Cref{Line-recalc-loop} of
$\recalc{}(\Phi,\sigma,v)$ is not satisfied and Lines \ref{Line-recalc-else}-\ref{Line-recalc-calcmu} are omitted.
Combining with \Cref{lemma:efficiency-var}, we have \begin{align*}
\tra{}(\Phi,\sigma,v)& \leq\poly(n,\Delta,q^k) + O\left( (k\abs{ \+{C}_v^\tau}+q\abs{ V_v^\tau})\cdot q^{\abs{ V_v^\tau}}\right)\\
&\leq \poly(n,\Delta,q^k) + O\left(\sum_{\tau\in \+{L}_g(\+{T}_{\sigma})}\left( (k\abs{ \+{C}_v^\tau}+q\abs{ V_v^\tau})\cdot q^{\abs{ V_v^\tau}}\right)\right),
\end{align*}
where the first inequality is by  $(\Phi,\sigma,v)$ satisfies \Cref{inputcondition-recalc} and the standard guarantee on the running time of exhaustive enumeration, and the last inequality is by $\sigma\in \+{L}_g(\+{T}_{\sigma})$.
The base case is proved.

For the induction step, we assume that $T_\sigma$ is a tree of depth $>0$. 
Thus by \Cref{RCTdef}, 
$f(\sigma)=\False$ and $\nextvar{\sigma}=u\neq \perp$ for some $u\in V$. 
Thus the condition in \Cref{Line-recalc-truncate} of
$\recalc{}(\Phi,\sigma,v)$ is not satisfied and the condition in \Cref{Line-recalc-loop} is satisfied.
According to Lines \ref{Line-recalc-mu}-\ref{Line-recalc-enu}, one can verify that 
\begin{equation}
\begin{aligned}\label{eq-tra-upperbound} 
&\tra{}(\Phi,\sigma,v) \leq \poly(n,\Delta,q^k) + \tra{}(\Phi,\sigma_{u\leftarrow \star},u) +  \sum_{x\in Q_{u}}
\tra{}(\Phi,\sigma_{u\leftarrow x},v).
\end{aligned}
\end{equation}
Let $S\triangleq \{\sigma_{u\leftarrow x}:x\in Q_u \cup \{\star\}\}$.
By the induction hypothesis,
we have 
\begin{equation}\label{eq-tra-sigma-sum-simgauleftarrowx} 
\begin{aligned}
\tra{}(\Phi,\sigma_{u\leftarrow \star},u) +  \sum_{x\in Q_{u}}
\tra{}(\Phi,\sigma_{u\leftarrow x},v)
\le \sum_{\tau \in S}
\left(\abs{\+T_{\tau}} \cdot \poly(n,\Delta,q^k) +  O\left(\sum_{X\in \+{L}_g(\+{T}_{\tau})}\left( (k\abs{ \+{C}_v^X}+q\abs{ V_v^X})\cdot q^{\abs{ V_v^X}}\right)\right)\right).
\end{aligned}
\end{equation}
Meanwhile, by \Cref{RCTdef}, one can verify that $T_{\sigma}$ is a tree consisting of a root $\sigma$ and  $q_{v}+1$ subtrees $T_{\tau}$ where $\tau\in S$. 
Thus, we have 
$1+\sum_{\tau\in S}\abs{\+T_{\tau}} = \abs{\+T_{\sigma}}$
and 
$$ \bigcup_{\tau\in S} \+{L}_g(\+{T}_{\tau}) = \+{L}_g(\+{T}_{\sigma}).$$
Moreover, it is easy to verify that 
$\+{L}_g(\+{T}_{\tau})\cap \+{L}_g(\+{T}_{\tau'}) = \emptyset$ for different 
$\tau,\tau'\in S$.
Combing with (\ref{eq-tra-sigma-sum-simgauleftarrowx}), we have
\begin{equation*}
\begin{aligned}
\tra{}(\Phi,\sigma_{u\leftarrow \star},u) +  \sum_{x\in Q_{u}}
\tra{}(\Phi,\sigma_{u\leftarrow x},v)
\le 
\left(\abs{\+T_{\sigma}}-1\right) \cdot \poly(n,\Delta,q^k) +  O\left(\sum_{\tau\in \+{L}_g(\+{T}_{\sigma})}\left( (k\abs{ \+{C}_v^\tau}+q\abs{ V_v^\tau})\cdot q^{\abs{ V_v^\tau}}\right)\right).
\end{aligned}
\end{equation*}
Combining with (\ref{eq-tra-upperbound}), we have
\begin{equation*}
\begin{aligned}
&\tra{}(\Phi,\sigma,v) \leq &\abs{\+T_{\sigma}} \cdot \poly(n,\Delta,q^k) +  O\left(\sum_{\tau\in \+{L}_g(\+{T}_{\sigma})}\left( (k\abs{ \+{C}_v^\tau}+q\abs{ V_v^\tau})\cdot q^{\abs{ V_v^\tau}}\right)\right),
\end{aligned}
\end{equation*}
which finishes the proof of the induction step. Then the lemma is immediate.
\end{proof}

Let $\sigma\in\+{Q}^*$ be a partial assignment such that only one variable $v\in V$ has $\sigma(v) = \star$.
The following lemma shows that the path  $\pth(\sigma) =
(\sigma_0,\sigma_1,\ldots,\sigma_{\ell})$ generated from such $\sigma$ cannot be too long. 
Its proof is deferred to \Cref{sec:proof of shortpath2}.

\begin{lemma}\label{shortpath2}
Let $\sigma\in\+{Q}^*$ 
be a partial assignment with exactly one variable $v\in V$ having $\sigma(v) = \star$, 
and let $\pth(\sigma) =
(\sigma_0,\sigma_1,\ldots,\sigma_{\ell})$.
Then it always holds that
\[
\ell\le kL\Delta^2. 
\]
\end{lemma}

The following lemma further relates the size of $\+C^{\sigma_{\ell}}_v$ with the sizes of $\vst{\sigma_{\ell}}$ and $\csfrozen{\sigma_{\ell}}$.
It is formally proved in \Cref{sec:proof of lem-cvxl-in-cbad2}. 

\begin{lemma}\label{lem-cvxl-in-cbad2}
Let $\sigma\in\qs$ 
be a partial assignment with exactly one variable $v\in V$ having $\sigma(v) = \star$. Let $ \pth(\sigma)=(\sigma_0,\sigma_1,\cdots,\sigma_{\ell}) $. If $\nextvar{\sigma_{\ell}}=\perp$, we have $\abs{\+C^{\sigma_{\ell}}_v} \leq  \Delta \cdot \abs{\csfrozen{\sigma_{\ell}}}+\Delta\cdot \abs{\vst{\sigma_{\ell}}}\leq   L\Delta^2$.
\end{lemma}

Combining Lemmas \ref{lem-pathinrct-largerprob-inpath}, \ref{shortpath2} and \ref{lem-cvxl-in-cbad2} , we have the following corollary.
\begin{corollary}\label{cor-depth-sizectauv-rct}
Let  $(\Phi,\sigma,v)$ satisfy \Cref{inputcondition-recalc}
and $v$ is the 
only vertex in $V$ with $\sigma(v) = \star$.
Recall that $\+{T}_{\sigma}=(T_{\sigma},\rho_{\sigma})$ is the RCT rooted at $\sigma$.
For each leaf $\tau$ in $T_{\sigma}$, we have
the depth of $\tau$ in $T_{\sigma}$ is no more than 
$k L\Delta^2$.
Moreover, if $\tau\in \+{L}_g(\+{T}_{\sigma})$, then
$\abs{\+C^{\tau}_v} \leq  \Delta \cdot \abs{\csfrozen{\tau}}+\Delta\cdot \abs{\vst{\tau}}\leq  L\Delta^2 $.
\end{corollary}

Now we can prove \Cref{lem-cvxl-in-cbad2-cor}.
\begin{proof}[Proof of \Cref{lem-cvxl-in-cbad2-cor}]
Given $(\Phi,\sigma,v)$ satisfying \Cref{inputcondition-calc},
by \Cref{Alg:calc}
the nontrivial costs in $\calc{}(\Phi,\sigma,v)$ is to calculate $ \recalc{}(\Phi,\sigma_{v\gets \star},v)$.
Thus, to prove the proposition, it is sufficient to prove that
$\tra(\Phi,\sigma_{v\leftarrow \star},v)\leq \poly(n,q^{kL\Delta^2})$.
Let $\tau = \sigma_{v\leftarrow \star}$.
By $(\Phi,\sigma,v)$ satisfies \Cref{inputcondition-calc},
we have $(\Phi,\tau,v)$ satisfies \Cref{inputcondition-recalc}
and $v$ is the only vertex in $V$ with $\tau(v)=\star$.
Thus we have 
\begin{equation*}
\begin{aligned}
&\sum_{X\in \+{L}_g(\+{T}_{\tau})}\left( (k\abs{ \+{C}_v^X}+q\abs{ V_v^X})\cdot q^{\abs{ V_v^X}}\right)
\leq \sum_{X\in \+{L}_g(\+{T}_{\tau})}\left( (k\abs{ \+{C}_v^X}+qk\abs{ \+{C}_v^X}+q)\cdot q^{k\abs{ \+{C}_v^X}+1}\right)\\
\leq &\sum_{X\in \+{L}_g(\+{T}_{\tau})}\left( (k L\Delta^2+q
k L\Delta^2)\cdot q^{2k L\Delta^2}\right) \leq \abs{\+T_{\tau}}\cdot (k L\Delta^2+2q
k L\Delta^2)\cdot q^{2k L\Delta^2},
\end{aligned}
\end{equation*}
where the first inequality is by 
$\abs{ V_v^\tau} \leq k\abs{ \+{C}_v^\tau}+1$, the second one is by 
 \Cref{cor-depth-sizectauv-rct},
and the last one is by $\+{L}_g(\+{T}_{\tau})\leq \abs{\+T_{\tau}}$.
Combining with \Cref{RCTtime}, we have 
$$
\tra{}(\Phi,\tau,v)\le \abs{\+T_{\tau}} \cdot \poly(n,\Delta,q^k) +  \abs{\+T_{\tau}}\cdot O\left((k L\Delta^2+2q
k L\Delta^2)\cdot q^{2k L\Delta^2}\right) \leq \abs{\+T_{\tau}}\cdot \poly(n,q^{k L\Delta^2}).
$$
By \Cref{cor-depth-sizectauv-rct}, we have the depth of $T_{\tau}$ is at most $kL\Delta^2$. In addition, we have $$\abs{\+T_{\tau}}\leq \sum\limits_{i=0}^{kL\Delta^2}(1+q)^{i}\leq 2(1+q)^{kL\Delta^2}$$ by $T_{\tau}$ is a tree where each node has at most $q+1$ children,
Therefore, we have 
$$
\tra{}(\Phi,\tau,v)\leq \abs{\+T_{\tau}}\cdot \poly(n,q^{k L\Delta^2}) \leq \poly(n,q^{k L\Delta^2}),
$$
which finishes the proof.
\end{proof}

 Recall the definition of $X^n$ in \Cref{def-pas-cmain}.
 We need the following lemma, which is an analogy of \Cref{WTsize},
 whose proof is deferred to \Cref{sec:proof of lem-bigWT-rs}.

\begin{lemma}\label{lem-bigWT-rs}
For every $v\in V$, 
there exists a generalized $\{2,3\}$-tree $T = \{v\}\circ E$ in $H_{\Phi}$ where
$
E\subseteq \cfrozen{X^n}$
and 
$
\Delta^2\abs{E}\geq \abs{\+{C}^{X^n}_v}$.
In addition, if $\abs{\+{C}^{X^n}_v}\geq L\Delta^2$, then there exists a generalized $\{2,3\}$-tree $T=\{v\}\circ E$  in $H_{\Phi}$ with some auxiliary tree rooted at $v$ where
$ E\subseteq \cfrozen{X^n}$ and $L\leq 1+\Delta\cdot \abs{E}\leq L\Delta$.
\end{lemma}

We are now ready to prove \Cref{lem-exeef2}.

\begin{proof}[Proof of \Cref{lem-exeef2}]
For simplification, let $\tau = X^n$.
Let $\set{(V^{\tau}_i,\+{C}^{\tau}_i)\mid 1\leq i\leq K}$ denote all connected components in $H_{\Phi^{\tau}}$.
By the standard guarantee on the running time of exhaustive enumeration, it is sufficient to show that $\abs{V^{\tau}_i}\leq kL\Delta^2+1$ for each $i\leq K$.
Assume for contradiction that there exists some $i\in [K]$ such that $\abs{V^{\tau}_i}> kL\Delta^2+1$.
Then we have $\abs{\+{C}^{\tau}_v}> L\Delta^2$ by 
$\abs{ V_v^\tau} \leq k\abs{ \+{C}_v^\tau}+1$.
Combining with \Cref{lem-bigWT-rs}, we have there exists a generalized $\{2,3\}$-tree $T=\{v\}\circ E$  in $H_{\Phi}$ with some auxiliary tree rooted at $v$ where $E\subseteq \cfrozen{\tau}$ and $L\leq 1+\Delta\cdot \abs{E}\leq L\Delta$. Let $\+{U}$ denote the set $\{Y\in \+{T}^{t}_{v}: L\leq t\leq L\Delta\}$, then we have $T\in \+{U}$.
Therefore, we have
\begin{equation*}
    \begin{aligned}
    &\digamma(\tau)\\
(\text{by \eqref{eq-definition-F-2}})\quad=   &\sum\limits_{i=L}^{L\Delta}\sum\limits_{v\in V}\sum\limits_{S\in\+{T}^{i}_{v}}F(\tau,S\setminus \set{v})\\
(\text{by $F(\cdot,\cdot)\geq 0$, and $T\in \+{U}$})\quad\geq   &F(\tau,T\setminus \{v\})\\
(\text{by \eqref{eq-def-fsimgat}})\quad=   &\prod\limits_{c\in E}\pprime^{-1}\mathbb{P}[\neg c\mid \tau](1+\eta)^k\\
(\text{by $(1+\eta)^k>1$})\quad\geq   &\prod\limits_{c\in E}\left(\pprime^{-1}\mathbb{P}[\neg c\mid \tau]\right)\\
(\text{by $E\subseteq \cfrozen{\tau}$})\quad\geq  &1,
    \end{aligned}
\end{equation*}
which is contradictory with $\digamma(\tau)<1$.
Thus, we have $\abs{V^{\tau}_i}\leq kL\Delta^2+1$ for each $i\leq K$ and the lemma is immediate.
\end{proof}

\subsection{Analysis of the main counting algorithm}

We are now ready to present the analysis for the main counting algorithm. Recall the sequence of partial assignments $X^0,X^1,\dots,X^n$ that evolve in the main counting algorithm defined in \Cref{def-pas-cmain}.

A crucial lemma we will show is the following, which states $\digamma(X)$ is small throughout the main counting algorithm.

\begin{lemma}\label{inithcor2}
If $16\mathrm{e}p\Delta^3\leq \pprime$, $\eta\leq (2k)^{-1}$, $1-q\theta\leq (8\mathrm{e}k\Delta)^{-1}$ and $L\geq 9$, then it holds for all $0\le i\le n$ that $$\digamma(X^i)<8n\Delta\cdot 2^{-\lfloor\frac{L}{2\Delta}\rfloor}.$$
\end{lemma}
 
 We first show the following lemma, which states that $\digamma(\hollowstar^V)$ is small.
\begin{lemma}\label{inith-2}
Under the condition of \Cref{inithcor2}, $$\digamma(\hollowstar^V)<8n\Delta\cdot 2^{-\lfloor\frac{L}{2\Delta}\rfloor}.$$
\end{lemma}
\begin{proof}
Recall by \eqref{eq-definition-FT-2} and \eqref{eq-definition-F-2} that 
\begin{equation}\label{fexp}
    \begin{aligned}
    \digamma(\hollowstar^V)=\sum\limits_{i=L}^{L\Delta}\sum\limits_{v\in V}\sum\limits_{T\in\+{T}^{i}_{v}}\left( \left(1-q\theta\right)^{\abs{U}-1}\cdot \left(p\pprime^{-1} (1+\eta)^k\right)^{\abs{E}}\right)
    \end{aligned}
\end{equation}

By \eqref{eq-definition-FT-2}, in order to bound the term in $\digamma(\hollowstar^V)$, instead of bounding \eqref{fexp} over all possible generalized $\{2,3\}$-tree $T$,  it is sufficient to bound the term in \eqref{fexp} over all possible auxiliary trees $T^*$. This is because distinct generalized $\{2,3\}$-trees must have distinct auxiliary trees by \Cref{WTdef}.

Fix any $v\in V$. We define the following multi-type Galton-Watson process that generates all possible auxiliary trees of some generalized $\{2,3\}$-tree rooted at $v$ as a subgraph with the corresponding probability. 

\begin{definition}[A multi-type Galton-Watson process generating auxiliary trees as subgraphs]\label{GWprocess}
For each $v\in V$, we define the following multi-type Galton-Watson process that generates a rooted directed tree $T^*$ with vertex set $V(T^*)\subseteq V\cup \+{C}$ that generates all possible auxiliary trees of generalized $\{2,3\}$-trees rooted at $v$ as a subgraph with the corresponding probability:

\begin{enumerate}
    \item The root of $T^*$ is $v$ and the depth of $v$ is $0$.
    \item For $i=0,1,\dots,$: for all nodes $v\in V(T^*)$ of depth $i$ in the current $T^*$
    \begin{enumerate}
        \item If $v\in V$:
        \begin{enumerate}
            \item For each $u\in V$ such that there exists $c\in \+{C}$ where $v,u\in \var{c}$,  add $u$ as a child of $v$ independently with probability $1-q\theta$. Note that there are at most $k\Delta$ such $u$.
            \item For each $c\in \+{C}$ such that there exists $c'\in \+{C}$ where $v\in \var{c'}\wedge \text{dist}_{G(\+{C})}(c,c')=1$, add $c$ as a child of $v$ independently with probability $p\pprime^{-1}(1+\eta)^k$. Note that there are at most  $\Delta^2$ such $c$.
        \end{enumerate}
        \item If $v\in \+{C}$:
        \begin{enumerate}
            \item For each $u\in V$ such that there exists $c\in \+{C}$ where $u\in \var{c}\wedge \text{dist}_{G(\+{C})}(v,c)=1\text{ or }2$,  add $u$ as a child of $v$ independently with probability $1-q\theta$. Note that there are at most  $k\Delta^2$ such $u$.
            \item For each $c\in \+{C}$ such that $\text{dist}_{G(\+{C})}(c,c')=2\text{ or }3$,  add $c$ as a child of $v$ independently with probability $p\pprime^{-1} (1+\eta)^k$. Note that there are at most  $\Delta^3$ such $c$.
        \end{enumerate}
    \end{enumerate}
\end{enumerate}
\end{definition}
Note that the process in \Cref{GWprocess} may generate trees that violate the rule of auxiliary trees, as vertices may be repeated and constraints in the tree may not be pairwise disjoint. Nevertheless, by comparing with \Cref{WTdef}, it can be verified that this process generates all possible auxiliary trees $T$  rooted at $v$ of some generalized $\{2,3\}$-tree as a subgraph with probability exactly $F(\hollowstar^V,V(T))$. By \eqref{eq-definition-F-2}, it then suffices to bound the sum of the probability such process generating a tree $T$ as a subgraph over all possible auxiliary trees $T$ rooted at $v$ satisfying  $L\leq \abs{V(T^*)\cap V}+\Delta\cdot \abs{V(T^*)\cap \+C}\leq L\Delta$. However, this is still not convenient enough for calculation, so we further define the following two-type Galton-Watson process.

\begin{definition}[A two-type Galton-Watson process]\label{GWprocess2}
We define the following multi-type Galton-Watson process that generates a tree $T$ with vertex set $V(T)$ consisting nodes of two types.

\begin{enumerate}
    \item The root of $T$, $r_T$, is of type $1$ or type $2$, and of depth $0$.
    \item For $i=0,1,\dots,$: for all nodes $v\in V(T)$ of depth $i$ in the current $T$
    \begin{enumerate}
        \item If $v$ is of type $1$:
        \begin{enumerate}
            \item Independently repeat $k\Delta$ times: generate a node of type $1$ as a child of $v$ with probability $1-q\theta$.
            \item Independently repeat $\Delta^2$ times: generate a node of type $2$ as a child of $v$ with probability $p\pprime^{-1}(1+\eta)^k$.
        \end{enumerate}
        \item If $v$ is of type $2$:
        \begin{enumerate}
            \item Independently repeat $k\Delta^2$ times: generate a node of type $1$ as a child of $v$ with probability $1-q\theta$.
            \item Independently repeat $\Delta^3$ times: generate a node of type $2$ as a child of $v$ with probability $p\pprime^{-1}(1+\eta)^k$.
        \end{enumerate}
    \end{enumerate}
\end{enumerate}
\end{definition}

It is easy to construct an injection between each tree generated as a subgraph by the process in \Cref{GWprocess} and each tree generated as a subgraph by  the process in \Cref{GWprocess2} with the same probability. Then it is sufficient to bound the sum of probability of  the process in \Cref{GWprocess2} generating a tree $T$ as a subgraph over all $T$ satisfying $ L\leq a+\Delta\cdot b\leq L\Delta$ with root $r_T$, where $a$ and $b$ represent the number of type $1$ nodes and type $2$ nodes in $V(T)$, respectively. 

If $r_T$ is of type 1, 
let $f_1(x)$ be the generating function for the random tree generated in  the process in \Cref{GWprocess2} where for each $m\geq 0$, the $m$-th coefficient $[x^m]f_1(x)$ represents the sum of probability that  the process in \Cref{GWprocess2} generates a tree $T$ as a subgraph over all directed tree $T$ rooted at $r_T$ satisfying $a+\Delta\cdot b=m$, where $a,b$ are the numbers of type 1 and type 2 nodes in $V(T)$, respectively. 
Otherwise, $r_T$ is of type 2, and we define $f_2(x)$ similarly.
Let $p_1=1-q\theta$ and $p_2= p\pprime^{-1}\cdot (1+\eta)^k$.
By \Cref{GWprocess2} 
we have 
\begin{align*}
    f_1 &= x(1+p_1f_1)^{k\Delta}\cdot (1+p_2f_2)^{\Delta^2},\\
    f_2 &= x^{\Delta}(1+p_1f_1)^{k\Delta^2}\cdot (1+p_2f_2)^{\Delta^3}.
\end{align*}
Thus, we have 
$f_2 = f_1^{\Delta}$ and then
\begin{align*}
    f_1 &= x(1+p_1f_1)^{k\Delta}\cdot (1+p_2f_1^{\Delta})^{\Delta^2}.
\end{align*}
Let $g$ be the functional inverse of $f_1$. Formally,
$g(f_1(x))=f_1(g(x))=x$.
By $f_1(0)=0$, we also have $g(0) = 0$ and
$$\lim_{y\to 0} \frac{f_1(y) - f_1(0)}{y - 0} = [x]f_1(x) \geq 1.$$
Thus, we have 
$$g'(0) = \lim_{u\to 0} \frac{g(u) - 0 }{u - 0} = \lim_{f_1(y)\to 0} \frac{y - 0 }{f_1(y) - 0}  \neq 0,$$
where the last equality is by $f_1(0) = 0$.

By $g(0)=0$ and $g'(0)\neq 0$, the condition of Lagrange inversion theorem is satisfied. 
By applying the theorem we have for $m \geq 1$,
\begin{equation}
\begin{aligned}\label{eq-efficiency-generationfuncition}
   \relax [x^{m}]f_1(x)&=\frac{1}{m}[u^{m-1}]\left(\frac{u}{g(u)}\right)^{m}\\
   &=\frac{1}{m}[u^{m-1}]\left((1+p_1u)^{k\Delta}\cdot (1+p_2u^{\Delta})^{\Delta^2}\right)^{m}\\
   &= \frac{1}{m}\sum\limits_{i=0}^{\lfloor\frac{m}{\Delta}\rfloor}\left([u^{(m-1-\Delta i)}]\left(1+p_1u\right)^{k\Delta m}\cdot [u^{(\Delta i)}]\left(1+p_2u^{\Delta}\right)^{\Delta^2 m} \right)\\
   &\leq \frac{1}{m}\left(\sum\limits_{i=0}^{\lfloor\frac{m}{2\Delta}\rfloor} [u^{(m-1-\Delta i)}]\left(1+p_1u\right)^{k\Delta m}+\sum\limits_{i=\lfloor\frac{m}{2\Delta}\rfloor+1}^{\lfloor\frac{m}{\Delta}\rfloor}[u^{(\Delta i)}]\left(1+p_2u^{\Delta}\right)^{\Delta^2 m}\right)
\end{aligned}
\end{equation}
In addition, let $t = m- 1 - \Delta i$.
For each $0\leq i\leq\lfloor\frac{m}{2\Delta}\rfloor$,
we have 
\begin{align*}
 &[u^{t}]\left(1+p_1u\right)^{mk\Delta } 
 =  p_1^{t}\cdot \binom{mk\Delta}{t}
\end{align*}
Moreover, for each $m\geq 8$ and $0\leq i\leq\lfloor\frac{m}{2\Delta}\rfloor$, we have $mk\Delta \leq 1.2(m-1)k\Delta \leq 4tk\Delta$.
Thus
$$\binom{mk\Delta}{t}\leq \binom{4tk\Delta}{t}\leq
\left(\frac{4\mathrm{e}tk\Delta}{t}\right)^{t}\leq
\left(4\mathrm{e}k\Delta\right)^{t},$$
where the second inequality is by that for each $0< \gamma\leq \beta$ where $\gamma,\beta$ are integers,
$$\binom{\beta}{\gamma}\leq \left(\frac{\mathrm{e}\beta}{\gamma}\right)^{\gamma}.$$ 
Thus, we have
\begin{equation}\label{eq-efficiency-firstpart}
\begin{aligned}
 &[u^{t}]\left(1+p_1u\right)^{mk\Delta } 
 =  p_1^{t}\cdot \binom{nk\Delta}{t} \leq p_1^{t}\left(4\mathrm{e}k\Delta\right)^{t} = \left(4\mathrm{e}p_1k\Delta\right)^{t}.
\end{aligned}
\end{equation}

Similarly,
for each $\lfloor\frac{m}{2\Delta}\rfloor<i \leq \lfloor\frac{m}{\Delta}\rfloor$,
we have 
\begin{align*}
[u^{\Delta i}]\left(1+p_2u^{\Delta}\right)^{m\Delta^2 }=  p_2^{i}\cdot \binom{m\Delta^2}{i}.
\end{align*}

Moreover, for each $m\geq 8$ and $\lfloor\frac{m}{2\Delta}\rfloor<i \leq \lfloor\frac{m}{\Delta}\rfloor$, we have $m\Delta^2 \leq 2i\Delta^3$.
Thus
$$\binom{m\Delta^2}{i}\leq \binom{2i\Delta^3}{i}\leq
\left(\frac{2\mathrm{e}i\Delta^3}{i}\right)^{i}=
\left(2\mathrm{e}\Delta^3\right)^{i},$$
where the second inequality is also by that for each $0< \gamma\leq \beta$ where $\gamma,\beta$ are integers,
$$\binom{\beta}{\gamma}\leq \left(\frac{\mathrm{e}\beta}{\gamma}\right)^{\gamma}.$$ 
Thus, we have
\begin{align*}
 &[u^{\Delta i}]\left(1+p_2u^{\Delta}\right)^{m\Delta^2 }=  p_2^{i}\cdot \binom{m\Delta^2}{i}
\leq p_2^{i}\left(2\mathrm{e}\Delta^3\right)^{i} = \left(2\mathrm{e}p_2\Delta^3\right)^{i}.
\end{align*}
Combining with \eqref{eq-efficiency-generationfuncition} and \eqref{eq-efficiency-firstpart}, we have 
\begin{align*}
    \relax [x^{m}]f_1(x)&\leq m^{-1}\left(\sum\limits_{i=0}^{\lfloor\frac{m}{2\Delta}\rfloor} \left(4\mathrm{e}p_1k\Delta\right)^{m-1-\Delta i}+\sum\limits_{i=\lfloor\frac{m}{2\Delta}\rfloor+1}^{\lfloor\frac{m}{\Delta}\rfloor}\left(2\mathrm{e}p_2\Delta^3\right)^{i}\right)
\end{align*}
Note that by $16\mathrm{e}p\Delta^3\leq \pprime$, $\eta\leq (2k)^{-1}$ and $1-q\theta\leq (8\mathrm{e}k\Delta)^{-1}$ we have $p_1\leq (8\mathrm{e}k\Delta)^{-1}$ and $p_2\leq (8\mathrm{e}\Delta^3)^{-1}$, hence we have
for each $m\geq 8$,
\begin{align*}
    \relax [x^{m}]f_1(x)&\leq m^{-1}\left(\sum\limits_{i=0}^{\lfloor\frac{m}{2\Delta}\rfloor} 2^{(\Delta i+1-m)}+\sum\limits_{i=\lfloor\frac{m}{2\Delta}\rfloor+1}^{\lfloor\frac{m}{\Delta}\rfloor}2^{-i}\right)\leq m^{-1}\left(2^{\left(2-\frac{m}{2}\right)}+2^{-\lfloor\frac{m}{2\Delta}\rfloor}\right)\leq  2^{-\lfloor\frac{m}{2\Delta}\rfloor}.
\end{align*}
Therefore by the analysis above we have
$$\digamma(\hollowstar^V)\leq n\cdot \sum\limits_{i=L}^{L\Delta}2^{-\lfloor\frac{i-1}{2\Delta}\rfloor}\leq 8n\Delta\cdot 2^{-\lfloor\frac{L}{2\Delta}\rfloor}.   $$

\end{proof}

The following lemma is a property for the upper bound function $\digamma(\cdot)$ defined in \Cref{def:potential-refined} .

\begin{lemma}\label{potentialfact2}
For any partial assignment $\sigma\in \qs$, any variable $u\in V$ with $\sigma(u)=\hollowstar$, we have
$$ \sum\limits_{x\in Q_u}\digamma(\sigma_{u\gets x})=q_u\cdot \digamma(\sigma).$$
\end{lemma}
\begin{proof}
Given $T=U\circ E\in \+{T}^i_{v}$ for any $L\leq i\leq L\Delta$ and $v\in V$, by the definition of $\+{T}^i_{v}$, 
we have $\var{c}\cap \var{c'}\neq \emptyset$ for any different $c,c'\in E$.
Then there exists at most one unique constraint $c_0\in E$ such that $u\in \var{c_0}$.
Let $S = T\setminus \{c_0\}$ if $u\in \var{E}$ and $S = T$ otherwise.
Thus for each $c\in S\cap \+{C}$, we have $u\not\in \var{c}$.
Then
$\mathbb{P}[\neg c\mid \sigma_{u\gets x}]=\mathbb{P}[\neg c\mid \sigma]$ and 
$F(\sigma_{u\gets x},c)=F(\sigma,c)$ for each $x\in Q_u$.
Therefore, 
\begin{equation*}
\begin{aligned}
    F(\sigma_{u\gets x},S)=(1-q\theta)^{\abs{S\cap V}}\prod\limits_{c\in S\cap \+{C}}F(\sigma_{u\gets x},c)=(1-q\theta)^{\abs{S\cap V}}\prod\limits_{c\in S\cap \+{C}}F(\sigma_{u},c)=F(\sigma,S).
\end{aligned}
\end{equation*}
Note that by the definition of $F(\cdot,\cdot)$ it is easy to verify that for any partial assignment $\tau\in \qs$ and any two non-intersecting subsets $S_1, S_2\subseteq V\cup \+C$ we have
$$ F(\tau,S_1\cup S_2)=F(\tau,S_1)\cdot F(\tau,S_2).$$
Thus we have 
\begin{equation*}
\begin{aligned}
    \sum\limits_{x\in Q_u}F(\sigma_{u\gets x},T)&=\sum\limits_{x\in Q_u}\left(F(\sigma_{u\gets x},T\setminus S)\cdot F(\sigma_{u\gets x},S) \right)=\left(\sum\limits_{x\in Q_u}F(\sigma_{u\gets x},T\setminus S)\right)F(\sigma,S).
\end{aligned}
\end{equation*}
In addition, we claim that
\begin{equation}\label{eq-sum-xinqu-fsigmauxtsetminuss}
\sum\limits_{x\in Q_u}F(\sigma_{u\gets x},T\setminus S) =  q_u F(\sigma,T\setminus S).
\end{equation}
Thus, we have
\begin{equation*}
\begin{aligned}
    \sum\limits_{x\in Q_u}F(\sigma_{u\gets x},T)&=q_u F(\sigma,T\setminus S)F(\sigma,S)=q_u F(\sigma,T).
\end{aligned}
\end{equation*}
Therefore, combining with \eqref{eq-definition-F-2} we have
\begin{equation*}
\begin{aligned}
     \sum\limits_{x\in Q_u}\digamma(\sigma_{u\gets x})= &\sum\limits_{x\in Q_u}\sum\limits_{i=L}^{L\Delta}\sum\limits_{v\in V} \sum\limits_{T\in \+{T}^{i}_{v}}F(\sigma_{u\gets x},T\setminus \set{v} )= \sum\limits_{i=L}^{L\Delta}\sum\limits_{v\in V} \sum\limits_{T\in \+{T}^{i}_{v}}\sum\limits_{x\in Q_u}F(\sigma_{u\gets x},T\setminus \set{v})\\=&q_u\sum\limits_{i=L}^{L\Delta}\sum\limits_{v\in V} \sum\limits_{T\in \+{T}^{i}_{v}}F(\sigma,T\setminus \set{v})
=q_u\cdot \digamma(\sigma).
\end{aligned}
\end{equation*}

In the following, we prove \eqref{eq-sum-xinqu-fsigmauxtsetminuss}.
If $S = T\setminus \{c_0\}$, by $\sigma(u)=\hollowstar$ we have
$$\sum\limits_{x\in Q_u}\mathbb{P}[\neg c\mid \sigma_{u\gets x}] = q_u\mathbb{P}[\neg c\mid \sigma].$$
Thus
\begin{equation*}
    \begin{aligned}
    \sum\limits_{x\in Q_u}F(\sigma_{u\gets x},T\setminus S) &= \sum\limits_{x\in Q_u}F(\sigma_{u\gets x},c_0)=\sum\limits_{x\in Q_u}\left(\pprime^{-1}\mathbb{P}[\neg c_0\mid \sigma_{u\gets x}](1+\eta)^k\right)\\
    &=\pprime^{-1}(1+\eta)^k \sum\limits_{x\in Q_u}\mathbb{P}[\neg c_0\mid \sigma_{u\gets x}] =q_u\pprime^{-1}(1+\eta)^k\mathbb{P}[\neg c_0\mid \sigma]\\
    &=q_uF(\sigma,c_0)=q_uF(\sigma,T\setminus S).
    \end{aligned}
\end{equation*}
Otherwise, $T\setminus S$ is empty. We also have
\begin{equation*}
\sum\limits_{x\in Q_u}F(\sigma_{u\gets x},T\setminus S) = \sum\limits_{x\in Q_u}\prod\limits_{c\in T\setminus S}F(\sigma_{u\gets x},c)= q_u\prod\limits_{c\in T\setminus S}F(\sigma,c)=q_uF(\sigma,T\setminus S),
\end{equation*}
where we assume that a product over an empty set is 1.
Thus, we have \eqref{eq-sum-xinqu-fsigmauxtsetminuss} always holds and the lemma is proved. 
\end{proof}

 We are now ready to prove \Cref{inithcor2}.
\begin{proof}[Proof of \Cref{inithcor2}]
We prove the lemma by induction on $i$ where $i\in \{0,1,\cdots,n\}$. 
Recall the sequence of partial assignments $X^0,X^1,\dots,X^n$ that evolve in the main counting algorithm defined in \Cref{def-pas-cmain}.
For the base case where $i = 0$, we have $X^0=\hollowstar^V$. By \Cref{inith-2}, we have $$\digamma(\hollowstar^V)< 8n\Delta\cdot 2^{-\lfloor\frac{L}{2\Delta}\rfloor}.$$

For the induction step, it then suffices to show for each $i\in [n]$,
\begin{equation}
    \digamma(X^{i})\leq \digamma(X^{i-1}).
\end{equation}
Recall $v_i$ in Line \hyperlink{line2}{2} of the main counting algorithm.
Given $i\in [n]$, we have $X^{i-1}(v_i) = \hollowstar$.
If $v_i$ is involved in some $X^{i-1}$-frozen constraint, we have the condition in Line \hyperlink{line2}{2} of the main counting algorithm is not satisfied, then we have $X^i=X^{i-1}$ and  $\digamma(X^i)= \digamma(X^{i-1})$. Otherwise the condition is satisfied, and by Line \hyperlink{line2b}{2(b)} of the main counting algorithm we have
$$ \digamma(X^i)=\min\limits_{x\in Q_v}\digamma(X^{i-1}_{v_i\gets x})\leq q_{v_i}^{-1}\sum\limits_{x\in Q_v}\digamma(X^{i-1}_{v_i\gets x})=\digamma(X^{i-1}),$$
where the last equality is by \Cref{potentialfact2} and $X^{i-1}(v_i) = \hollowstar$. 
\end{proof}

We are now ready to prove Theorem \ref{thm:main-counting-refined}.

\begin{proof}[Proof of Theorem \ref{thm:main-counting-refined}]
Let 
 $$\pprime=\left(16\mathrm{e^2}q^2k\Delta^2\right)^{-1}$$ 
and
$$L=100\Delta\left\lceil\log{\left(\frac{qn\Delta}{\varepsilon}\right)}\right\rceil.$$
We choose the truncation condition $f(\cdot)$ in \Cref{def:truncate-refined}, the upper bound function $\digamma(\cdot)$ in \Cref{def:potential-refined}.

We first prove the bound on the running time of the main counting algorithm. 
The followings are the nontrivial costs in the main counting algorithm:
\begin{itemize}
    \item the cost of estimate the marginal distribution $\mu^X_{v_i}$ with $\calc$ in Line \hyperlink{line2a}{2(a)};
    \item the cost of calculating $\digamma(\cdot)$ in Line \hyperlink{line2b}{2(b)};
    \item the cost of the exhaustive enumeration in Line \hyperlink{line3}{3}.
\end{itemize}
Let $\tcmain{}$ be the running time of the main counting algorithm. 
Recall that $\texe$ denotes the cost of the exhaustive enumeration in the main counting algorithm, and $\tma$ denotes the cost of $\calc$.
Moreover, one can prove that the total cost of calculating $\digamma(\cdot)$ in the main counting algorithm is $n^{\poly(\log q,\Delta,k)}$.
Thus we have
\begin{equation}\label{eq-main-counting-1}
    \tcmain{}\leq \texe{}(\Phi,X^n)+\sum\limits_{i=1}^{n}\tma{}(\Phi,X^{i-1},v_i)+n^{\poly(\log q,\Delta,k)},
\end{equation} 
By \Cref{lemma:invariant-counting} and \Cref{lem-cvxl-in-cbad2-cor}, we also have for each $i\in [n]$,
$$\tma{}(\Phi,X^{i-1},v_i)=\poly(n, q^{k L\Delta^2}).$$
In addition, by \Cref{lem-exeef2} we have
\begin{align}\label{eq-upperbound-texe}
\texe{}(\Phi,X^n)=\poly(n,q^{k L\Delta^2})
\end{align}
Combining with \eqref{eq-main-counting-1} and the definition of $L$, we have 
$$\tcmain{}\leq \poly(n, q^{k L\Delta^2})+n^{\poly(\log q,\Delta,k)}=O\left(\left(\frac{n}{\varepsilon}\right)^{\poly(\log{q},\Delta,k)}\right).$$

In the following, we prove the bound on the cost of calculating $\digamma(\cdot)$.
For each $v\in V$ and $L\leq i\leq L\Delta$, 
by the definition of $\+{T}^{i}_v$ in \eqref{eq-definition-FT-1} and \eqref{eq-definition-FT-2}, it is immediate that $\+{T}^{i}_v$ can be embedded as a subtree in a graph $G$ with vertex set $V\cup \+C$ and degree bounded by $2\Delta^3$,
where the subtree is rooted at $v$ of size at most $i$. 
Note that for any 
degree-bounded graph $G$ with maximum degree $D$, the number of subtrees of $G$ with $t$ vertices and a specific vertex as root, is at most $(\mathrm{e}D)^{t-1}/2$~\cite[Lemma 2.1]{borgs2013left}. 
Then we have 
$$\abs{\+{T}^{i}_v} \leq  \frac{(2\mathrm{e}\Delta^3)^{i-1}}{2} \leq \frac{(2\mathrm{e}\Delta^3)^{L\Delta}}{2}  \leq  n^{\poly(\log{q},\Delta,k)}.$$
To construct the set $\+{T}^{i}_v$,
one can first construct each possible subtree $T$ of $G$ where $T$ is rooted at $v$ and of size no more than $i$,
and then check whether $T\in \+{T}^{i}_v$.
Therefore, the set $\+{T}^{i}_v$ can also be constructed in $n^{\poly(k,\Delta)}$ cost. 
Combining with \eqref{eq-def-fsimgat} and \eqref{eq-definition-F-2}, 
we have for any $\sigma\in \qs$,
$\digamma(\sigma)$ can be calculated within cost
$$L\Delta\cdot n\cdot n^{\poly(k,\Delta)} \leq n^{\poly(\log q,\Delta,k)}.$$
In addition, by Line \hyperlink{line2}{2} of the main counting algorithm,
we have $\digamma(\cdot)$ is calculated at most $qn$ times.
Thus, 
the total cost of calculating $\digamma(\cdot)$ in the main counting algorithm is $n^{\poly(\log q,\Delta,k)}$.

In the next, we prove \eqref{eq-upperbound-texe}.
Recall the definition of $\eta$ and $\theta$ in \eqref{eq:parameter-theta}.
By \eqref{eq:main-thm-LLL-condition} and the above definitions of $\pprime$ and $L$,
it then can be verified that 
$$p<\pprime<(\mathrm{e}q\Delta)^{-1}, \quad 16\mathrm{e}p\Delta^3\leq \pprime, \quad \eta\leq  (2k)^{-1}, \quad 1-q\theta\leq (8\mathrm{e}k\Delta)^{-1}, \quad L\geq 9.$$
Therefore, the conditions of \Cref{inithcor2} are satisfied.
Thus, for large enough $n$ we have 
$$\digamma(X^n)<8n\Delta\cdot 2^{-\lfloor\frac{L}{2\Delta}\rfloor}\leq 8n\Delta\left(\frac{q n\Delta}{\varepsilon}\right)^{-50}<1.$$
Combining with \Cref{lem-exeef2}, \eqref{eq-upperbound-texe} is immediate.
The upper bound of $\tcmain{}$ is proved.

\sloppy
At last, we prove the bound on the relative error.
Let $S = \{i\in [n]\mid v_i\notin \vfix{X^{i-1}}\}$.
For each $i\in S$, let $\xi_{i}$ be the distribution returned by $\calc{}(\Phi,X^{i-1},v_i)$.
Thus, we have
$$\tv(\xi_i,\mu^{X^{i-1}}_{v_i})<8n\Delta\cdot 2^{-\lfloor\frac{L}{2\Delta}\rfloor}<\frac{\varepsilon}{8nq},$$
where the first inequality is by $L\geq 9$ and \Cref{tvdbound2}, and the second inequality is by \Cref{inithcor2}.

Combining with \Cref{lemma:invariant-counting} and \Cref{prop:theta-recalc-lower-bound}, we have
\[
    \xi_i(X^i(v_i))\geq \mu^{X^{i-1}}_{v_i}(X^{i}(v_i))-\frac{\varepsilon}{8nq}\geq \mu^{X^{i-1}}_{v_i}(X^{i}(v_i))\left(1-\frac{\varepsilon}{4n}\right)
\]
Thus, by Line \hyperlink{line2b}{2(b)} of the main counting algorithm, 
we have
\begin{align*}
\widehat{Z} &=  \abs{\+{S}_{X^n}}\cdot \prod_{i\in S}(\xi_i(X^{i}(v_i)))^{-1} \leq  \abs{\+{S}_{X^n}}\cdot \prod_{i\in S}\left((1 + \varepsilon(2n)^{-1})(\mu^{X^{i-1}}_{v_i}(X^{i}(v_i)))^{-1}\right) 
\\&\leq (1 + \varepsilon(2n)^{-1})^n\cdot \abs{\+{S}_{X^n}}\cdot \prod_{i\in S}(\mu^{X^{i-1}}_{v_i}(X^{i}(v_i)))^{-1} \leq  (1 + \varepsilon)\cdot \abs{\+{S}_{X^n}}\cdot \prod_{i\in S}(\mu^{X^{i-1}}_{v_i}(X^{i}(v_i)))^{-1} = (1 + \varepsilon)Z_{\Phi},
\end{align*}
where the last equality is by \eqref{eq-telescope}.
Similarly, one can also prove $(1-\varepsilon)Z_{\Phi} \leq \widehat{Z}$.
\end{proof}

\section{On improving the JPV algorithm}\label{sec:JPV}
In this section, we explain how to use the generalized $\{2,3\}$-to improve the analysis of the algorithm in~\cite{Vishesh21towards} with an improved LLL condition $p\Delta^5\lesssim 1$.

As stated in the technique overview and \Cref{sec:counting}, the algorithm presented in~\cite{Vishesh21towards} uses the same framework for the main counting algorithm, only with the subroutine for estimating the marginal probabilities replaced with the procedure of setting up a linear program to mimicry the transition probabilities of an idealized coupling procedure. We then include the definition of the idealized coupling procedure and the subroutine for estimating the marginal probability in~\cite{Vishesh21towards} with conformed notation for a complete illustration:

\begin{definition}[Idealized coupling procedure in~\cite{Vishesh21towards}]\label{jpvcoupling}
Fix any tuple $(\Phi,\sigma,v)$ satisfying \Cref{inputcondition-calc} and any $a,b\in Q_v$, the idealized coupling procedure presented in~\cite{Vishesh21towards} is equivalent to the following:
\begin{enumerate}
    \item Initialize the partial assignments $X \leftarrow X_0=\sigma_{v\gets a}$, $Y \leftarrow Y_0=\sigma_{v\gets b}$, $Z\leftarrow \sigma_{v\gets \star}$. \label{jpvcoupling-1}
    \item Choose $u\leftarrow \nextvar{Z}$, if $u=\perp$, terminates.  \label{jpvcoupling-2}
    \item Sample a pair of values $(x, y )$ according to the maximal coupling of the marginal distribution
of $\mu^X_v$ and $\mu^Y_v$.  \label{jpvcoupling-3}
    \item Update $X$ by assigning $X\leftarrow   X_{v\gets x}$, and update $Y$ by assigning $Y\leftarrow   Y_{v\gets y}$. If $x=y$, update $Z\leftarrow Z_{v\gets x}$, otherwise update $Z\leftarrow Z_{v\gets \star}$.  \label{jpvcoupling-4}
    \item Return to \Cref{jpvcoupling-1}. \label{jpvcoupling-5}
\end{enumerate}
For any pair of partial assignments $(X,Y)\in \qs\times \qs$, we let $\mu_{\textsf{cp}}(X,Y)$ to denote the probability that the idealized coupling
procedure reaches $(X,Y)$.
\end{definition}

The idealized coupling procedure in \Cref{jpvcoupling} inspires the following definition of a (truncated) idealized deterministic rooted decision tree $\MSC{T}$.
\begin{definition}[(Truncated) idealized deterministic rooted decision tree $\MSC{T}$ in~\cite{Vishesh21towards}]\label{jpvtree}
Fix any tuple $(\Phi,\sigma,v)$ satisfying \Cref{inputcondition-calc} and any two distinct values $a,b\in Q_v$ as in \Cref{jpvcoupling}, the (truncated) idealized deterministic rooted decision tree $\MSC{T}$ presented in~\cite{Vishesh21towards} is equivalent to the following:
\begin{enumerate}
    \item The root of $\MSC{T}$ consists of the partial assignments $(X_0=\sigma_{v\gets a}, Y_0=\sigma_{v\gets b})$. Also for any $\sigma,\tau\in \qs$, we define a partial assignment $h(\sigma,\tau):\left(\qs\right)^2\rightarrow \qs$ that captures the ``discrepancy set" between $\sigma$ and $\tau$ as:
    $$ \forall v\in V,h(\sigma,\tau)(v)\defeq \begin{cases} \sigma(v) & \sigma(v)=\tau(v)\\ \star &\text{otherwise}  \end{cases} $$
    \label{jpvtree-1}
    \item For each node $(X, Y)\in \MSC{T} $, if $\nextvar{h(X,Y)}=\perp$ or $f(h(X,Y))=\True$, then $(X,Y)$ is a leaf node of $\MSC{T}$, where $f(\cdot)$ is some truncation condition.     \label{jpvtree-2}
    \item Otherwise, let $u=\nextvar{h(X,Y)}$. The children of $(X, Y)$ in $\MSC{T}$ consist of all
possible extensions of $(X, Y)$ obtained by assigning a pair of values in $Q_u\times Q_u$ to the variable $u$. \label{jpvtree-3}
Similar as in \Cref{RCTdef}, we also let $\+{L}(\MSC{T})$ be the set of leaf nodes in $\MSC{T}$. Let $\+{L}_g(\MSC{T})\defeq\{(X,Y)\in \+{L}(\MSC{T}): f(h(X,Y))=\False \}$ and $\+{L}_b(\MSC{T})\defeq\{X\in \+{L}(\MSC{T}):  f(h(X,Y))=\True\}$ be the sets of leaf nodes $(X,Y)\in \+L(\MSC{T})$ with $h(X,Y)$ don't and do satisfy the truncation condition, respectively.
\end{enumerate}
\end{definition}

One may find a striking resemblance between the RCT defined in \Cref{RCTdef} and the idealized deterministic rooted decision tree defined in \Cref{jpvtree}. A difference is that the RCT appeared implicitly in the analysis of the counting algorithm presented in \Cref{sec:counting}, and the decision tree here is defined explicitly and used directly in the algorithm in~\cite{Vishesh21towards} to appear later. We interpret such similarity as an intrinsic property of the problem instance, which makes possible the improvement of both algorithms using the same refined combinatorial structure of generalized $\{2,3\}$-tree.

We then show the crucial subroutine for estimating the marginal probability in~\cite{Vishesh21towards}.
\begin{definition}[Subroutine for estimating the marginal probability in~\cite{Vishesh21towards}]\label{jpvlp}
Fix any tuple $(\Phi,\sigma,v)$ satisfying \Cref{inputcondition-calc},
the subroutine for estimating the marginal probability presented in~\cite{Vishesh21towards} is equivalent to the following procedure:

For any $a,b\in Q_v$, let $\MSC{T}$ be the idealized deterministic rooted decision tree defined in \Cref{jpvtree}. Set up the following linear program with variables $r_{-},r_{+}$ and $\hat{p}^{X}_{X,Y},\hat{p}^{Y}_{X,Y}$ for each $(X,Y)\in \MSC{T}$:
\begin{enumerate}
    \item For all $(X, Y) \in \+{L}(\MSC{T})$, $0 \leq \hat{p}^{X}_{X,Y},\hat{p}^{Y}_{X,Y} \leq 1$. \label{jpvlp-1}
    \item For every $(X, Y) \in \+{L}_g(\MSC{T})$,
    $$ r_{-}\leq \frac {\hat{p}^{X}_{X,Y}\abs{\+{S}_X}}{\hat{p}^{Y}_{X,Y}\abs{\+{S}_Y}}\leq r_{+},  $$
    where $\frac{\abs{\+{S}_X}}{\abs{\+{S}_Y}}$ is computed through exhaustive enumeration over the connected component in $H_{\Phi^{h(X,Y)}}$ containing $v$.  \label{jpvlp-2}
    \item $\hat{p}^{X_0}_{X_0,Y_0}=\hat{p}^{Y_0}_{X_0,Y_0}=1$. Moreover, for every node $(X, Y) \in \MSC{T}\setminus  \+{L}(\MSC{T})$ and $u = \nextvar{h(X,Y)}$,
    $$\hat{p}^{X}_{X,Y} =\sum_{b\in Q_u}  \hat{p}^{X_{u\gets a}}_{X_{u\gets a},Y_{u\gets b}}\text{ for all }a\in Q_u$$
    $$\hat{p}^{Y}_{X,Y} =\sum_{b\in Q_u}  \hat{p}^{Y_{u\gets a}}_{X_{u\gets b},Y_{u\gets a}}\text{ for all }a\in Q_u$$ \label{jpvlp-3}
    \item For every node $(X, Y) \in \MSC{T}\setminus  \+{L}(\MSC{T})$, letting $u = \nextvar{h(X,Y)}$, for all $a\in Q_u$,
    $$\sum_{\substack{b\in Q_u\\ b\neq a}}  \hat{p}^{X_{u\gets a}}_{X_{u\gets a},Y_{u\gets b}}\leq 1-q\theta$$
    $$\sum_{\substack{b\in Q_u\\ b\neq a}}  \hat{p}^{Y_{u\gets a}}_{X_{u\gets b},Y_{u\gets a}}\leq 1-q\theta,$$
    where $\theta,\eta$ is defined as in \eqref{eq:parameter-theta}. \label{jpvlp-4}
\end{enumerate}
\end{definition}

It can be verified the following lemma holds by by taking $$\hat{p}^{X}_{X,Y}=\frac{\mu_{\textsf{cp}}(X,Y)}{\mu[X\mid X_0]}, \hat{p}^{Y}_{X,Y}=\frac{\mu_{\textsf{cp}}(X,Y)}{\mu[Y\mid Y_0]}$$ 
for each $(X,Y)\in \MSC{T}$ and verifying all items in \Cref{jpvlp}. Particularly, \Cref{jpvlp-4} of \Cref{jpvlp} can be shown using a similar argument as in \Cref{localuniformitycor}.
\begin{lemma}\label{lpcor}
The
LP defined in \Cref{jpvlp} is feasible for $r_{-}=r_{+}=\frac{\abs{\+{S}_{X_0}}}{\abs{\+{S}_{Y_0}}}$.
\end{lemma}

The following lemma holds by \Cref{shortpath2} and standard guarantees on the running time of linear programming.
\begin{lemma}\label{lpeff}
For every $r_{-}, r_{+}, \eta$ which can be represented in $\poly(n, q)$ bits, the feasibility of the
LP defined in \Cref{jpvlp} can be checked in time $\poly(n, q^{k\Delta L})$.
\end{lemma}

One crucial thing is that the feasibility of the above LP (for appropriately chosen $\alpha$ and truncation condition $f(\cdot)$) implies that $r_{-}$ (respectively $r_{+}$) is an approximate
lower (respectively upper) bound for $\abs{\+{S}_{X_0}} /\abs{\+{S}_{Y_0}}$. Given this, one will be able to use binary search to approximate  $\abs{\+{S}_{X_0}} /\abs{\+{S}_{Y_0}}$.

For improving the analysis in~\cite{Vishesh21towards}, we choose the truncation condition $f(\cdot)$ in \Cref{def:truncate-refined} for \Cref{jpvtree-2} of \Cref{jpvtree}. It suffices to show the following improved lemma.
\begin{lemma}[Improved version of \textbf{Lemma 5.1} in~\cite{Vishesh21towards}]\label{jpvimprovedlemma}
Recall that $X_0=\sigma_{v\gets a}, Y_0=\sigma_{v\gets b}$. If $16\mathrm{e}p\Delta^3\leq \pprime$, $\eta\leq (2k)^{-1}$, $1-q\theta\leq (8\mathrm{e}k\Delta)^{-1}$ and $L>1$,
$$\frac{1}{\abs{\+{S}_{X_0}}}\sum\limits_{\tau\in \+{S}_{X_0}}\sum\limits_{(X,Y)\in \+{L}_b(\MSC{T}):X\rightarrow \tau}\hat{p}^{X}_{X,Y}\leq \digamma(\sigma)$$
$$\frac{1}{\abs{\+{S}_{Y_0}}}\sum\limits_{\tau\in \+{S}_{Y_0}}\sum\limits_{(X,Y)\in \+{L}_b(\MSC{T}):Y\rightarrow \tau}\hat{p}^{Y}_{X,Y}\leq \digamma(\sigma),$$
where $x\rightarrow y$ means $x$ extends $y$, and $\digamma(\cdot)$ is the upper bound function in \Cref{def:potential-refined}.
\end{lemma}
Given \Cref{jpvimprovedlemma}, note that by \Cref{jpvlp-3} of \Cref{jpvlp}, we have
\begin{equation}\label{jpv-sum1}
    \sum\limits_{(X,Y)\in    \+{L}(\MSC{T}): X\rightarrow \tau} \hat{p}^{X}_{X,Y}=1\text{ for all }\tau\in \+{S}_{X_0}
\end{equation} 
Therefore, we have 
\begin{align*}
\abs{\+{S}_{X_0}}=&\sum\limits_{\tau\in \+{S}_{X_0}}\sum\limits_{(X,Y)\in \+{L}(\MSC{T}):X\rightarrow \tau}\hat{p}^{X}_{X,Y}\\
=&\sum\limits_{\tau\in \+{S}_{X_0}}\sum\limits_{(X,Y)\in \+{L}_g(\MSC{T}):X\rightarrow \tau}\hat{p}^{X}_{X,Y}+\sum\limits_{\tau\in \+{S}_{X_0}}\sum\limits_{(X,Y)\in \+{L}_b(\MSC{T}):X\rightarrow \tau}\hat{p}^{X}_{X,Y}\\
=&\sum\limits_{(X,Y)\in \+{L}_g(\MSC{T})}\left(\hat{p}^{X}_{X,Y}\cdot \abs{\+S_{X}}\right)\pm\digamma(\sigma)\cdot \abs{\+{S}_{X_0}},
\end{align*}
where the first equality is by \eqref{jpv-sum1} and the last equality is by interchanging sums and \Cref{jpvimprovedlemma}.
A similar estimate also holds for $\abs{\+{S}_{Y_0} }$.
Thus, we have
\begin{align*}
 \frac{\abs{\+S_{X_0}}\cdot (1\pm \digamma(\sigma))}{\abs{\+S_{Y_0}}\cdot (1\pm \digamma(\sigma))}&=\frac{\sum\limits_{(X,Y)\in \+{L}_g(\MSC{T})}\left(\hat{p}^{X}_{X,Y}\cdot \abs{\+S_{X}}\right)}{\sum\limits_{(X,Y)\in \+{L}_g(\MSC{T})}\left(\hat{p}^{Y}_{X,Y}\cdot \abs{\+S_{Y}}\right)}\\
(\text{By \Cref{jpvlp-2} of \Cref{jpvlp}})\quad &\in \left[\frac{r_{-}\cdot \sum\limits_{(X,Y)\in \+{L}_g(\MSC{T})}\left(\hat{p}^{Y}_{X,Y}\cdot \abs{\+S_{Y}}\right)}{ \sum\limits_{(X,Y)\in \+{L}_g(\MSC{T})}\left(\hat{p}^{Y}_{X,Y}\cdot \abs{\+S_{Y}}\right)},\frac{r_{+}\cdot \sum\limits_{(X,Y)\in \+{L}_g(\MSC{T})}\left(\hat{p}^{Y}_{X,Y}\cdot \abs{\+S_{Y}}\right)}{ \sum\limits_{(X,Y)\in \+{L}_g(\MSC{T})}\left(\hat{p}^{Y}_{X,Y}\cdot \abs{\+S_{Y}}\right)}\right]\\
 &\in [r_{-},r_{+}].
\end{align*}

With this guarantee, one may approximate $\abs{\+{S}_{X_0}} /\abs{\+{S}_{Y_0}}$ and therefore the estimate marginal distribution $\mu^{\sigma}_v$ with total variation distance between the true marginal distribution bounded above by some linear function with $\digamma(\sigma)$. Note that \Cref{lem-cvxl-in-cbad2-cor} also holds by the same reasoning for the replaced marginal approximator. Moreover, as we only replaced the subroutine for approximating marginals, both \Cref{lem-exeef2} and \Cref{inithcor2} still hold. Then combining \Cref{lpeff} and going through the same proof as Theorem \ref{thm:main-counting-refined}, one can improve the analysis for the algorithm presented in~\cite{Vishesh21towards} to work in the regime of $p\Delta^5\lesssim 1$.

\begin{proof}[Proof of \Cref{jpvimprovedlemma}]
Consider the following process of generating a random root-to-leaf paths of $\MSC{T}$. At a non-leaf node
$(X, Y) \in  \MSC{T}\setminus \+{L}(\MSC{T})$ , sample a value $a$ for $u = \nextvar{h(X,Y)}$ according to $\mu^{X}_v$ and set $X'\leftarrow X_{u\gets a}$. Then, choose a random element $b\in Q_u$ and go to the node $(X', Y_{u\gets b}) \in \MSC{T}$, where
the probability of choosing each $b \in Q_u$ is 

\begin{equation}\label{eq:lptransprop}
    p(X,Y,X',Y_{u(b)})=\frac{\hat{p}^{X'}_{X',Y_{u\gets b}}}{\hat{p}^{X}_{X,Y}}, 
\end{equation}    
Note that by \Cref{jpvlp-3} of \Cref{jpvlp}, one can verify that $p(X, Y, X', Y_{u\gets (\cdot)})$ is a  probability distribution. Let $(X^*,Y^* )$ denote the
random leaf of $\MSC{T}$ returned by this process and let $\hat{\mu}$ denote the probability distribution on $\+{L}(\MSC{T})$ induced by this process.

Let $(X_{\ell} , Y_{\ell} )\in  \+{L}(\MSC{T})$ and denote the corresponding root-to-leaf path by $(X_0, Y_0), \dots , (X_{\ell} , Y_{\ell} )$. 
Then,
$$\hat{\mu}[(X^*,Y^*)=(X_{\ell} , Y_{\ell} )]=\prod\limits_{t=1}^{\ell}\mu(X_t\mid X_{t-1})\times \prod\limits_{t=1}^{\ell}p(X_{t-1},Y_{t-1},X_t,Y_t)=\frac{\abs{\+{S}_{X_{\ell}}}}{\abs{\+{S}_{X_{0}}}}\cdot \frac{\hat{p}^{X_{\ell}}_{X_{\ell},Y_{\ell}}}{\hat{p}^{X_{0}}_{X_{0},Y_{0}}}=\frac{\abs{\+{S}_{X_{\ell}}}}{\abs{\+{S}_{X_{0}}}}\cdot\hat{p}^{X_{\ell}}_{X_{\ell},Y_{\ell}}, $$
where the first equality is by chain rule, the second one by \eqref{eq:lptransprop} and the last one by \Cref{jpvlp-3} of \Cref{jpvlp}.
Therefore,
\begin{align*}\frac{1}{\abs{\+{S}_{X_0}}}\sum\limits_{\tau\in \+{S}_{X_0}}\sum\limits_{(X,Y)\in \+{L}_b(\MSC{T}):X\rightarrow \tau}\hat{p}^{X}_{X,Y}&=\sum\limits_{\tau\in \+{S}_{X_0}}\sum\limits_{(X,Y)\in \+{L}_b(\MSC{T}):X\rightarrow \tau}\frac{\hat{\mu}[(X^*,Y^*)=(X,Y)]}{\abs{\+{S}_{X}}} \\
&=\sum\limits_{(X,Y)\in \+{L}_b(\MSC{T})}\hat{\mu}[(X^*,Y^*)=(X,Y)]
=\hat{\mu}[(X^*,Y^*)\in \+{L}_b(\MSC{T})].
\end{align*}

For each $0\leq i\leq \ell$ we let $Z_{i}=h(X_i,Y_i)$ and let $Z^*=h(X^*,Y^*)$. Note that by \Cref{jpvtree-2} of \Cref{jpvtree}, the stopping rule of the above process only depends on $h(X,Y)$, and one can then verify that $\hat{\mu}[(X^*,Y^*)\in \+{L}_b(\MSC{T})]=\hat{\mu}[f(Z^*)=\True]$. 

It remains to bound $\hat{\mu}[f(Z^*)=\True]$. Note that this process of generating a root-to-leaf path also leads to a process that generates a sequence of constraints $Z_0,\dots,Z_{\ell}$ that satisfy the following two properties:
\begin{enumerate}
    \item 
    if $\nextvar{Z_i}=\perp$ or $f(Z_i)=\True$, the sequence stops at $Z_i$;\label{jpvrp-1}
    \item
    otherwise $u=\nextvar{Z_i}\in V$, 
    the partial assignment $Z_{i+1}\in\qs$ is generated from $Z_{i}$ by randomly giving $u$ a value $x\in \qus{u}$, \label{jpvrp-2}
    such that
    \begin{enumerate}
        \item $
        \Pr{Z_{i+1}=\Mod{(Z_i)}{u}{\star}}\leq 1-q\theta. \label{jpvrp-2a}
        $
        \item $
        \forall x\in Q_u,
      \Pr{Z_{i+1}=\Mod{(Z_i)}{u}{x}}\leq \frac{1}{q_u}(1+\eta)
        $ \label{jpvrp-2b}
    \end{enumerate}
\end{enumerate}
We then show these properties. \Cref{jpvrp-1} is from the stopping rule of the process of generating a root-to-leaf path. \Cref{jpvrp-2a} is from combining \eqref{eq:lptransprop} and \Cref{jpvlp-4} of \Cref{jpvlp}. \Cref{jpvrp-2b} is from combining   the rule of the process of generating a root-to-leaf path and that
$$  \forall x\in Q_u,
      \Pr{Z_{i+1}=\Mod{(Z_i)}{u}{x}} \leq \Pr{X_{i+1}=\Mod{(X_i)}{u}{x}} \leq \mu^{X_i}_{v}(x)\leq \frac{1}{q_u}(1+\eta). $$
Note that the process above that generates $Z_0,\dots,Z_{\ell}$ is pretty much similar to the process $\pth$ defined in \Cref{pathdef}. Moreover, it can be verified that the proofs in \Cref{WTprob}, \Cref{WTprobcor} and eventually in \Cref{tvdbound2} still can apply for this process, and going through these proofs for this process leads to the desired result. 
\end{proof}

\appendix

\bibliographystyle{alpha}
\bibliography{references} 

\clearpage

\section{Generalized \texorpdfstring{$\{2,3\}$}{2,3}-tree as witnesses for useful properties}\label{sec:WT-proofs}
In this section, we prove several technical lemmas (\Cref{WTsize}, \Cref{shortpath2}, \Cref{lem-cvxl-in-cbad2}   and \Cref{lem-bigWT-rs}). \Cref{shortpath2} states that for some partial assignment $\sigma\in \qs$, the length of  $\pth(\sigma)=(\sigma_0,\dots,\sigma_{\ell})$ is bounded.  \Cref{lem-cvxl-in-cbad2} relates the size of $\+C^{\sigma_{\ell}}_v$ with the sizes of $\vst{\sigma_{\ell}}$ and $\csfrozen{\sigma_{\ell}}$. \Cref{WTsize} and \Cref{lem-bigWT-rs} state that for certain properties such as $\vst{\sigma}$ or $\csfrozen{\sigma}$,
when the class of variables/constraints with that property becomes too large, 
a generalized $\{2,3\}$-tree with certain properties inevitably appears within the class.

To aid our proof, we introduce the definition of $\gvc$, a graph with a vertex set over all variables and constraints of the CSP formula.
\begin{definition}[Graph of variables and constraints]\label{def:graphvc}
Let $\Phi =(V,\+Q,\+C)$ be the CSP formula.
Define  $\gvc=(V\cup \+{C},E)$ as the graph where vertices are $V\cup \+{C}$ and there is an edge between two vertices $u,v$ if and only if one of the following holds:
\begin{enumerate}
    \item $u,v\in V$ and there exists some $c\in \+{C}$ such that $u,v\in \var{c}$.\label{graphvc-1}
    \item $u,v\in \+{C}$ and $\text{dist}_{\Lin{H_\Phi}}(u,v)=1\text{ or }2$.\label{graphvc-2}
    \item $u\in V,v\in \+{C}$ and  there exists some $c\in \+{C}$ such that $u\in \var{c}\land \text{dist}_{\Lin{H_\Phi}}(c,v)=1$.\label{graphvc-3}
\end{enumerate}
Furthermore, for any $S\subseteq V\cup \+C$, we let $\gvc(S)$ denote the subgraph of $\gvc$ induced by $S$.
\end{definition} 

We state the following lemma, which is immediate by ~\cite[Lemma 6.32]{he2022sampling}.

\begin{lemma}\label{cor-mono}
Let $\sigma\in \qs$ and $\pth(\sigma) = (\sigma_0,\sigma_1,\cdots,\sigma_{\ell})$.
For every $0\leq i\leq j \leq \ell$,
it holds that 
\[
\vst{\sigma_i}\subseteq \vst{\sigma_j},
\]
and
\[
\+{C}^{\sigma_i}_{\+{P}}\subseteq \+{C}^{\sigma_j}_{\+{P}},
\]
where $\+{P}$ can be any property $\+{P}\in\{\, \mathsf{frozen},\,\, \star\text{-}\mathsf{con},\,\, \star\text{-}\mathsf{frozen}\,\}$.
\end{lemma}

\subsection{Proof of \Cref{WTsize}} \label{sec:proof of WTsize}
We first need the two following lemmas.
\begin{lemma}\label{lem-connect2}
Assume the condition of \Cref{bigWT}.
Then $\gvc\left(\csfrozen{\sigma_i}\cup \vst{\sigma_i}\right)$ is connected for each $0 \leq i \leq \ell$.

\end{lemma}
\begin{proof}
We prove this lemma by induction on $i$. 
For simplicity, we say a variable or constraint $c$ is connected to a subset $S\subseteq V\cup \+{C}$ in  $\gvc$ if $c$ is connected to some $c'\in S$. The base case is when $i=0$.
By the condition of the lemma, $v$ is the only variable satisfying
$\sigma(v) = \star$.
Combining with $\sigma_0 = \sigma$, 
we have $v$ is the only variable satisfying $\sigma_0(v) = \star$. Therefore,
$\vst{\sigma_0} = \set{v}$.
In addition, 
we have the following claim: each $c\in \csfrozen{\sigma_0}$ is connected to $v$ in $\gvc(\csfrozen{\sigma_0}\cup \vst{\sigma_0})$.
Combining with the claim,
we have $\gvc(\csfrozen{\sigma_0}\cup \vst{\sigma_0})$ is connected.

Now we prove the claim, which completes the proof of the base case.
By $c\in \csfrozen{\sigma_0}$,
we have $c\in \ccon{\sigma_0}\cap \cfrozen{\sigma_0}$.
By $c\in \ccon{\sigma_0}$ and \Cref{definition:boundary-variables},
we have $\vcon{\sigma_0}\cap \var{c}\neq \emptyset$.
Combining with $v$ is the only variable satisfying $\sigma_0(v) = \star$ and the definition of $\vcon{\sigma_0}$,
we have there exists a connected path $c^{\sigma_{0}}_1,c^{\sigma_{0}}_2,\cdots,c^{\sigma_{0}}_{t}=c^{\sigma_{0}}\in \+C^{\sigma_{0}}$ such that $\sigma_{0}(v)=\star$, $v \in \var{c^{\sigma_0}_1}$
and $\var{c^{\sigma_{0}}_j}\subseteq V^{\sigma_0}\cap \vfix{\sigma_{0}}$ for each $j <t$.
If $c = c_1$,  then $v \in \var{c}$ and the claim is immediate by the definition of $\gvc$.
In the following, we assume $c \neq c_1$. Let $w_j\in \left(\var{c^{\sigma_0}_{j}}\cap\var{c^{\sigma_0}_{j+1}}\right)$ for each $j<t$.
    Then $w_j\not \in \Lambda(\sigma_{0})$.
By $w_j\in \var{c^{\sigma_0}_j}$ and
$\var{c^{\sigma_0}_j}\subseteq \vfix{\sigma_{0}}$,
we have $w_j\in \vfix{\sigma_{0}}$.
Combining with $w_j\not \in \Lambda(\sigma_{0})$,
we have either $\sigma_{0}(w_j) = \star$, where we set $\widehat{c}_j=w_j$; or 
$w_j \in \var{\widehat{c}_j}$ for some $\widehat{c}_j\in \cfrozen{\sigma_{0}}$. Note that $\widehat{c}_j$ can be either a variable or a constraint. In the former case, we have $\widehat{c}_j\in \vst{\sigma_0}$. In the latter case,
 By $w_j$ is connected to $v$ in $\hfix{\sigma_0}$ through the path $c^{\sigma_0}_1,c^{\sigma_0}_2,\cdots,c^{\sigma_0}_{j}$, 
we have 
$w_j\in \vcon{\sigma_0}$.
Thus, we have $\widehat{c}_j \in \ccon{\sigma_{0}}$ by \Cref{definition:boundary-variables}.
Combining with $\widehat{c}_j\in \cfrozen{\sigma_{0}}$,
we have $\widehat{c}_j\in \csfrozen{\sigma_{0}}$.
In summary, we always have $\widehat{c}_j\in \vst{\sigma_0}\cup \csfrozen{\sigma_{0}}$.
\sloppy Moreover, for each $j < t-1$, 
if $\widehat{c}_j\in \+C$, we have $w_j \in \var{c^{\sigma_0}_{j+1}}\cap \var{\widehat{c}^{\sigma_0}_j}$, otherwise we have $\widehat{c}_j=w_j$.
Thus by \Cref{def:graphvc}, it can be verified that $\widehat{c}_j$ and $\widehat{c}_{j+1}$ are adjacent in $\gvc$.
In addition, if $\widehat{c}_1\in \+C$, we have
$w_1 \in \var{c_1}\cap \var{\widehat{c}_1}$ , otherwise we have $\var{\widehat{c}_1}=w_1\in \var{c_1}$, hence $c_1$ and $\widehat{c}_{1}$ are adjacent in $\gvc$. Similarly, we have  $\widehat{c}_{t-1}$ and $c_{t}$ are adjacent in $\gvc$
Thus, we have 
$v,c_1,\widehat{c}_1,\widehat{c}_2,\cdots,\widehat{c}_{t-1},c_{t}=c$ is a connected path in $\gvc$.
Combining with $v \in \vst{\sigma_{0}}$ and
$\widehat{c}_j \in \csfrozen{\sigma_{0}}$ for each $j< t$, the claim is immediate.

For the induction step, we prove this lemma for each $i>0$.
We claim that each $v\in \vst{\sigma_i}$ is connected to $\vst{\sigma_{i-1}}$ in $\gvc(\csfrozen{\sigma_i}\cup \vst{\sigma_i})$. In addition,
we can prove each  $c\in \csfrozen{\sigma_{i}}$ is connected to  $\vst{\sigma_{i-1}}$ in $\gvc(\csfrozen{\sigma_i}\cup \vst{\sigma_i})$ by a similar argument to the base case.
Moreover, by the induction hypothesis we have $\gvc(\csfrozen{\sigma{i-1}}\cup \vst{\sigma_{i-1}})$ is connected.
Combining with $\csfrozen{\sigma_{i-1}} \subseteq \csfrozen{\sigma_{i}}$ and $\vst{\sigma_{i-1}} \subseteq \vst{\sigma_{i}}$ by \Cref{cor-mono},
 we have $\gvc(\csfrozen{\sigma_{i}}\cup \vst{\sigma_{i}})$ is connected.

Now we prove the claim that each $v\in \vst{\sigma_i}$ is connected to $\vst{\sigma_{i-1}}$ in $\gvc$, which completes the proof of the lemma.
If $v\in \vst{\sigma_{i-1}}$, the claim is immediate by $\vst{\sigma_{i-1}}\subseteq \vst{\sigma_{i}}$.
In the following, we assume $v\in \vst{\sigma_i}\setminus \vst{\sigma_{i-1}}$, where by \Cref{pathdef}  we have $v=\nextvar{\sigma_{i-1}}$. 
By the definition of $\nextvar{\cdot}$, we have $v\in \vinf{\sigma_{i-1}}$ and then $v \in \var{\widehat{c}}$ for some constraint $\widehat{c} \in \ccon{\sigma_{i-1}}$.
In addition, by $\widehat{c}\in \ccon{\sigma_{i-1}}$ one can verify that
there exists a variable $w\neq v$ and a connected path $c^{\sigma_{i-1}}_1,c^{\sigma_{i-1}}_2,\cdots,c^{\sigma_{i-1}}_{t} = \widehat{c}^{\sigma_{i-1}}\in\+C^{\sigma_{i-1}}$ such that $\sigma_{i-1}(w)=\star$, $w \in \var{c^{\sigma_{i-1}}_1}$ and $\var{c^{\sigma_{i-1}}_j}\subseteq V^{\sigma_{i-1}}\cap \vfix{\sigma_{i-1}}$ for each $j<t$.
Then there are two possibilities for $\widehat{c}$.
\begin{itemize}
\item If $\widehat{c} = c_1$, we have $v,w\in \var{c_1}$.
Therefore,
$v$ is connected to $w$ in $\gvc(\csfrozen{\sigma_{i-1}}\cup \vst{\sigma_{i-1}}\cup \{v\})$. Also by $\sigma_{i-1}(w)=\star$ we have $w\in \vst{\sigma_{i-1}}$. In addition, by \Cref{cor-mono} we have 
$$\csfrozen{\sigma_{i-1}}\cup \vst{\sigma_{i-1}}\cup \{v\}\subseteq \csfrozen{\sigma_{i-1}}\cup \vst{\sigma_{i}}\subseteq \csfrozen{\sigma_{i}}\cup \vst{\sigma_{i}}.$$
Thus the claim is immediate.
\item Otherwise, $\widehat{c} \neq c_1$.
Similarly to the base case, 
one can find a connected path $c_1,\widehat{c}_1,\widehat{c}_2,\cdots,\widehat{c}_{t-1},\\c_{t}=\widehat{c}$ in $\gvc$,
where $w \in \var{c_1}$,
$\widehat{c}_j \in \vst{\sigma_{i-1}}\cup \csfrozen{\sigma_{i-1}}$ for each $j< t$, 
and there exists $w_{t-1} \in \var{c_{t}}\cap \var{\widehat{c}_{t-1}}$.
Recall that $v\in \var{\widehat{c}}$ and $w\in \var{c_1}$
Thus, $w,c_1,\widehat{c}_1,\widehat{c}_2,\cdots,\widehat{c}_{t-1},v$ is also a connected path in $\gvc$.
Combining with $w\in \vst{\sigma_{i-1}}$,
$\widehat{c}_j \in \vst{\sigma_{i-1}}\cup \csfrozen{\sigma_{i-1}}$ for each $j< t$,
we have $v$ is connected to $\vst{\sigma_{i-1}}$ in $\gvc(\csfrozen{\sigma_i}\cup \vst{\sigma_i})$.
Thus the claim is immediate.
\end{itemize}
\end{proof}

\begin{lemma}
\label{bigWT}

Let $\sigma\in\qs$ 
be a partial assignment satisfying that exactly one variable $v\in V$ has $\sigma(v) = \star$.
Let $\pth(\sigma) =
(\sigma_0,\sigma_1,\ldots,\sigma_{\ell})$.
Then there always exists a generalized $\{2,3\}$-tree $T=U\circ E$ in $H_{\Phi}$ with some auxiliary tree rooted at $v$ such that 
\[
U=\vst{\sigma_{\ell}},
\quad
E\subseteq \csfrozen{\sigma_{\ell}}, 
\quad \text{ and }\quad
\Delta\cdot \abs{E}\geq \abs{\csfrozen{\sigma_{\ell}}}
\]
\end{lemma}

\begin{proof}
 We construct $T$ along with one of its auxiliary tree $T^*$ by greedily starting from a single root $v$. We maintain a set $B$ of "valid vertices", initially set as $B\gets \csfrozen{\sigma_{\ell}}\cup \vst{\sigma_{\ell}}\setminus \{v\}$. Each time we choose a vertex $u$ in $B$ that is nearest to $T$ in $\gvc$, i.e., $u=\mathop{\arg\min}\limits_{w\in B}\min\limits_{x\in T}\text{dist}_{\gvc}(w,x)$, then let $w$ be the vertex nearest to $u$ in $T$. We add $w$ along with the arc $(w,u)$ in $T^*$, then update $B$ as follows:
\begin{enumerate}[label=(\alph*)]
    \item If $u\in V$, then we update $B\gets B\setminus \{u\}$\;
    \item If $u\in \+{C}$, we update $B\gets B\setminus \Gamma^+(u)$, where $\Gamma^+(u)=\{c\in \+{C}\mid \var{u}\cap \var{c}\neq \emptyset\}$.\label{removeb}
\end{enumerate}
If $B=\emptyset$, the process stops.
We claim that when the process stops, we have $T=U\circ E$ is a generalized $\{2,3\}$-tree in $H_{\Phi}$ satisfying $
U=\vst{\sigma_{\ell}},
E\subseteq \cfrozen{\sigma_{\ell}}\text{ and }
\Delta\cdot \abs{E}\geq \cfrozen{\sigma_{\ell}}. $

We first show that $T$ is a generalized $\{2,3\}$-tree in $H_{\Phi}$. \Cref{WT-1} of \Cref{WTdef} is immediate by \Cref{removeb} of the process. It then suffices to show $T^*$ is a valid auxiliary tree. 
For \Cref{WT-2} of \Cref{WTdef}, note that by \Cref{lem-connect2} we have $\gvc(\csfrozen{\sigma_{\ell}}\cup \vst{\sigma_{\ell}})$ is connected. Also from the process, we know each $u\in \csfrozen{\sigma_{\ell}}\cup \vst{\sigma_{\ell}}$ is either added into $T$ or removed in \Cref{removeb}. If $u$ is removed in \Cref{removeb}, then $u\in \+{C}$ and there exists $c\in \+{C}$ such that $\var{c}\cap \var{u}\neq \emptyset$ and $c$ is added into $T$.  For each $u\neq v\in T$, let $w$ be the only father of $u$ in $T^*$, then we have the following cases:
\begin{itemize}
    \item $\text{dist}_{\gvc}(u,w)=1$:  Then the arc $(u,w)$ must satisfy  \Cref{WT-2} of \Cref{WTdef} by comparing \Cref{def:graphvc} with \Cref{WT-2} of \Cref{WTdef}.
    \item Otherwise it must follow that $w\in \+{C}$ and $\text{dist}_{\gvc}(u,w')=1$ for some constraint $w'\in \+{C}$ removed during \Cref{removeb} such that $\var{w}\cap \var{w'}\neq \emptyset$. This is by our choice of $u$ and $w$ from the process and the fact that $\gvc(\csfrozen{\sigma_{\ell}}\cup \vst{\sigma_{\ell}})$ is connected. Then the arc $(u,w)$ must also satisfy  \Cref{WT-2} of \Cref{WTdef} by comparing \Cref{def:graphvc} with \Cref{WT-2} of \Cref{WTdef}.
\end{itemize}

This shows that  $T^*$ is a valid auxiliary tree rooted at $v$; therefore, $T$ is a generalized $\{2,3\}$-tree in $H_{\Phi}$ satisfying the condition.

The claims $U=\vst{\sigma_{\ell}}$ and 
$E\subseteq \cfrozen{\sigma_{\ell}}$ are trivial from the process. Note that in each step of the process at most $\Delta$ vertices in $\csfrozen{\sigma_{\ell}}$ or one vertex in $\vst{\sigma_{\ell}}$ are removed. Therefore we have $\Delta\cdot \abs{E}\geq \abs{\csfrozen{\sigma_{\ell}}}$. This completes the proof.
\end{proof}

Now we are ready to prove \Cref{WTsize}.
\begin{proof}[Proof of \Cref{WTsize}]

Note that by \Cref{def:truncate-refined}, $f(\sigma_{\ell})=\True$ says $\abs{\vst{\sigma_{\ell}}}+\Delta\cdot \abs{\csfrozen{\sigma_{\ell}}}\geq L\Delta$. We then analyze two cases separately: $\ell=0$ and $\ell\geq 1$.

 For the case when $\ell=0$, by $\sigma\in\qs$ 
is a partial assignment with exactly one variable $v\in V$ having $\sigma(v) = \star$, it suffices to take $U=\{v\}$ and repeatedly add available constraints and corresponding edges in $\csfrozen{\sigma_\ell}$ into the auxiliary tree $T^*$ to construct $T=U\circ E$, as in the proof of \Cref{bigWT}, until 
$\Delta\cdot \abs{E}+\abs{U}\geq L$.  
Then we have $L\leq \Delta\cdot \abs{E}+\abs{U}\leq L\Delta$  by combining  \Cref{bigWT}, the assumption $\abs{\vst{\sigma_{\ell}}}+\Delta\cdot \abs{\csfrozen{\sigma_{\ell}}}\geq L\Delta$ and  $L+\Delta\leq L\Delta$ from the assumption that $L>1$ and $\Delta\geq 2$.

Otherwise we have $\ell\geq 1$. By \Cref{rp-1} of \Cref{pathdef} we have $\abs{\vstar{\sigma_{\ell-1}}}+\Delta\cdot \abs{\csfrozen{\sigma_{\ell-1}}}<L\Delta$ and hence by \Cref{bigWT} there exists a generalized $\{2,3\}$-tree $T'=U'\circ E'$ in $H_{\Phi}$ satisfying $\abs{U'}+\Delta\cdot \abs{E'}< L\Delta$ with some auxiliary tree rooted at $v$ such that $\+{E}^{\sigma_{\ell-1}}_{T'}$ happens. Let $u=\nextvar{\sigma_{\ell-1}}$. We then construct $T=U\circ E$ from $T'$ by adding $u$ and the corresponding edge into $T^*$ if $\sigma_{\ell}(u)=\star$, and then repeatedly adding available constraints in $\csfrozen{\sigma_\ell}$ and corresponding edges into $T^*$ until $\Delta\cdot \abs{E}+\abs{U}\geq L$. 
Then we have $L\leq \Delta\cdot \abs{E}+\abs{U}\leq L\Delta$  by combining  \Cref{bigWT}, the assumption $\abs{\vst{\sigma_{\ell}}}+\Delta\cdot \abs{\csfrozen{\sigma_{\ell}}}\geq L\Delta$ and  $L+\Delta\leq L\Delta$ from the assumption that $L>1$ and $\Delta\geq 2$.
\end{proof}

\subsection{Proof of \Cref{shortpath2}}\label{sec:proof of shortpath2}
\begin{proof}[Proof of \Cref{shortpath2}]
Fix any $0\leq i\leq \ell$.  We claim that 
for each $0\leq  j <i$,
\begin{enumerate}
    \item either there exist some $c_j,c_j'$ such that 
$\nextvar{\sigma_j}\in \var{c_j}$, $c_j'\subseteq \csfrozen{\sigma_{i}}$, and $\var{c_j}\cap\var{c_j'}\neq \emptyset$;\label{shortpath2-item1}
    \item  or there exist some $c_j,u_j$ such that
$\nextvar{\sigma_j},u_j\in \var{c_j}$ and $u_j\in \vst{\sigma_i}$. \label{shortpath2-item2}
\end{enumerate}
Therefore, for each $0\leq j<i$, $\nextvar{\sigma_{j}}$ is in a constraint $c$ where either
$\var{c}\cap \var{c'}\neq \emptyset \text{ for some }c'\in \csfrozen{\sigma_{i}}$, or $u\in \var{c}$ for some $u\in \vstar{\sigma_i}$.
Combining with $\abs{\var{c}}\leq k$, we have for each $0\leq i\leq \ell$,
\begin{align*}
i &\leq k\cdot \abs{ \{c\in \+{C}: \var{c}\cap \var{c'}\neq \emptyset \text{ for some }c'\in \csfrozen{\sigma_{i}}\text{ or } u\in \var{c} \text{ for some }u\in \vst{\sigma_{i}} \}}\\
&\leq k\Delta\cdot (\abs{\csfrozen{\sigma_{i}}}+\abs{ \vst{\sigma_{i}}})\leq k\Delta\cdot (\Delta\cdot \abs{\csfrozen{\sigma_{i}}}+\abs{ \vst{\sigma_{i}}}).
\end{align*}

The case when $\ell=0$ is trivial. We then assume $\ell\geq 1$. By \Cref{rp-1} of \Cref{pathdef}, we then have two cases.
\begin{itemize}
    \item If $\Delta\cdot \abs{\csfrozen{\sigma_{\ell}}}+\abs{ \vst{\sigma_{\ell}}}<L\Delta$, we directly obtain $\ell\leq k\Delta\cdot (\Delta\cdot \abs{\csfrozen{\sigma_{\ell}}}+\abs{ \vst{\sigma_{\ell}}})< kL\Delta^2$.
    \item If $\Delta\cdot \abs{\csfrozen{\sigma_{\ell}}}+\abs{ \vst{\sigma_{\ell}}}\geq L\Delta$, we have $\Delta\cdot \abs{\csfrozen{\sigma_{\ell-1}}}+\abs{ \vst{\sigma_{\ell-1}}}< L\Delta$ by \Cref{rp-1} of \Cref{pathdef}, hence $\ell-1<kL\Delta^2$ and therefore $\ell\leq kL\Delta^2$.
\end{itemize}

Now we prove the claim.
Note that by $\pth(\sigma)=(\sigma_0,\dots,\sigma_{\ell})$, $0\leq i\leq \ell$ and \Cref{pathdef}, we have $\nextvar{\sigma_{j}}\neq \perp$ for each $0\leq j<i$.
Assume that $\nextvar{\sigma_{j}} = u_j$.
By \Cref{definition:boundary-variables},
we have $u_j \in \vinf{\sigma_{j}}\neq \emptyset$.
Combining with the definition of $\vinf{\sigma_{j}}$,
we have there exists some $c_j\in \+{C}^{\sigma_j}$, $w_j\in \vcon{\sigma_j}$ such that $u_j,w_j\in\vbl(c_j)$.
By $w_j\in \vcon{\sigma_{j}}$,
we have $w_j\in  V^{\sigma_{j}}\cap \vfix{\sigma_{j}}$.
By $w_j\in V^{\sigma_{j}}$,
we have $w_j\not\in \Lambda(\sigma_{j})$.
Combining with 
$w_j\in \vfix{\sigma_{j}}$,
we have either $\sigma_j(w_j) = \star$ or
$w_j\in \widehat{c}_j$ for some $\widehat{c}_j\in \cfrozen{\sigma_{j}}$.
If $\sigma_j(w_j) = \star$,
we have $w_j\in \vst{\sigma_j}\subseteq \vst{\sigma_i}$ and $c_j,w_j$ satisfies \Cref{shortpath2-item2}.
Otherwise, $w_j\in \var{\widehat{c}_j}$ for some $\widehat{c}_j\in \cfrozen{\sigma_{j}}$.
In addition, by $w_j\in \vcon{\sigma_{j}}$ and $w_j\in \var{\widehat{c}_j}$, 
we have $\widehat{c}_j\in \ccon{\sigma_{j}}$.
Combining with $\widehat{c}_j\in \cfrozen{\sigma_{j}}$,
we have 
$\widehat{c}_j\in \csfrozen{\sigma_{j}}$. By $w_j\in\vbl(c_j)$ and $w_j\in \var{\widehat{c}_j}$, we have $\var{c_j}\cap\var{\widehat{c}_j}\neq \emptyset$ and $c_j,\widehat{c}_j$ satisfies \Cref{shortpath2-item1}.
This justifies the claim.

\end{proof}

\subsection{Proof of \Cref{lem-cvxl-in-cbad2}}\label{sec:proof of lem-cvxl-in-cbad2}
It is sufficient to show that $\abs{\ccon{\sigma_{\ell}}}\leq \Delta \cdot \abs{\csfrozen{\sigma_{\ell}}}+\Delta\cdot \abs{\vst{\sigma_{\ell}}} $. We show this by proving that
for each $c\in \ccon{\sigma_{\ell}
}$, either there exists some $u\in \var{c}$ such that $u\in \vst{\sigma_{\ell}}$, or there exists some $c'\in \csfrozen{\sigma_{\ell}}$ such that $\var{c}\cap \var{c'}\neq \emptyset$.

For each $c\in \ccon{\sigma_{\ell}
}$,
by \Cref{definition:boundary-variables},
we have there exists some $u\in \vcon{\sigma_{\ell}}\cap \var{c^{\sigma_{\ell}}}$.
By $u\in \vcon{\sigma_{\ell}}$,
we have $u\in  V^{\sigma_{\ell}}\cap \vfix{\sigma_{\ell}}$.
By $u\in V^{\sigma_{\ell}}$,
we have $u\not\in \Lambda(\sigma_{\ell})$.
Combining with 
$u\in \vfix{\sigma_{\ell}}$,
we have either $\sigma_{\ell}(u) = \star$ or
$u\in c'$ for some $c'\in \cfrozen{\sigma_{\ell}}$.
If $\sigma_{\ell}(u) = \star$,
we have $u\in \vstar{\sigma_{\ell}}$.
Otherwise, $u\in c'$ for some $c'\in \cfrozen{\sigma_{\ell}}$.
In addition, we also have $c'\in \ccon{\sigma_{\ell}}$ by $u\in \vcon{\sigma_{\ell}}$ and $u\in \var{c'}$.
Combining with $c'\in \cfrozen{\sigma_{\ell}}$,
we have $c'\in \csfrozen{\sigma_{\ell}}$.
This completes the proof.

\subsection{Proof of  \Cref{lem-bigWT-rs}}\label{sec:proof of lem-bigWT-rs}

Let $X =X^n$ where $X^0,X^1,\cdots,X^n$ is
the partial assignment sequence in \Cref{def-pas-cmain}.
The following lemma is immediate by ~\cite[Lemma C.2]{he2022sampling}.

\begin{lemma}\label{lem-subset-rs}
$ V\setminus \Lambda(X) \subseteq \var{\cfrozen{X}}$.
\end{lemma}

Now we can prove \Cref{lem-bigWT-rs}.
\begin{proof}[Proof of \Cref{lem-bigWT-rs}]
Let $\{\Phi_i^X= (V_i^X,\+{C}_i^X)\}\mid 1\leq i\leq K\}$ be the decomposition of $\Phi^X$.
If $v\not \in V_i$ for each $i\in [k]$, we have 
$\+C^X_v = \emptyset$ and
the lemma is trivial.
In the following, we assume
\emph{w.l.o.g.} that 
$v\in V^X_i$
for some $i\in K$.
Then we have 
$\Phi^X_v = (V^X_i,\+C^X_i)$.
Let 
$$S \triangleq  \left\{c\in \cfrozen{X}\mid c^X \in \+C^X_{i}\right\}.$$

At first, we prove that there exists some $c_v\in S$ such that $v \in \var{c_v}$.
By $v\in V^X_i$,
we have $v \not\in \Lambda(X)$. 
Combining with \Cref{lem-subset-rs},
we have there exists some $c_v \in \cfrozen{X}$ such that $v \in \var{c_v}$.
In addition, by $v \in \var{c_v}$ and $c_v \in \cfrozen{X}$, we also have $c^X_v \in \+C^X_{i}$.
Combining with $c_v \in \cfrozen{X}$, we have $c_v \in S$.

Now we prove $\abs{\+C^X_{i}}\leq \Delta\abs{S}$.
For each $c^X\in \+C^X_i$,
we have there exists a connected path $c^X_1,c^X_2,\cdots,c^X_{t} = c^X\in \+C^X_i$ such that $v \in \var{c^X_1}$.
Let $v' \in \var{c^X}$.
We have $v' \not\in \Lambda(X)$.
Combining with \Cref{lem-subset-rs},
we have $v' \in \var{\widehat{c}}$ for some $\widehat{c}\in \cfrozen{X}$.
Then we have $\widehat{c}^{X} \in \+C^X_{i}$ because
there exists a connected path $c^X_1,c^X_2,\cdots,c^X_{t},\widehat{c}^{X}\in \+C^X_i$ where $v \in \var{c^X_1}$.
Combining $\widehat{c}\in \cfrozen{X}$ with $\widehat{c}^{X} \in \+C^X_{i}$,
we have $c\in S$.
In summary, for each each $c^X\in \+C^X_i$, there exists some $\widehat{c}\in S$ such that $\var{c^X}\cap \var{\widehat{c}^X}\neq \emptyset$.
Thus, we have  $\abs{\+C^X_{i}}\leq \Delta\abs{S}$.

In the next, we prove that $\gvc(S)$ is connected.
It is enough to prove that any two different constraints $c,\widehat{c}\in S$ are connected in $\gvc(S)$.
Given $c,\widehat{c}\in S$, 
we have $c^X,\widehat{c}^X$ are in $\+C^X_{i}$.
Therefore, 
we have there exists a connected path $c^X = c^X_1,c^X_2,\cdots,c^X_t = \widehat{c}^X\in \+C^X_{i}$.
If $t\leq 3$, obviously $c$ and $\widehat{c}$ are connected in $G^2(S)$.
In the following, we assume that $t>3$.
Let $w_j \in \left(\var{c^X_{j}}\cap \var{c^X_{j+1}}\right)$ for each $j< t$.
Then we have 
$w_j\not \in \Lambda(X)$.
Combining with \Cref{lem-subset-rs},
we have $w_j \in \var{\widehat{c}^X_j}$ for some $\widehat{c}_j\in \cfrozen{X}$.
Moreover, we also have $\widehat{c}^X\in \+C^X_{i}$,
because $\widehat{c}^X_j$ is connected to $c^X$ through $c^X_2,\cdots,c^X_j\in \+C^X_{i}$.
Thus, we have $\widehat{c}_j\in S$.
In addition, for each $\widehat{c}_{j},\widehat{c}_{j+1}$ where $j< t-1$,
we have $\widehat{c}_{j}$ and $\widehat{c}_{j+1}$  are connected in $G^2(\+C)$,
because $w_j\in \var{\widehat{c}_{j}}\cap \var{c_{j+1}}$
and $w_{j+1}\in \var{\widehat{c}_{j+1}}\cap \var{c_{j+1}}$.
Thus, the constraints $c = c_1,\widehat{c}_1,\widehat{c}_2,\cdots,\widehat{c}_{t-1},c_t = \widehat{c}$ forms a connected path in $\gvc$.
Combining with $\widehat{c}_j\in S$ for each $j\leq t-1$ and $c,\widehat{c}\in S$, we have the constraints $c,\widehat{c}$ are connected in $\gvc(S)$.

In summary, we have $c_v\in S\subseteq \cfrozen{X}$, $\Delta\abs{S}\geq \abs{\+C^X_{i}}$ and $\gvc(S)$ is connected. Combing with $v\in \var{c_v}$  we have $\gvc(S\cup \{v\})$ is also connected.
By going through the process in the proof of \Cref{bigWT}, 
we have there exists a subset of constraints and vertices $T\subseteq S\cup \{v\}$ such that $T=\set{v}\circ E$ is a generalized $\{2,3\}$-tree in $H_{\Phi}$ with some auxiliary tree rooted at $v$ and $$\abs{E}\geq \abs{S}/\Delta\geq \abs{\+C^X_{i}}/\Delta^2 = \abs{\+C^X_{v}}/\Delta^2.$$

In addition, if $\abs{\+{C}^{X}_v}\geq L\Delta^2$, then similar to the proof of \Cref{WTsize},we take $U=\{v\}$ and repeatedly add available constraints and corresponding edges in $S$ into the auxiliary tree $T^*$ to construct $T=U\circ E$ until $\Delta\cdot \abs{E}+\abs{U}\geq L$. Then we have $L\leq \Delta\cdot \abs{E}+\abs{U}\leq L\Delta$  by combining $S\geq \abs{\+C^{X}_{v}}/{\Delta}\geq L\Delta$ and  $L+\Delta\leq L\Delta$ from the assumption that $L>1$ and $\Delta\geq 2$, and the lemma is immediate.
\end{proof}

\end{document}